\documentclass[acmsmall]{acmart}

\usepackage{bm}
\usepackage{multirow}
\usepackage{booktabs}
\usepackage{multicol}
\usepackage{amsmath}
\usepackage{subfigure}
\usepackage{graphicx}

\DeclareMathOperator{\sgn}{sgn}
\newtheorem{Theorem}{Theorem}

\AtBeginDocument{%
  \providecommand\BibTeX{{%
    \normalfont B\kern-0.5em{\scshape i\kern-0.25em b}\kern-0.8em\TeX}}}

\setcopyright{acmcopyright}
\copyrightyear{2018}
\acmYear{2018}
\acmDOI{10.1145/1122445.1122456}

\acmJournal{JACM}
\acmVolume{37}
\acmNumber{4}
\acmArticle{111}
\acmMonth{8}

\begin{document}

\title{Supervised Deep Hashing for High-dimensional and Heterogeneous Case-based Reasoning}

\author{Qi Zhang}
\email{qi.zhang-13@student.uts.edu.au,zhangqi_cs@bit.edu.cn}
\orcid{0000-0002-1037-1361}
\affiliation{%
  \institution{University of Technology Sydney}
  \streetaddress{15 Broadway, Ultimo}
  \city{Sydney}
  \state{NSW}
  \country{Australia}
  \postcode{2007}
}
\affiliation{%
  \institution{Beijing Institute of Technology}
  \streetaddress{No 5 Yard, Zhong Guancun South Stree, Haidian District}
  \city{Beijing}
  \country{China}
  \postcode{100081}
}

\author{Liang Hu}
\email{rainmilk@gmail.com}
\affiliation{%
  \institution{Tongji University}
  \streetaddress{1239 Siping Road, Yangpu District}
  \city{Shanghai}
  \country{China}
  \postcode{200070}
}
\affiliation{%
  \institution{DeepBlue Academy of Sciences}
  \streetaddress{No. 369, Weining Road, Xinjing Town, Changning District}
  \city{Shanghai}
  \country{China}
  \postcode{200240}
}

\author{Chongyang Shi}
\email{cy_shi@bit.edu.cn}
\affiliation{%
  \institution{Beijing Institute of Technology}
  \streetaddress{No 5 Yard, Zhong Guancun South Stree, Haidian District}
  \city{Beijing}
  \country{China}
  \postcode{100081}
}


\author{Ke Liu}
\email{liuke0222@126.com}
\affiliation{
 \institution{Beijing Normal University}
 \city{Beijing}
 \country{China}
 \postcode{100875}
}

\author{Longbing Cao}
\authornote{Corresponding Author.}
\email{longbing.cao@uts.edu.au}
\affiliation{%
  \institution{University of Technology Sydney}
  \streetaddress{15 Broadway, Ultimo}
  \city{Sydney}
  \state{NSW}
  \country{Australia}
  \postcode{2007}
}

\renewcommand{\shortauthors}{Trovato and Tobin, et al.}

\begin{abstract}
Case-based Reasoning (CBR) on high-dimensional and heterogeneous data is a trending yet challenging and computationally expensive task in the real world. A promising approach is to obtain low-dimensional hash codes representing cases and perform a similarity retrieval of cases in a Hamming space. However, previous methods based on data-independent hashing rely on random projections or manual construction, inapplicable to address specific data issues (e.g., high-dimensionality and heterogeneity) due to their insensitivity to data characteristics. To address these issues, this work introduces a novel deep hashing network to learn similarity-preserving compact hash codes for efficient case retrieval and proposes a deep-hashing-enabled CBR model HeCBR. Specifically, we introduce position embedding to represent heterogeneous features and utilize a multilinear interaction layer to obtain case embeddings, which effectively filtrates zero-valued features to tackle high-dimensionality and sparsity and captures inter-feature couplings. Then, we feed the case embeddings into fully-connected layers, and subsequently a hash layer generates hash codes with a quantization regularizer to control the quantization loss during relaxation. To cater for incremental learning of CBR, we further propose an adaptive learning strategy to update the hash function. Extensive experiments on public datasets show HeCBR greatly reduces storage and significantly accelerates the case retrieval. HeCBR achieves desirable performance compared with the state-of-the-art CBR methods and performs significantly better than hashing-based CBR methods in classification.
\end{abstract}



\maketitle

\section{Introduction}
Case-based Reasoning (CBR) is an incremental learning methodology of analogy solution making, inspired by cognitive science that humans handle new problems by referring to past analogous experiences (cases)~\cite{ZhangSNC19}. Generally, a complete CBR system includes four key steps: case retrieval, case reuse, case revision, and case retention. Case retrieval assesses the similarity between a target case and past cases and obtains the most similar past cases. Case reuse suggests a solution for the target case according to the solutions of past cases. Case revision verifies the correctness of the suggested solution and revises the solution if necessary. Case retention maintains the solved cases into a case base for future problem-solving. Due to its good interpretability and practicability, CBR has been applied successfully to classification~\cite{ZhangSNC19}, diagnosis~\cite{BegumAFXF11}, decision support~\cite{GuLBZW12}, fault detection~\cite{YanWZZ14}, and various fields, e.g., finance~\cite{Chuang13,SartoriMG16}, industry~\cite{Chan05,KhosravaniNW19,LaoCHYTP12}, manufacturing~\cite{lim2015fast,KhosravaniN20}, and medicine~\cite{Montani11,LamySGBS19}.

Similarity measures play a decisive role in obtaining similar cases and thus largely affect the performance of CBR systems. Appropriate similarity measures should maximize the model fitness to data characteristics and explore the inherent data characteristics for handling the underlying problems~\cite{dst_Cao15}. However, enormous complex data, specifically high-dimensional and heterogeneous data, has penetrated every corner of our lives, bringing significant challenges to quantifying data characteristics and building accurate similarity measures. In addition, the increasing amount of cases retained during the incremental learning of CBR often leads to a huge knowledge base (case base) with enormous cases and bring a huge computation burden of similarity calculation and ranking in case retrieval. The aforementioned issues usually degrade the performance and efficiency of CBR systems. Intuitively, it is crucial to probe the intrinsic characteristics of high-dimensional and heterogeneous data and build efficient data-aware similarity measures, where approaches such as approximate nearest neighbor search (e.g., hashing techniques studied in this work) are promising to improve the retrieval efficiency.


\subsection{Challenges of High-dimensional and Heterogeneous Data}
Measuring the similarity of large amounts of high-dimensional data is critical for data mining and machine learning tasks and applications, e.g., information retrieval~\cite{ZhongR0GZ20}, clustering~\cite{DangDYH20}, and hashing~\cite{ShiMZZYSLM20}. The volume and complexity of such data usually entails prohibitively large storage and time consumption. Handcrafting an accurate similarity measure is challenging since it is usually the case that only partial unknown features are relevant to the task at hand~\cite{LiuBS15}. Therefore, many traditional similarity measures, e.g., Euclidean distance, cosine similarity, and Jaccard coefficient, perform poorly or even work out of action under high dimensionality. To address these issues, existing research usually projects high-dimensional data into a low-dimensional space by dimensionality reduction like manifold learning~\cite{Echihabi20} or principal component analysis~\cite{ZhuLZXYW17} and then learns an approximate similarity measure in the reduced space. However, high-dimensional data is often accompanied by the issue of sparsity where many useful features are rarely observed, and datapoints are scattered in multiple lower dimensional manifolds. It brings new challenges that these dimensionality reduction techniques may be inapplicable and parameters for the reduction easily lead to severe overfitting to the data.

In addition to the high-dimensionality issue, it is also crucial to consider the heterogeneity embodied in complex cases in similarity assessment. In heterogeneous data, attributes may follow different distributions and show different significance. Intuitively, this may lead to the inaccuracy of most handcrafted similarity measures, which adopt consistent attribute similarity or linear (e.g., average) aggregation functions~\cite{MuangprathubBP13}. To tackle the heterogeneity of case data in the real world, recent studies usually utilize different similarity measures on different types of attributes, e.g., using the Euclidean distance for numerical attributes and the Hamming distance for categorical attributes~\cite{RezvanHS13}, and then optimize the allocation of attribute weights and aggregation functions for similarity assessment~\cite{YanSG14,ZhangSSN16,ZhangSNC19}. Those methods optimize similarity measures by linear or non-linear aggregation functions in a data-independent manner, which cannot depict the complex attribute coupling relationships and heterogeneity among attributes~\cite{ipm_Cao15}. Accordingly, a more promising but challenging approach is to capture the heterogeneity of attributes~\cite{HuCCGXY16,DalleauCS20,ZhengLSZLW17} and learn data-aware similarity metrics~\cite{NguyenWHK12,WangSC13} to capture feature couplings~\cite{ChengMWC13,ijcai_PangCC16,untie}.

\subsection{Gaps of Efficient Case-based Reasoning}
The CBR problem-solving process mainly includes the successive execution of case retrieval and case reuse since case revision and retention are usually performed offline. Case retrieval is the most time-consuming step which usually consists of traversing all cases for similarity calculation and ranking the cases according to their similarity scores. With the increasing number of cases maintained in the case base, the retrieval efficiency and storage burden of CBR become increasingly critical and affect the applicability of CBR systems. To improve the efficiency, recent studies focus on reducing the search space by building indices to structurally organize the case base \cite{LiuC12,Sanchez-RuizO14} or partition and index cases by clustering~\cite{ZhangSNC19} and mapping~\cite{GuoHP14,ZhuHQMP15}. However, these methods usually rely on domain expertise and similarity measures to partition and index the case base and need additional time costs to update the case base and maintain its indexing structure. This may greatly degrade the performance and efficiency of CBR systems on high-dimensional and heterogeneous data where domain expertise can hardly penetrate.

Another promising way to accelerate similarity retrieval is approximate nearest neighbor search such as locality-sensitive hashing (LSH)~\cite{IndykM98,ZhaoXCLSZ09}. LSH maps similar datapoints (cases) into same `buckets' (encoded with low-dimensional binary codes) with high probability, which preserves the similarity relationships between datapoints and can be regarded as a way of dimensionality reduction on high-dimensional data~\cite{DatarIIM04}. Several studies introduce LSH into CBR systems to approximate the nearest neighbor search process and scale traditional CBR systems to large-scale data~\cite{JalaliL15,WoodbridgeMBS16,JalaliL18}. The studies show CBR systems equipped with hashing techniques can greatly improve retrieval efficiency and achieve desirable performance with expected loss in accuracy. Recently, a tremendous amount of research shows that data-dependent hashing (a.k.a., learn to hash) methods perform significantly better than data-independent hashing (e.g., LSH)~\cite{ShenSLS15,WangZSSS18,ZhuLCLZ20}. Data-dependent hashing learns similarity-preserving compact hash codes from data with/without (supervised/unsupervised) given similarity information by various flexible hash functions, e.g., PCA~\cite{WangKC12}, kernel functions~\cite{ShenSLS15}, and deep neural networks~\cite{WuDL0W19}. Naturally, data-dependent hashing facilitates to capture complex data characteristics and potentially improves the performance of hashing-based CBR systems, which, however, has not yet been introduced into CBR.


\subsection{Contributions}
To address the above challenges and gaps, we introduce supervised deep hashing in CBR systems and propose a hash function with an adaptive hashing network to build a Hashing-enabled CBR system (short for HeCBR). Specifically, each feature is represented by the multiplication of its `amplitude' (position embedding for the feature) and `frequency' (the feature value). Subsequently, a multilinear interaction layer is introduced to aggregate the feature embeddings to capture multiview feature couplings and obtain case embeddings. The above two embedding layers filter out zero-valued features and calculate the case embeddings in multilinear time complexity to efficiently handle the high-dimensionality and heterogeneity issues. Then, the case embeddings are fed into fully-connected layers and a hash layer to generate the binary hash codes. To learn the hashing network, we construct the similarity groundtruth from the solutions of cases (precisely two cases being similar if they have same/similar solutions, e.g., classification label, dissimilar otherwise). We then optimize the learning objectives of minimizing the loss between the similarity groundtruth and case similarity in the hash space and constrain the loss with a quantization regularizer. Considering the nature of incremental learning in CBR, we further propose a mechanism of incremental learning combining an adaptive learning objective and an update strategy to update the hash function and hash codes respectively for retaining solved cases. Accordingly, we construct a hashing-enabled CBR model to integrate the adaptive hashing network to index the case base and accelerate the similarity retrieval speed. The overview problem-solving process of HeCBR is shown in Fig. \ref{fig:ahcbc}. In summary, the contributions of this work mainly include:

\begin{figure}
    \centering
    \includegraphics[width=0.5\linewidth]{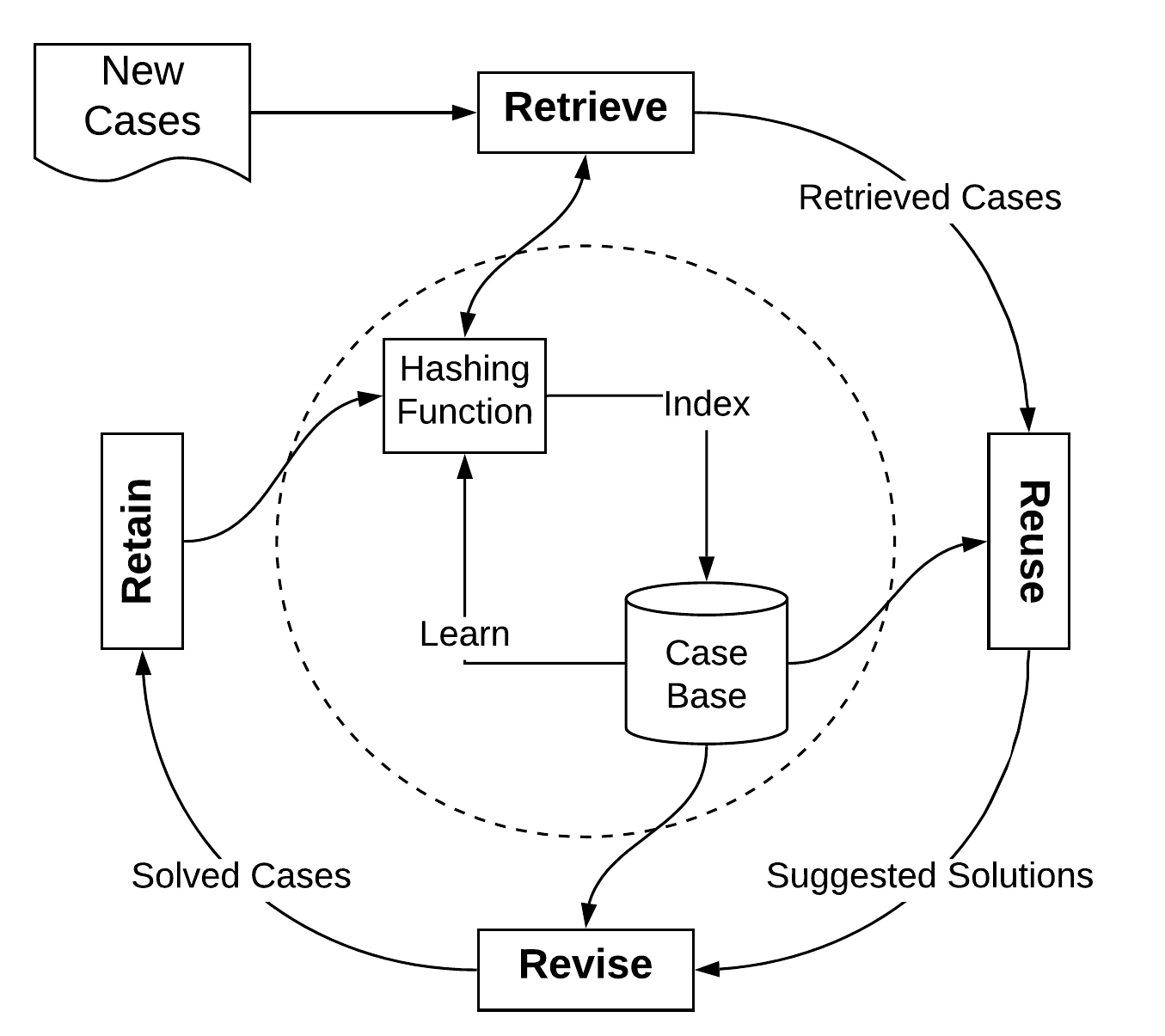}
    \caption{The problem-solving process of case-based reasoning with adaptive hashing.}
    \label{fig:ahcbc}
\end{figure}

\begin{itemize}
     \item We propose an adaptive hashing network to build a hashing-enabled CBR system. The proposed hashing network learns similarity-preserving compact hash codes of the cases and the corresponding hash function simultaneously. It is beneficial to map the high-dimensional cases into the low-dimensional hash space to reduce storage cost and approximate the nearest neighbor search in the hash space to improve similarity retrieval efficiency.
    \item To tackle the challenges and gaps of high-dimensional heterogeneous case data, we introduce position embedding to perform feature embedding and propose a multilinear interaction layer to represent complex cases. The network design filters out all zero-valued features and captures the feature couplings, assisting in efficiently handling high-dimensionality and heterogeneity while preserving similarity relationships.
	\item We further introduce an incremental learning mechanism to retain the solved cases. Specifically, we propose an adaptive learning objective to update the hash function and an updating strategy to efficiently update the hash codes of all cases.
\end{itemize}

We perform a collection of experiments on eight real-world datasets of different applications and make comparisons with other state-of-the-art hashing-enabled CBR models and several typical CBR models to investigate the effectiveness of our proposed HeCBR. All experimental results demonstrate that our HeCBR significantly outperforms the hashing-enabled methods in terms of classification tasks and achieves desirable performance compared with typical CBR methods. Theoretical and empirical analysis show that HeCBR can greatly reduce storage cost and significantly improve efficiency in relation to typical CBR models.

\section{Related Work}
Our work aims to improve the efficiency and scalability of CBR systems with high-dimensional and heterogeneous data by leveraging supervised deep hashing. We therefore first discuss the state-of-the-art methods aiming to improve the efficiency and scalability of CBR systems. Then, we provide a brief review of the current studies on handling the high-dimensionality and heterogeneity-related issues. In addition, we also present the recent progress on the data-dependent hashing. 

\subsection{Scalable Case-based Reasoning}
Due to the increase of case number in case base during retaining the newly solved cases, the efficiency and scalability of CBR systems have always been  crucial hindering the CBR systems from being exploited in applying large-scale and complex cases. Early studies pay attention to the strategies of case retention which expect to reduce redundant cases and maintain informative cases only~\cite{ShiuP04,MantarasMBLSCFMCFKAW05,ZhangSNC19}. Those methods leverage similarity-threshold filters to prune redundant cases or perform structure reduction to refine the structure of the case base, which alleviates the rapid growth of the case base to some extent. However, they highly rely on domain experience or expertise to calculate the similarity and set threshold and degrade the applicability of CBR systems on large-scale real life data.

A more promising approach to handle the efficiency and scalability issues is to build well-organized structures to index cases. For example, the typical CBR system D´ej´a Vu~\cite{SmythKC01} organizes software-design cases by hierarchically storing case description and the solutions in different layers and achieves desirable performance and efficiency. \citeauthor{LiuC12}~\cite{LiuC12} introduce an effective and efficient Z indexing approach to index cases and divide the case base into small sets in a tree. \citeauthor{Sanchez-RuizO14}~\cite{Sanchez-RuizO14} propose Least-Common Subsumer (LCS) trees to organize plan cases. Those methods leverage a hierarchical tree structure to index cases and narrow down the search space for the goal of efficiency improvement. In addition, many researchers adopt clustering techniques to accelerate case retrieval. A growing hierarchical self-organizing map (GHSMO)~\cite{GuoHP14,ZhuHQMP15} is introduced to categorize similar cases into same clusters and then index the clusters. \citeauthor{MuangprathubBP13} and \citeauthor{ZhangSNC19} further introduce a complete concept lattice for conceptual clustering to structurally organize and index cases. Like the structure-based CBR models, the clustering-based models improve the efficiency and accuracy of CBR systems, however all those models cost a large amount of time to construct and update the organizational structure and index and often require extra storage to maintain the structure and index, leading to low applicability in the real world.

Hashing techniques as a special indexing approach has also been applied into CBR systems. Hashing methods, e.g., LSH~\cite{IndykM98,ZhaoXCLSZ09} and E2LSH~\cite{DatarIIM04}, project (high-dimensional) cases into low-dimensional binary vectors (hash codes) and maintain the similarity information from the original space. The methods not only provide an efficient index of cases also enable approximate nearest neighbor search in the hash space and substantial data compression for the case base. Most previous studies introduce only LSH (data-independent hashing) as the underlying hash-based nearest neighbor search algorithm for large-scale cases. For example, \citeauthor{JalaliL15} presents a case study using Map Reduce and LSH to make the ensembles of adaptation of regression (EAR) in CBR feasible for large case bases and subsequently develops foundational scale-up methods using LSH for fast approximate nearest neighbor search of both cases and adaption rules for industrial scale prediction~\cite{JalaliL18}. \citeauthor{WoodbridgeMBS16} introduce LSH as an alternative to improve biomedical signal search results and accelerate search speed. Those methods prove the effectiveness of hashing methods in fast similarity search in CBR systems and also benefit performance improvement. However, data-independent hashing has its intuitive drawbacks that it often needs longer hash bits and cannot capture data characteristics. Recently, ~\citeauthor{JiangZHYM16} design a supervised hashing method based on linearly combined kernel functions associated with individual features from images for scalable histopathological image analysis in a CBR system~\cite{JiangZHYM16} to handle image data with efficiency improvement. Nevertheless, there exist quite few studies adopting  advanced data-dependent hashing techniques in CBR systems.

\subsection{High-dimensionality and Heterogeneity Issues}
In this section, we first briefly review the techniques for handling high-dimensionality and heterogeneity  respectively and then discuss the literature studying both  issues simultaneously.

The most common technique for handling the high-dimensionality  is dimensionality reduction~\cite{PandoveGR18,EspadotoMKHT21}, which roughly includes manifold learning~\cite{ZhuCHWZ18,TangFT21}, feature selection/extraction~\cite{zhuZWZZ18,TsamardinosBKPC19}, and the encoder-decoder framework~\cite{KusnerPH17,PangCCL18,KosiorekSTH19}. Those methods often  project high-dimensional data into low-dimensional representation and expect that the representation preserves certain patterns, e.g., neighbors, distances or clusters, or maintains the maximum mutual information with the original data. Regarding the incidental sparsity issues, some recent studies~\cite{WuHMY17,KrishnanLH18} further attempt to tackle the underfitting  incurred by  data sparsity. In addition, \citeauthor{0001C17} borrows the idea of factorization machine and integrates it with neural networks to handle sparsity~\cite{0001C17}. Inspired by them, we extend the idea of factorization machine and build a multilinear interaction layer to handle the high dimensionality and sparsity. Note that a vast literature focusing on high-dimensionality is not mentioned here, please refer to~\cite{PandoveGR18,EspadotoMKHT21} for more details.

Another common problem is heterogeneous data, with the underlying generation process changing across data sets or domains~\cite{ZhangHZGS17}. \citeauthor{WangSKC18} proposes a flexible information-based framework specializing the maximum entropy principle and the least effort principle to a principled multimodality information fusion formalism~\cite{WangSKC18}. \citeauthor{8878018} introduces transfer learning techniques to the heterogeneity between source and target domains and propose an evidence-based heterogeneous transfer classification model~\cite{8878018}. \citeauthor{ZhuCLYK18} proposes the HELIC model to capture both value-to-attribute and attribute-to-class hierarchical couplings to reveal the intrinsic heterogeneity in data. In addition, a growing amount of research pays attention to heterogeneous information network where both attributes (or datapoints) and their relation may be heterogeneous~\cite{0004JCGCA19,HongLYLFWQY20,untie}. Those methods model data heterogeneity and couplings and achieve significant performance improvement, proving the rationality and practicality of capturing heterogeneity.

Generally, real file high-dimensional data displays heterogeneity due to either heteroscedastic variance or other forms of non-location-scale covariate effect~\cite{wang2012quantile}, indicating the necessity of addressing the issues of heterogeneity and high dimensionality simultaneously. \citeauthor{HaoSLC17} combine the high-dimensional version of Expectation Conditional Maximization algorithm and graphical lasso penalty to jointly estimate multiple graphical models on heterogeneous and high-dimensional observations~\cite{HaoSLC17}. \citeauthor{wang2012quantile} introduces quantile regression to model heterogeneous data and regularize quantile regression with a non-convex penalty function to deal with ultra-high dimension. \citeauthor{PangC20}~\cite{PangC20} introduce a heterogeneous univariate outlier ensemble framework which ensembles a set of heterogeneous univariate outlier detectors optimized to capture different distribution of each individual feature. Inspired by the success of these methods, we tailor the proposed network to flexibly capture complex data characteristics and improve the practicality of CBR systems in real-life applications.

\subsection{Data-dependent Hashing}
Parallel to the data-independent hash, this paper mainly discusses data-dependent hashing from the perspective of non-deep hashing and deep hashing.

Data-dependent hashing achieves superior performance and learns the similarity-preserving hash functions from data by minimizing the gap between the similarity in the original space and that in the hash code space. Early studies present various approaches focusing on non-deep hashing, which relies on handcraft features to learn hash codes and hash functions. According to hash functions, those methods are roughly categorized into linear hash functions, kernel hash functions, and eigenfunction hash functions~\cite{WangZSSS18}. For example, spectral hashing~\cite{WeissTF08} and hashing with graphs~\cite{LiuWKC11} are representative algorithms with eigenfunction hash functions. ICA hashing~\cite{HeCRB11} and LDA hashing~\cite{SongYYHS13} are linear hashing methods. Non-deep hashing often introduces code balance constraints to avoid learning collapse and facilitate the generation of compact hash codes, achieving desirable performance in similarity retrieval.

With rapid progress made in deep representation learning, deep hashing has achieved significantly better performance than non-deep hashing and has thus been widely applied. Deep hashing methods build neural hash functions to obtain robust and powerful feature representations for complex data and learn neural hash functions and hash codes simultaneously~\cite{LiuWSC19}. The earliest work in deep hashing, semantic hashing~\cite{SalakhutdinovH09}, utilizes a deep generative model to handle text data. Subsequently, extensive studies, e.g., \cite{ZhuL0C16,YangLC18,CaoLL018}, introduce successive CNN-based networks to capture high-level features from images and customize learning objectives to preserve similarity. Inspired by these deep hashing methods, we build a specific neural network to learn feature representations from high-dimensional and heterogeneous data and utilize a pairwise loss function to learn neural hash functions and hash codes.



\section{Deep Hashing Network}
Supervised learning has been prevalent and successful in achieving high-quality semantic hash codes. Below, we outline the general settings in supervised hashing. Assume $\mathcal{X}\subset\mathbb{R}^d$ and $\mathcal{Y}\subset\{-1,1\}^r$ refer to the input (original) space and binary hash space respectively, where $d$ and $r$ denote their respective dimensions. Let us denote the given/calculated pairwise supervision of similarity information $\mathbf{S}\in\{0,1\}^{n\times n}$ for $n$ data points where $s_{ij}=1$ if data points $\mathbf{x}_i$ and $\mathbf{x}_j$ in $\mathcal{X}$ are semantically similar, and $s_{ij}=0$ otherwise. The aim of supervised hashing is to learn a mapping function $H_\phi:=\mathcal{X}\to\mathcal{Y}$ with parameters $\phi$ (e.g., a neural network) by minimizing the gaps between the similarity in the input space and that in the hash space.

Taking advantage of recent advances in deep hashing, we introduce a deep hashing network tailored for high-dimensional and heterogeneous data. The network architecture is shown in Fig. \ref{fig:network} which contains four components: feature embedding, multiview feature interaction, fully-connected layers, and a hashing layer. Next, we introduce the components in detail.

\begin{figure}
    \centering
    \includegraphics[width=0.9\linewidth]{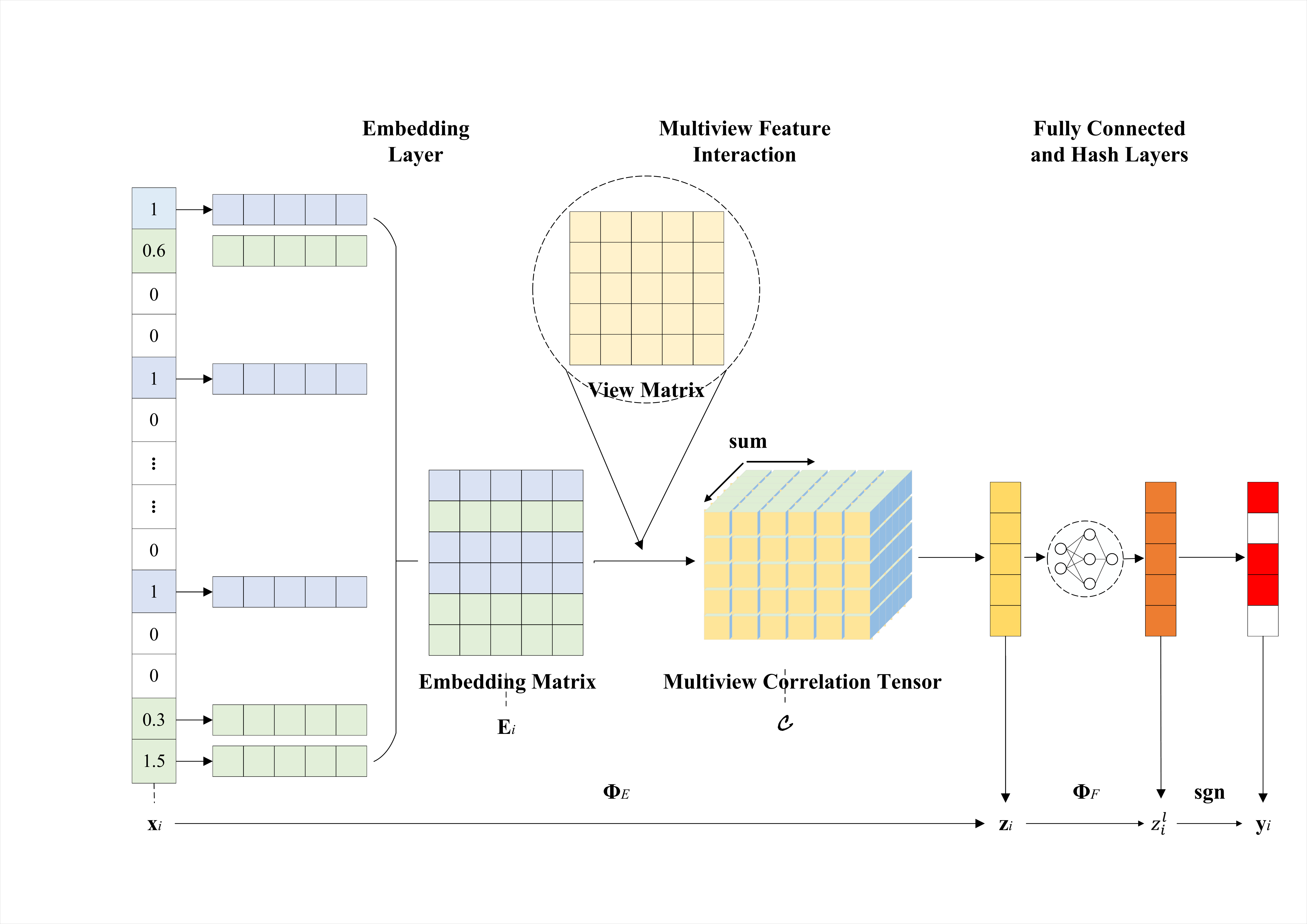}
    \caption{The architecture of deep hashing network.}
    \label{fig:network}
\end{figure}

\subsection{Feature Embedding}
Due to the high dimensionality and heterogeneity of case data, it is challenging to directly feed the feature vector of each case into neural networks. Inspired by the idea of performing feature embedding in \cite{9404857}, we introduce position embedding to represent heterogeneous features in a unified distributional space to tackle the above issues. Specifically, assume $\mathbf{X}:=\{\mathbf{x}_i\in\mathcal{X}\}_{i=1}^n$ as the input data of $n$ cases, and let $\mathbf{x}_i=[x_{i1},x_{i2}\cdots,x_{id}]^T\in \mathbb{R}^{d\times 1}$ denote the feature vector of the $i$-th case, i.e., $\mathbf{x}_i\in \mathbf{X}$. Note that all categorical features in the raw feature vector are converted to binary features by one-hot encoding. Given the $j$-th feature in the $i$-th case, i.e., $x_{ij}\in\mathbf{x}_i$ where $i\in \{1,\cdots,n\}$ and $j\in \{1,\cdots,d\}$, we have:
\begin{equation}
\label{eqn:feature_embedding}
    \mathbf{e}_{ij}=\left\{
        \begin{aligned}
            &\mathbf{0}, \ &x_{ij}=0\\
            &x_{ij}\mathbf{w}^p_j,\ &x_{ij}\neq 0
        \end{aligned}
        \right.
\end{equation}
where $\mathbf{e}_{ij}$ denotes the embedding vector for the feature $x_{ij}$, and $\mathbf{w}_j^p\in \mathbb{R}^{k_w\times 1}$ denote the embedding vector for the $j$-th position and $k_w$ denotes the embedding dimension. Here, we introduce a position embedding matrix $\mathbf{W}^p=\{\mathbf{w}_1^p,\mathbf{w}_2^p,\cdots,\mathbf{w}_d^p\}\in\mathbb{R}^{k_w\times d}$ to represent each feature position, where $\mathbf{w}_j$ corresponds to the $j$-th position in each feature vector. Accordingly, we have $\mathbf{E}_i=\{\mathbf{e}_{i1},\cdots,\mathbf{e}_{id}\}\in\mathbb{R}^{k_w\times d}$ as the feature embedding matrix for case $\mathbf{x}_i$.

According to Eq. (\ref{eqn:feature_embedding}), the embedding vector of each feature is the multiplication between its feature value and position embedding vector. Intuitively, the position embedding projects heterogeneous features onto the unified space $\mathbb{R}^{k_w}$ where it captures the heterogeneous feature couplings. In addition, since the embeddings of zero-valued (binary or numerical) features are constantly $\mathbf{0}\in \mathbb{R}^{k_w\times 1}$, the resultant large proportion of all-zero vectors in each feature embedding matrix $\mathbf{E}$ facilitates to filtrate zero-valued features and to alleviate the sparsity  in high-dimensional cases (see more analysis in Section \ref{sec:mfi}).

\subsection{Multiview Feature Interaction}
\label{sec:mfi}
To obtain the representation for each case, we propose a multiview feature interaction module to aggregate feature embeddings. Specifically, we introduce a view matrix $\mathbf{V}\in\mathbb{R}^{k_w\times k_v}$ ($k_v$ to denote the view dimension) and then resort to CANDECOMP/PARAFAC (CP) factorization to calculate the case representation below. Given case $x_i$, we have:
\begin{equation}
\label{eqn:case_embedding}
    \mathbf{z}_i=\sum_{p=1}^{d}\sum_{q=p+1}^{d}\left.\left[\sum_{j=1}^{k_w}\mathbf{E}_{i,j\cdot}\circ\mathbf{E}_{i,j\cdot}\circ\mathbf{V}_{j\cdot}\right]\right|_{pq\cdot}
\end{equation}
where $\mathbf{E}_{i,j\cdot}\in\mathbb{R}^{d\times 1}$ denotes the $j$-th row in the feature embedding matrix $\mathbf{E}_i$, $\mathbf{V}_j\cdot\in\mathbb{R}^{k_v\times 1}$ is the $j$-th row in the view matrix $\mathbf{V}$ and $\circ$ denotes the outer product. In Eq. (\ref{eqn:case_embedding}), we calculate the multiview correlations (a tensor) among all features in $\mathbf{x}_i$, i.e., the CP term in the brackets denoted as $\mathcal{C}\in\mathbb{R}^{d\times d\times k_v}$, and we then sum over the correlation matrix to each view\footnote{Since the tensor $\mathcal{C}$ is symmetric, the summation is performed upon the upper-right elements of the correlation matrix of each view.}, i.e., the first two dimensions of the tensor $\mathcal{C}$, to generate the representation vector $\mathbf{z}_i\in\mathbb{R}^{k_v\times 1}$ for case $\mathbf{x}_i$. For convenience, we denote the projection process of feature embedding and multiview feature interaction as $\Phi_E:\mathbf{X}\mapsto\mathbf{Z}$, where $\mathbf{Z}=\{\mathbf{z}_1,\mathbf{z}_2,\cdots,\mathbf{z}_n\}\in\mathbb{R}^{k_v\times n}$.

Let us expand Eq. (\ref{eqn:case_embedding}) and probe the computation of each element in the resultant $\mathbf{z}_i$ for the case $\mathbf{x}_i$:
\begin{equation}
\label{eqn:elements}
z_{ik}=\sum_{p=1}^{d}\sum_{q=p+1}^{d}\left.\left[\sum_{j=1}^{k_w}v_{jk}\mathbf{E}_{i,j\cdot}\circ\mathbf{E}_{i,j\cdot}\right]\right|_{pq}=\sum_{p=1}^{d}\sum_{q=p+1}^{d}\langle\mathbf{E}_{i,\cdot p},\mathbf{E}_{i,\cdot q}*\mathbf{V}_{\cdot k}\rangle=\sum_{p=1}^{d}\sum_{q=p+1}^{d}\langle\mathbf{e}_{ip},\mathbf{e}_{iq}*\mathbf{V}_{\cdot k}\rangle.
\end{equation}
where $\mathbf{E}_{i,\cdot p}\in\mathbb{R}^{k_w\times 1}$ and $\mathbf{V}_{\cdot k}\in\mathbb{R}^{k_w\times 1}$ are the $p$-th column (i.e., the embedding vector of the $p$-th feature $\mathbf{e}_{ip}$) and $k$-th column $\mathbf{V}$ respectively, and $\langle\cdot,\cdot\rangle$ denotes the inner product and $*$ denotes the element-wise product. 

From Eq. (\ref{eqn:elements}), we easily find that all-zero vectors (i.e., the embeddings of zero-valued features) are absorbed in the summation and do not count in calculating the embedding vector $\mathbf{z}_i$ for the case $\mathbf{x}_i$. This treatment  efficiently extracts informative features from high-dimensional sparse cases. In addition, the multiview feature interaction module calculates the pairwise (second-order) interaction between any two features by  inner product and introduces a view matrix to diversify the calculation of feature interactions. Intuitively, the module captures inter-feature couplings between (heterogeneous) features from multiple views referring to the learned view matrix $\mathbf{V}$, accordingly learning diversified feature couplings for high-level layers. In addition, the interaction module can capture first-order interactions by manipulating the feature vector of each case $\mathbf{x}$, specifically concatenating a constant value of $1$ with $\mathbf{x}$ and obtaining an extended feature vector $[1,\mathbf{x}]$. In summary, \textit{Feature Multiview Interaction} module has three advantages:
\begin{itemize}
	\item It learns feature interactions from multiple views to capture diverse inter-feature couplings between heterogeneous features, which is more effective than feed-forward neural networks in capturing intrinsic feature correlations~\cite{BeutelCJXLGC18,9404857}.
	\item It effectively filtrates zero-valued features and generates informative representation for each case, shielding the influence of the high dimensionality and unpredictable sparsity.
	\item It does not introduce extra parameters except for the view matrix and efficiently performs the calculation with time complexity of $O(d k_w k_v)$. We provide detailed analysis in Section \ref{sec:complexity_analysis}.
\end{itemize}

\subsection{Fully-connected and Hash Layers}
After obtaining the dense representation vectors for all cases, we feed case representation into full-connected layers and subsequently a hash layer, which learn high-level semantic representations and  generate binary hash codes respectively. Specifically, we feed the representation $\mathbf{z}_i$ of any case $\mathbf{x}_i$ into $l$ fully connected layers. Each layer learns a nonlinear mapping:
\begin{equation}
\mathbf{z}_i^{\ell}=\sigma^{\ell}(\mathbf{W}^{\ell}\mathbf{z}_i^{\ell - 1}+\mathbf{b}^{\ell})
\end{equation}
where $\mathbf{W}^{\ell}$ and $\mathbf{b}^{\ell}$ are the weight and bias parameters of the $\ell$-th layer, $\sigma^{\ell}$ is the corresponding activation function, and $\mathbf{z}_i^{\ell}$ denotes the $\ell$-th layer hidden representation of case $\mathbf{x}_i$ (note that $\mathbf{z}_i^0=\mathbf{z}_i$). For convenience, let us denote the $l$ fully connected layers as $\Phi_F: \mathbf{Z}\mapsto\mathbf{Z}^l$, where $\mathbf{Z}^l=\{\mathbf{z}_1^l,\mathbf{z}_2^l,\cdots,\mathbf{z}_n^l\}\in\mathbb{R}^{r\times n}$ and $r$ denotes the dimension of hash codes. We then obtain the hash code by feeding $l$-layer output $\mathbf{z}_i^l$ into the hash layer which contains a sign function. Formally, we have the hash code for the case $\mathbf{x}_i$: $\mathbf{y}_i=\sgn(\mathbf{z}_i^l)$ and $\mathbf{Y}=\{\mathbf{y}_1,\cdots,\mathbf{y}_n\}\in\{-1,1\}^{r\times n}$. Accordingly, the hash mapping function, i.e., $H_\phi:\mathbf{X}\to\mathbf{Y}$, is defined below:%
\begin{equation}
    \label{equ:hashfunction}
    H_\phi(\mathbf{X})=\sgn\left(\Phi_F(\Phi_E(\mathbf{X}))\right)=\mathbf{Y}
\end{equation}
where $H_\phi$ represents the hashing network shown in Fig. \ref{fig:network}, $\phi=\{\mathbf{V},\mathbf{W}^p,\mathbf{W}^1,\cdots,\mathbf{W}^l,\mathbf{b}^1,\cdots,\mathbf{b}^l\}$ denotes the network parameters and $\operatorname{sgn(\cdot)}$ denotes the elementwise sign function, which results in -1 if its input is negative, otherwise 1.

\subsection{Learning Objectives}
Generally, the Hamming distance $d_H(\cdot,\cdot)$ is utilized to measure the distance between hash codes, which can be achieved by inner product $\langle\cdot,\cdot\rangle$, i.e., $d_H(\cdot,\cdot)=\frac{1}{2}(r-\langle\cdot,\cdot\rangle)$. Accordingly, we apply the inner product to define the following likelihood function: given a pair of cases $\mathbf{x}_i,\mathbf{x}_j$ and their hash codes $\mathbf{y}_i,\mathbf{y}_j$, to estimate the similarity $s_{ij}$ of them, we have
\begin{equation}
P(s_{ij}|\hat{s}_{ij})=\left\{
        \begin{aligned}
            \sigma(\alpha\hat{s}_{ij}),\ s_{ij}=1\\
            1-\sigma(\alpha\hat{s}_{ij}),\ s_{ij}=0
        \end{aligned}
        \right.
\end{equation}
where $\hat{s}_{ij}=\langle \mathbf{y}_i,\mathbf{y}_j\rangle$ denotes the estimated similarity between hash codes, $\sigma$ is a Sigmoid function to scale the inner product into a distribution, and $\alpha\in (0,1]$ is a scaling hyperparameter to control the bandwidth of Sigmoid function. Smaller $\alpha$ gives rise to a smaller saturation zone, where the Sigmoid function has zero gradient. In addition, the larger $\hat{s}_{ij}$ is, the larger $P(s_{ij}=1|\hat{s}_{ij})$ will be, i.e., a larger similarity between hash codes $\mathbf{y}_i$ and $\mathbf{y}_j$ implies a higher probability of two cases $\mathbf{x}_i$ and $\mathbf{x}_j$ being similar. Accordingly, to achieve the objective  of preserving the similarity between cases in the hash space, we maximize the likelihood of all pairs of cases (in the training set):
\begin{equation}
    \prod_{s_{ij}\in\mathbf{S}}P(s_{ij}|\hat{s}_{ij})=\prod_{s_{ij}\in\mathbf{S}}\sigma(\alpha\hat{s}_{ij})^{s_{ij}}(1-\alpha\hat{s}_{ij})^{1-s_{ij}}
\end{equation}
Taking the negative logarithm of the likelihood, we have the objective function w.r.t. the cross entropy loss:
\begin{equation}
    \label{equ:similarity_loss}
    \begin{aligned}
        \min_{H_\phi}\mathcal{L}(\mathbf{Y},\mathbf{S})&=\min_{H}-\sum_{s_{ij}\in\mathbf{S}} \log P(s_{ij}|\hat{s}_{ij})\\
        &=\min_{H_\phi}-\sum_{s_{ij}\in\mathbf{S}}(\alpha s_{ij}\hat{s}_{ij}-\log(1+e^{\alpha \hat{s}_{ij}}))\\
		 &=\min_{H_\phi}-\alpha\mathbf{S}*(\mathbf{Y}^T\mathbf{Y})+\log(1+e^{\alpha\mathbf{Y}^T\mathbf{Y}})
    \end{aligned}
\end{equation}
where $\hat{s}_{ij}=\langle \mathbf{y}_i,\mathbf{y}_j\rangle=\langle H_\phi(\mathbf{x}_i),H_\phi(\mathbf{x}_j)\rangle$.



Due to the binary constraints in Eq. (\ref{equ:hashfunction}), we relax the hash code to the final output of the fully-connected layers, i.e., $\mathbf{y}_i\approx\mathbf{z}_i^l$, to approximate the non-differentiable $\sgn$ function. Therefore, we have the approximate mapping function: $\bar{H}_{\phi}(\mathbf{X})=\mathbf{Z}^l=\Phi_F(\Phi_E(\mathbf{X}))$. To squash the $l$ layer representation $\mathbf{Z}^l$ within $[-1,1]$, we encourage $\mathbf{Z}^l$ to be binary by utilizing the Sigmoid-like function: $\sigma^{l}(x)=2/(1+e^{x})-1$. However, the continuous relaxation will cause two important issues: 1) introducing uncontrollable quantization error when binarizing $\bar{H}_{\phi}(\mathbf{X})$ to $\mathbf{Y}$, and 2) raising approximation error by performing inner product on $\bar{H}_{\phi}(\mathbf{X})$ as the surrogate of $\mathbf{Y}$~\cite{ZhuL0C16}. To control the quantization error and approximation error, we introduce a quantization regularizer to minimize the difference between the hash codes and their relaxation surrogate shown below:
\begin{equation}
    \label{equ:upper_bound}
    \mathcal{R}=\Vert \mathbf{Z}^l-\mathbf{Y} \Vert _2=\Vert\bar H_{\phi}(\mathbf{X})- H_{\phi}(\mathbf{X})\Vert _2.
\end{equation}
\begin{Theorem}[Regularizer Upper Bound] The quantization regularizer on $H_{\phi}$ is upper bounded by $nd-\mathbf{tr}(\bar{H}_{\phi}(\mathbf{X})^T \bar{H}_{\phi}(\mathbf{X}))$, i.e., 
    \begin{equation}
        \Vert\bar H_{\phi}(\mathbf{X})- H_{\phi}(\mathbf{X})\Vert _2\leq nd-\mathbf{tr}(\bar{H}_{\phi}(\mathbf{X})^T \bar{H}_{\phi}(\mathbf{X})).
    \end{equation}
\end{Theorem}
\begin{proof}
Since we have $H_{\phi}(\mathbf{X})=\sgn(\bar{H}_{\phi}(\mathbf{X}))$, $H(\mathbf{X})$ and $\bar{H}_{\phi}(\mathbf{X})$ have the same sign, yielding that
    \begin{equation}
        \begin{aligned}
            \Vert\bar H_{\phi}(\mathbf{X})- H_{\phi}(\mathbf{X})\Vert _2&=\lVert|\bar{H}_{\phi}(\mathbf{X})|-|H_{\phi}(\mathbf{X})|\rVert_2\\
            &=\sum_{i=1}^n\lVert|\bar H_{\phi}(\mathbf{X}_i)|-\mathbf{1}\rVert_2\\
            &\leq \sum_{i=1}^n\lVert \bar H_{\phi}(\mathbf{X}_i)^T \bar H_{\phi}(\mathbf{X}_i)-d\rVert_1\\
            &=nd-\mathbf{tr}(\bar H_{\phi}(\mathbf{X})^T \bar H_{\phi}(\mathbf{X}))
        \end{aligned}
    \end{equation}
where $\mathbf{tr}$ is the trace function. Proved.
\end{proof}

Accordingly, we approximate hash codes $\mathbf{Y}$ with $\bar{H}_{\phi}(\mathbf{X})$ in Eq. (\ref{equ:similarity_loss}) and combine the loss function with the upper bound of the quantization regularizer, achieving the final learning objective for the proposed hashing network below:
\begin{equation}
    \label{equ:objective}
		\min_{\bar H_{\phi}}-\alpha\mathbf{S}*\Psi_{\bar H_{\phi}}+\log(1+e^{\alpha\Psi_{\bar H_{\phi}}})-\lambda\mathbf{tr}(\Psi_{\bar H_{\phi}})
\end{equation}
where $\Psi_{\bar H_{\phi}}=\bar H_{\phi}(\mathbf{X})^T \bar H_{\phi}(\mathbf{X})$, $\lambda\in (0,1)$ is a hyperparameter to balance the weights of similarity loss $\mathcal{L}$ and the quantization regularizer $\mathcal{R}$. By minimizing the learning objective in Eq. (\ref{equ:objective}), we can learn the model parameters of the proposed hashing network, i.e., $\bar H=\{\mathbf{W}^p,\mathbf{W}^1,\cdots,\mathbf{W}^l,\mathbf{b}^1,\cdots,\mathbf{b}^l\}$. Besides, hyperparameters $\{k_w,k_v,l,\alpha,\lambda\}$ are selected per empirical experiments and grid search. After the hashing network is well trained, we next introduce the network into case-based reasoning for effective storage and efficient case retrieval. In the stage of case retention, the network is adaptively updated during retaining new data (cases). For a better understanding of the CBR process, we place the detailed introduction to the adaptive update mechanisms in Section \ref{sec:case-retention}.

\section{Hashing-enabled Case-based Reasoning}

Leveraging the proposed hashing network, we propose a Hashing-enabled Case-Based Reasoning (HeCBR) model, in which the hash network is introduced to transform high-dimensional and heterogeneous cases into low-rank binary hash codes. The hash codes are similarity-preserving compact representations for cases, which provides effective and efficient similarity retrieval. As shown in Fig. \ref{fig:ahcbc}, next, we introduce each phase of the problem-solving process of HeCBR in detail.

\subsection{Case Representation}
A good case representation not only provides organization and indexing of cases for efficient case retrieval but also facilitates accurate similarity measurement. To achieve this, we adopt our proposed hash network to transform cases into compact hash codes. Based on hash codes, we construct a hash table (a form of an inverted index) to organize and index all cases. The hash table consists of buckets with each bucket indexed by a hash code. For example, if the dimension of hash code $r$ is $8$, then there are up to $2^8$ buckets with each bucket indexed by an $8$ bits hash code. Each case $\mathbf{x}_i$ is then placed into a bucket $H_{\phi}(\mathbf{x}_i)$. Due to the learned hash codes being similarity-preserving, the hash approach, different from the conventional hashing algorithm avoiding mapping two samples into the same bucket, essentially aims to maximize the probability of collision between similar cases and meanwhile minimize the probability of collision between dissimilar cases. Before constructing the hash table, we need to train the proposed hashing network $H_{\phi}$ based on the cases in the case base, where the input feature vectors of cases are given or extracted from the descriptions of the cases and the required similarity relations $\mathbf{S}$ are provided or calculated based on the solutions (labels) of the cases.

\subsection{Case Retrieval}
When a new case comes, we need to retrieve the most similar cases to the new case and leverage the solutions of the retrieved cases to solve the new case. Specifically, to search for the most similar cases to a new case $\mathbf{x}_c$, we first generate the hash code $\mathbf{y}_c$ for the case, i.e., $\mathbf{y}_c=H(\mathbf{x}_c)$. Then we retrieve the cases lying in the bucket indexed by hash code $\mathbf{y}_c$ and treat the cases as the candidates of the most similar cases to $\mathbf{x}_c$. Usually, this is followed by a reranking step: reranking the retrieved candidates according to the true distances computed using the original features of cases and attaining top-$N$ most similar cases. Note that if no cases exist in bucket $\mathbf{y}_c$, we can retrieve the cases in the nearest buckets (for example those within 2-hamming distance) with bucket $\mathbf{y}_c$ instead, and the indices of the nearest buckets can be obtained by modifying each bit in the hash code $\mathbf{b}_c$ into its alternative value in turn. Since case representation and case retention can be performed offline, case retrieval becomes the most time-consuming phase of CBR and is important to the applicability of CBR methods. Due to the benefits of hash table lookup, similar case candidates can be attained with time complexity of $O(1)$. The remaining time-consuming steps are to generate the hash code for the new case and rerank the candidate cases, which can also be completed in a short time period. We discuss the time complexity in detail in Section \ref{sec:complexity_analysis}.

\subsection{Case Reuse and Case Revision}
During the phase of case reuse, we utilize the retrieved top-$N$ most similar cases to suggest solutions for new cases. A general approach is to design a voting function and suggest the solution with the most votes for each new case. For convenience, we adopt a simple majority voting function where each case denotes one vote and thus the solution supported by most cases will be the optimal one. After suggesting solutions for new cases, we check the actual solutions for the new cases and revise the suggestion if the suggested solutions are different from the actual solutions. Next, the solved cases will be retained in the case base for future problem-solving.

\subsection{Case Retention}
\label{sec:case-retention}
In the phase of case retention, all solved cases will be retained and used to update the hash function (network) and the corresponding hash codes. During the update of the hash function $H_{phi}$, we need to consider whether or not to update the hash function at all. Accordingly, we  decide how to correct the hash function. Inspired by~\cite{CakirS15}, we adopt the hinge loss function to define an offset function:
\begin{equation}
	\label{eqn:offset}
	\mathcal{O}=\left\{
    \begin{aligned}
        \max(0,r\beta-\hat{s}_{ij})&, \ s_{ij}=1\\
		\max(0,r\beta+\hat{s}_{ij})&, \ s_{ij}=0
    \end{aligned}
\right.
\end{equation}
where $\beta\in[0,1]$ is a hyperparameter designating the extent to which the hash function may produce a loss. From Eq. (\ref{eqn:offset}), we can see the larger $\beta$ is, the larger the loss $l_{ij}$ is, indicating that more rigid similarity is preserved and more information is retained from the training pair. Imposing the offset function on the similarity loss function, we then obtain the learning objective $l_{ij}$ for adaptive update as follows:
\begin{equation}
	l_{ij}=\left\{
    \begin{aligned}
        &\max(0,r\beta-\hat{s}_{ij})\log(1+e^{-\alpha\hat{s}_{ij}}), \ s_{ij}=1\\
		&\max(0,r\beta+\hat{s}_{ij})\log(1+e^{\alpha\hat{s}_{ij}}), \ s_{ij}=0
    \end{aligned}
\right.
\end{equation}
where we do not consider the quantization regularizer for simplicity since higher $\beta$ acts as an alternative to the quantization regularizer. When $l_{ij}=0$, we do not perform any update, which effectively cuts invalid update that has small gradients. In addition, since updating the hash codes (table) is usually time-consuming, we perform the update every $n_u=100$ new cases to avoid frequently update the hash function and hash codes. The treatment also avoids the update overfitting to certain new cases. The update process can be done offline to improve efficiency.

\subsection{Complexity Analysis}
\label{sec:complexity_analysis}

\subsubsection{Time complexity of multiview feature interaction} Recalling Eq. (\ref{eqn:elements}), we reformulate the calculate for each element in $\mathbf{z}_i$ for $i$-th case. For convenience, we omit the subscript $i$ in the following.
\begin{equation}
\begin{aligned}
z_{k}&=\sum_{p=1}^{d}\sum_{q=p+1}^{d}\langle\mathbf{e}_{p},\mathbf{e}_{q}*\mathbf{V}_{\cdot k}\rangle\\
&=\frac{1}{2}\sum_{p=1}^{d}\sum_{q=1}^{d}\langle\mathbf{e}_{p},\mathbf{e}_{q}*\mathbf{V}_{\cdot k}\rangle-\frac{1}{2}\sum_{p=1}^{d}\langle\mathbf{e}_{p},\mathbf{e}_{p}*\mathbf{V}_{\cdot k}\rangle\\
&=\frac{1}{2}\left(\sum_{p=1}^{d}\sum_{q=1}^{d}\sum_{m=1}^{k_w} e_{mp} e_{mq} v_{mk}-\sum_{p=1}^{d}\sum_{m=1}^{k_w} e_{mp}^2 v_{mk}\right)\\
&=\frac{1}{2}\sum_{m=1}^{k_w}\left(\left(\sum_{p=1}^{d}e_{mp}\right)^2 v_{mk} - \sum_{p=1}^{d}e_{mp}^2 v_{mk}\right)
\end{aligned}\nonumber
\end{equation}
From the above equations, we know the computation complexity of $z_{k}$ is in $O(d k_w)$, and the complexity of calculating the embedding vector $\mathbf{z}$ for each case is $O(d k_w k_v)$. Besides, since only non-zero features kick in, we only need to sum over all non-zero features during the calculation. Thus, the complexity for calculating $\mathbf{z}$ is reduced to $O(d_n k_w k_v)$ where $d_n\ll d$ denotes the number of non-zero features.

\subsubsection{Time complexity of HeCBR}
The time complexity of HeCBR mainly comes from the phase of case retrieval in practical applications, since the other two time-consuming phases, i.e., case representation and case retention, can be done offline. In case retrieval, the time complexity consists of two parts: generating the hash code for each new case and reranking all the retrieved candidate cases, since the similar case candidates can be obtained with a time complexity of $O(1)$ in the hash table. The approximate time for generating hash code is $O(d_n k_w+d_n k_w k_v + lk_v^2+r)$ in which the four parts correspond to the time need for feature embedding, multiview feature interaction, fully-connect layers and the hash layer in the hash network respectively. Assume we obtain $n_r$ candidate cases, the time for reranking the candidates consists of $O(n_r d)$ for calculating the distances using the original features of cases and $O(n_r\log n_r)$ for ranking the distances if using the quick sort algorithm. Compared to the step of reranking, we can generate the hash codes for a batch of cases in parallel and even accelerate the calculation by a GPU processor. Thus, the time for generating hash codes can be ignored. Accordingly, the approximate time complexity of HeCBR is as follows: $O(n_r d+n_r\log n_r)$. Since we have $n_r\ll n$ which is usually the case and the average number of cases in each bucket is $n/2^r(\approx n_r)$, the complexity of HeCBR is far less than those CBR methods which traverse the case base, i.e., $n_r d+n_r\log n_r\ll n d+n\log n$. In addition, comparing to other clustering-based CBR methods, HeCBR obtains the candidate cases by hash table lookup with a time complexity $O(1)$ and also improves the retrieval efficiency.

\section{Experiments and Evaluation}
In this section, we conduct extensive experiments to investigate the following research problems:
\begin{itemize}
\item[Q1] How does HeCBR perform in terms of classification?
\item[Q2] How does HeCBR perform in terms of the similarity retrieval task?
\item[Q3] How robustly does HeCBR perform with different hyperparameters?
\item[Q4] How does the adaptive update in case retention affect the performance of HeCBR?
\end{itemize}

\begin{table}[!ht]
  \centering
  \caption{Data characteristics of eight high-dimensional sparse datasets.}
    \begin{tabular}{l|l|l|l|l|l}
    \toprule
    Dataset & Abbr. & \#instances & \#dimension & \#class & sparsity \\
    \midrule
    Internet Advertisements & ADV    & 3279  & 1557  & 2     & 0.01 \\
    Protein & PT    & 17766 & 357   & 3     & 0.29 \\
    Adult & ADT   & 45222 & 118   & 2     & 0.35 \\
    Dota2 & Dota  & 102944 & 172   & 2     & 0.08 \\
    Character Font Images & Font  & 391651 & 896   & 142   & 0.25 \\
    Movie Tweetings & MT    & 773442 & 90191 & 10    & 5.00E-05 \\
    Criteo & CT    & 1000000 & 199   & 2     & 0.2 \\
    MovieLen-1M & ML1M  & 1000209 & 9794  & 5     & 7.00E-04 \\
    \bottomrule
    \end{tabular}%
  \label{tab:data_stats}%
\end{table}%

\subsection{Experimental Settings}

\subsubsection{Datasets}
We verify the effectiveness of our proposed HeCBR on eight real-world high-dimensional heterogeneous datasets, which contains binary classification and multiclass classification problems and covers various domains such as transaction classification, movie rating prediction, font image classification, and protein classification: (1) Internet Advertisements (short for ADV) collects a set of possible advertisements on Internet pages, and the task is to classify an image into an advertisement or not\footnote{https://archive.ics.uci.edu/ml/datasets/Internet+Advertisements}. 
(2) Protein (PT), a multiclass classification dataset, contains protein information for studying the structure of proteins\footnote{https://www.csie.ntu.edu.tw/~cjlin/libsvmtools/datasets/}. (3) Adult (ADT) collects $45,222$ census records extracted from 1994 Census database, which contains both categorical and numeric attribute for classification task\footnote{https://archive.ics.uci.edu/ml/datasets/Adult}. (4) Dota2 (Dota) collects the battle formation from a popular computer game Data2 with two teams of 5 players, and the task is to predict which team wins. The above three datasets are collected from the UCI machine learning repository\footnote{https://archive.ics.uci.edu/ml/index.php}. 
(5) Character Font Images (Font) consist of images from $153$ character fonts and record a variety of font description, we select $142$ of $153$ character fonts in the experiments\footnote{https://archive.ics.uci.edu/ml/datasets/Character+Font+Images}. 
(6) Movie Tweetsing (MT) is a dataset consisting of ratings on movies that were contained in well-structured tweets on Twitter\footnote{http://github.com/sidooms/MovieTweetings}. 
(7) Criteo (CT) includes $45$ million user click records and contains both continous and categorical features\footnote{https://www.kaggle.com/c/criteo-display-ad-challenge}. Considering the computation burden, we randomly select $1$ million records from the Criteo dataset for evaluation in the experiments. 
(8) MovieLen-1m (ML1M) collects about $1$ million  anonymous ratings of approximately $3,900$ movies made by $6,040$ MovieLens users\footnote{https://grouplens.org/datasets/movielens/}.

The detailed characteristics of the eight datasets are reported in Table \ref{tab:data_stats}. The table shows the numbers of instances, features, classes and sparsity of each dataset where $sparsity$ reflects the proportion of non-zero features. Note that we convert the categorical features into binary features via one-hot encoding in each dataset. Table $\ref{tab:data_stats}$ reports the data characteristics of converted datasets.

\subsubsection{Baselines}
To investigate the performance of HeCBR, we first compare HeCBR with hashing-based methods including four state-of-the-art LSH-based CBR methods and four dimension reduction and representation methods as follows.

\begin{itemize}
	\item LSH: It incorporates the original locality-sensitive hash algorithm~\cite{IndykM98} to map the cases into binary hash codes for efficient case retrieval.
	\item WTAH: It introduces WTAHash, a sparse embedding method, to CBR that transforms the input into binary codes guaranteeing Hamming distance in the resultant space closely correlated with rank similarity measures.
	\item FlyH: It applies FlyHash~\cite{flyhash}, a hash algorithm inspired by fruit flies' olfactory circuits, to improve the performance of case retrieval on high-dimensional data.
	\item PMH: It uses PM-LSH~\cite{pmlsh}, a latest fast and accurate LSH framework based on a simple yet effective PM-tree, to improve the performance of computing $c$-ANN queries on high-dimensional data.
	\item ITQ: The method utilizes ITQ~\cite{GongLGP13}, which adopts PCA for dimension reduction and quantization, to generate binary representation for case retrieval.
	\item NFM: The method involves NFM~\cite{0001C17}, a deep neural factorization machine method employing feature embedding for handling high-dimensionality and sparsity, to extract case representation.
	\item SVAE: The method adopts SVAE~\cite{KrishnanLH18}, a sparse variational autoencoder specified to address high-dimensional and sparse data, to represent cases.
    \item M2V: The method introduce Mix2Vec~\cite{ZhuZCA20}, a state-of-the-art unsupervised mixed data representation based on mixed feature embedding, to obtain high-dimensional and heterogeneous case representation.
\end{itemize}

Note that NFM, SVAE, and M2V generate dense low-dimensional data representation rather than hash codes. To guarantee a fair comparison, we binarize the dense data representation to generate hash codes, where we add a new similarity-preserving objective under a joint-training manner for unsupervised representation methods, i.e., SVAE and M2V. Obviously, LSH, WTAH, FlyH, and PMH denotes the CBR baseline enabled by data-independent hash methods, while the other four denote competitors equipped with data-dependent (deep) hash methods. In addition, we compare HeCBR with the state-of-the-art CBR methods to investigate CBR performance in terms of retrieval efficiency.

\begin{itemize}
	\item SNCBR~\cite{PetrovicMS11}: It utilizes a heuristic simulated annealing algorithm to optimize the weight allocations in similarity calculation.
	\item GACBR~\cite{GuLZ17}: The method adopts a genetic algorithm to optimize the feature weights to improve similarity calculation.
	\item MCCBR~\cite{YanSG14}: The method obtains more rational weight allocations by the predefined evolution and communication rules and a regional sub-algorithm based on SA.
	\item ANNCBR~\cite{BiswasSinha-524}: It trains a neural network by predicting classification labels and treats connection weights as the corresponding attribute weights.
	\item HCBR~\cite{ZhangSNC19}: A state-of-the-art CBR model introduces conceptual clustering to capture structural relations among cases and incorporates the relations for calculating structural similarity.
\end{itemize}

These state-of-the-art methods are deliberately chosen for the following considerations: 1) we introduce state-of-the-art LSH into CBR methods since prior hash-based CBR models generally adopt LSH to improve retrieval efficiency; 2) we compare HeCBR with dimension reduction and representation methods to verify the superiority of HeCBR in addressing high-dimensionality and heterogeneity issues; 3) we choose the state-of-the-art CBR methods to investigate the retrieval efficiency and CBR performance of HeCBR over traditional CBR methods. To our best knowledge, few deep models are incorporated to improve the performance of CBR, let alone deep hash models. Therefore, we do not compare HeCBR with the state-of-the-art deep hashing methods.

\subsubsection{Evaluation Measures}
In the experiments, we perform 5-fold cross-validation and report the average evaluation results in the experiments. Specifically, we employ accuracy and AUC (area under the ROC curve) to evaluate classification performance. For multi-class problems, AUC is calculated below:
\begin{equation}
AUC=\frac{2}{|L|\times(|L|-1)}\sum_{i<j}\frac{A_{ij}+A_{ji}}{2}
\end{equation}
where $|L|$ denotes the number of class labels, and $A_{ij}$ and $A_{ji}$ are the AUC values calculated by considering only cases from classes $i$ and $j$. To evaluate the retrieval performance, we also adopt mean average precision (MAP) and precision@K (Prec@K). Specifically, given a top-$N$ ranked item set $\hat{R}_N$ and the target ground-truth item set $R$, Prec@$N$ is calculated as follows:
\begin{equation}
    Prec@N=\frac{|R\cap \hat{R}_N|}{N}.
\end{equation}
MAP@N is calculated via the mean of the average precision (AP@N) on all cases, and AP@N is defined by:
\begin{equation}
    AP@N=\frac{\sum_{i=1}^{N}Prec@i\times rel(i)}{\min(|R|,N)},
\end{equation}
where $rel(i)$ equals $1$ if $i\in R$, otherwise $0$.

We perform a grid search of the number of hashtables over $\{4,8,16,32,64\}$ and the bucket width over $\{5,10,15,20,25\}$ on validation sets to find the optimal configuration for the four LSH-based CBR baselines. In addition, we adopt the parameter settings recommended by the authors for the other comparative baselines. For our proposed HeCBR, we tune the embedding dimension $k_w$ and view dimension $k_v$ over $\{15,32,64,126,256\}$ by a grid search and adopt $k_w=64$, $k_v=64$, $r=36$ and $l=3$ for the fully-connected layers $\Phi_F$, i.e., the shape of dimension in $\Phi_F$ is fixed as $64-128-128-r$ if not specified. In addition, we perform a grid search over $\alpha\in\{0.2,0.4,0.6,0.8\}$ and $\lambda\in\{0,0.2,0.4,0.6,0.8\}$ with a step of $0.2$, to obtain the best results, and fixed $\beta=0.5$. To facilitate fair comparison, we adopt the same parameter settings for the variants of HeCBR. In addition, for all comparative methods, we select $N=10$ most similar cases to suggest class labels and report the best results if not specified.

\subsection{Classification Evaluation (Q1)}
In this section, we investigate the performance of HeCBR in terms of case-based classification to verify the advance of HeCBR in handling high-dimensionality and heterogeneity.

\begin{table}[!t]
  \centering
  
  \caption{Comparison of classification performance in terms of accuracy and the area under the ROC curve (AUC). The results are obtained with 36-bit binary codes and top-$10$ most similar cases.}
  \scalebox{0.95}{
    \begin{tabular}{l|ccccc|ccc|cc}
    \toprule
    \multicolumn{11}{c}{Accuracy} \\
    \toprule
    Dataset & LSH   & WATH  & FlyH  & PMH   & ITQ   & NFM   & SVAE  & M2V   & HeCBR$\dagger$ & HeCBR \\
    \midrule
    ADV   & 0.7975 & 0.7996 & 0.7966 & 0.8261 & \underline{0.9258} & 0.8589 & 0.8766 & 0.8522 & 0.9484 & \textbf{0.967} \\
    PT    & 0.385 & 0.3875 & 0.3916 & 0.4136 & \underline{\textbf{0.5054}} & 0.4189 & 0.4408 & 0.4182 & 0.4807 & 0.505 \\
    ADT   & 0.7102 & 0.6793 & 0.7223 & 0.7467 & 0.7573 & 0.7621 & 0.7631 & \underline{0.7641} & 0.8004 & \textbf{0.8129} \\
    Dota  & 0.4813 & 0.4956 & 0.4927 & 0.5011 & 0.5015 & 0.4993 & 0.5045 & \underline{0.5132} & 0.5315 & \textbf{0.5446} \\
    Font  & 0.4203 & 0.4171 & 0.4207 & 0.4234 & \underline{\textbf{0.5083}} & 0.4598 & 0.4721 & 0.4807 & 0.4941 & 0.5047 \\
    MT    & 0.2005 & 0.1858 & 0.1907 & 0.2031 & 0.1917 & 0.2038 & \underline{0.2087} & 0.2026 & 0.2128 & \textbf{0.2172} \\
    CT    & 0.7233 & 0.7192 & 0.7274 & 0.7342 & 0.7352 & 0.7292 & 0.7356 & \underline{0.7384} & 0.7413 & \textbf{0.7445} \\
    ML1M  & 0.2768 & 0.2822 & 0.2785 & 0.2821 & 0.2903 & 0.291 & 0.2933 & \underline{0.2982} & 0.3089 & \textbf{0.3168} \\
    \midrule
    Avg.R & 8.75 & 8.875  & 8.5 & 6.625 & 4.25 & 5.5   & 4.625     & 4.375  & 2.25  & \textbf{1.25} \\
    \bottomrule
    \multicolumn{11}{c}{AUC} \\
    \toprule
    Dataset & LSH   & WATH  & FlyH  & PMH   & ITQ   & NFM   & SVAE  & M2V   & HeCBR$\dagger$ & HeCBR \\
    \midrule
    ADV   & 0.578 & 0.5832 & 0.5724 & 0.6032 & \underline{0.8806} & 0.6288 & 0.7013 & 0.6877 & 0.907 & \textbf{0.9162} \\
    PT    & 0.5022 & 0.5167 & 0.5153 & 0.5233 & \underline{\textbf{0.6512}} & 0.5275 & 0.5382 & 0.5276 & 0.6328 & 0.6465 \\
    ADT   & 0.6427 & 0.6198 & 0.6832 & 0.6853 & 0.6994 & 0.7213 & 0.7412 & \underline{0.7447} & 0.8273 & \textbf{0.8347} \\
    Dota  & 0.5004 & 0.5061 & 0.5033 & \underline{0.5126} & 0.5087 & 0.5016 & 0.5041 & 0.507 & 0.5444 & \textbf{0.555} \\
    Font  & 0.5177 & 0.509 & 0.511 & 0.5144 & \underline{0.6126} & 0.5391 & 0.5402 & 0.5652 & 0.6017 & \textbf{0.6163} \\
    MT    & 0.5023 & 0.5011 & 0.5017 & 0.5033 & 0.5031 & 0.5115 & \underline{0.5186} & 0.5163 & 0.5215 & \textbf{0.5371} \\
    CT    & 0.5248 & 0.5248 & 0.5248 & 0.5248 & 0.5248 & \underline{0.5515} & 0.5492 & 0.5368 & 0.5763 & \textbf{0.5908} \\
    ML1M  & 0.5003 & 0.5005 & 0.5003 & 0.5103 & 0.5132 & 0.5189 & 0.5147 & \underline{0.5219} & 0.5546 & \textbf{0.5679} \\
    \midrule
    Avg.R & 8.6875 & 8.375   & 8.6875 & 6.65 & 4.375 & 5.5   & 4.5   & 4.875  & 2.25  & \textbf{1.125} \\
    \bottomrule
    \end{tabular}
    }%
  \label{tab:hash_comparison}%
\end{table}%

\subsubsection{Classification Performance}
To verify the effectiveness of HeCBR, we compare it with several state-of-the-art LSH-based CBR methods and representation learning methods in terms of the case-based classification task\footnote{Case-based classification is a common task in CBR and is convenient to evaluate the performance of CBR. To address a classification task, we adopt the majority voting to suggest the class label with the most votes (i.e., the label with the most number of supported cases) in top-$N$ most similar ones.}. The performance of all methods in terms of accuracy and AUC are reported in Table \ref{tab:hash_comparison} where the best results in each row are highlighted in bold and the best baseline method is underlined for each dataset. Avg.R denotes the average rank of each method over all the datasets, and HeCBR$\dagger$ denotes the variant of HeCBC which does not retain solved cases to update hash functions and hash codes. Table \ref{tab:hash_comparison} enables the following key observations.

\begin{itemize}
    \item First, compared with the baselines, HeCBR$\dagger$ and HeCBR achieve the best classification performance and rank top-2 in terms of the average rank of accuracy and AUC on all datasets. Especially, HeCBR$\dagger$ and HeCBR significantly improve accuracy by about $3.6\%-6.4\%$ and AUC by more than $6.2\%$ over the best baselines on Dota, ML1M, and ADT and obtain a desirable improvement of classification performance on the other datasets except for datasets PT and Font. All the results report that HeCBR$\dagger$ and HeCBR outperform the baselines in terms of the classification task.
    \item Second, LSH-enabled baselines perform much worse than the other comparative methods on most datasets, which verifies the superiority of data-dependent hashing methods over the data-independent ones in capturing data-specific features and addressing complex data issues, e.g., data heterogeneity. In addition, we observe the average accuracy and ACU of HeCBR$\dagger$ and HeCBR improve respectively by more than $9\%$ and $17\%$ over the LSH-enabled baselines. The results show that HeCBR outperforms state-of-the-art hash-enabled CBR methods which generally adopt hash methods from the LSH family to improve case retrieval efficiency.
    \item Third, comparing HeCBR$\dagger$ with the state-of-the-art data-dependent hash baselines, HeCBR$\dagger$ achieves higher average accuracy and AUC ranks on all datasets, indicating much better and more robust classification performance. Specifically, HeCBR$\dagger$ improves accuracy by more than $3.5\%$ and AUC by more than $6.3\%$ over the baselines on Dota, ML1M, and ADT. The results indicate that our proposed adaptive hashing network, especially the proposed Multiview Feature Interaction, is more effective than the baselines in handling high-dimensionality and heterogeneity and representing complex cases.
    \item Fourth, HeCBR$\dagger$ performs worse than ITQ on datasets PT and Font, which is attributable that high-dimensional datasets, i.e., PT and Font (image), have a large proportion of numeric features where PCA performs better than the feature representation methods (NFM, M2V and HeCBR) and reconstruction-based representation method (SVAE) to obtain effective and concise case representation under the high-dimensionality setting.
    \item \item Finally, HeCBR consistently performs better than HeCBR$\dagger$ and achieves obvious better accuracy and AUC than HeCBR$\dagger$ on all datasets. The results reflect that the incrementally retained solved cases to update hash functions is beneficial to retaining new knowledge for future problem-solving and our proposed update mechanism is effective in adaptively updating hash functions and hash codes.
\end{itemize}

\subsubsection{Ablation Study}
To investigate the effectiveness of the proposed feature embedding and multiview feature interaction, we introduce three variants (denoted as \textit{max}, \textit{concat} and \textit{plain}) of HeCBR for ablation study which replaces the proposed feature embedding or multiview feature interaction in HeCBR with specific designs:
\begin{itemize}
    \item \textit{max}: The variant performs the max pooling upon the feature embedding matrix $\mathbf{E}$, i.e., $max:\mathbf{E}\to\{\max(\mathbf{E}_{1\cdot}),\cdots,\max(\mathbf{E}_{k_w\cdot})\}\in\mathbb{R}^{k_w\times 1}$.
    \item \textit{concat}: The variant simply concatenates all feature embedding vectors, i.e., $concat:\mathbf{E}\to\operatorname{concat}(\{\mathbf{e}_i,\cdots,\mathbf{e}_d\})\in\mathbb{R}^{k_w\times d}$.
    \item \textit{plain}: The variant adopts a fully-connected layer to transform an original case vector (i.e., $\mathbf{x}$) to a $k_w$-sized vector, i.e., $plain:\mathbf{x}\to\mathbf{w}\mathbf{x}+b\in\mathbb{R}^{k_w\times 1}$.
\end{itemize}

Let's denote HeCBR as \textit{interaction} to indicate that HeCBR calculates the multiview feature interactions based on the feature embeddings.
The experimental results comparing the above four methods on the eight datasets are shown in Fig.\ref{fig:ablation_results}, where we perform the comparative methods with $k_w=64$ and different hash code dimensions $r\in\{12,24,36,48\}$ and show the corresponding hyperparmater settings on $\lambda$ and $\alpha$. From the results, we obtain the following observations:
\begin{itemize}
    \item Compared with the variants, \textit{interaction} achieves better performance in terms of accuracy, especially on ADV, PT, Dota, Font, and MT, and catches up with the variants as the code dimension increases on ADT, CT and ML1M. The results show the better superiority and stability of \textit{interaction}, i.e., HeCBR, under different data characteristics.
    
    \item Specifically, \textit{interaction} outperforms the variants \textit{max} and \textit{concat} on all datasets except for ADT and CT, which demonstrates the effectiveness of the proposed multiview feature interaction in capturing feature couplings. The variant \textit{max} downsamples sharpest features, while \textit{concat} retains all features. Their superiority on ADT and CT is attributed that they are suitable for dense data.
    
    \item Also, \textit{plain} performs worse than \textit{interaction} on most of the datasets, which further indicates the contribution of the proposed feature embedding for learning heterogeneity. In addition, \textit{plain} shows competitive and even better performance than \textit{max} and \textit{concat} on ADV, PT, Font and MT. This is reasonable since \textit{plain} performs direct transformation on raw features and learns effective representations from datasets with more numeric attributes, e.g., ADV, PT and Font.
\end{itemize}

In summary, \textit{interaction} achieves much better performance, and the accuracy of variants varies largely with different datasets. The results demonstrate the effectiveness of our proposed feature embedding and multiview feature interaction in handling high-dimensionality and heterogeneity issues.

\begin{figure*}[!t]
    \centering
    \subfigure[ADV $(\lambda=0.2,\alpha=0.6)$]{
    \begin{minipage}[t]{0.24\linewidth}
    \centering
    \includegraphics[width=\linewidth]{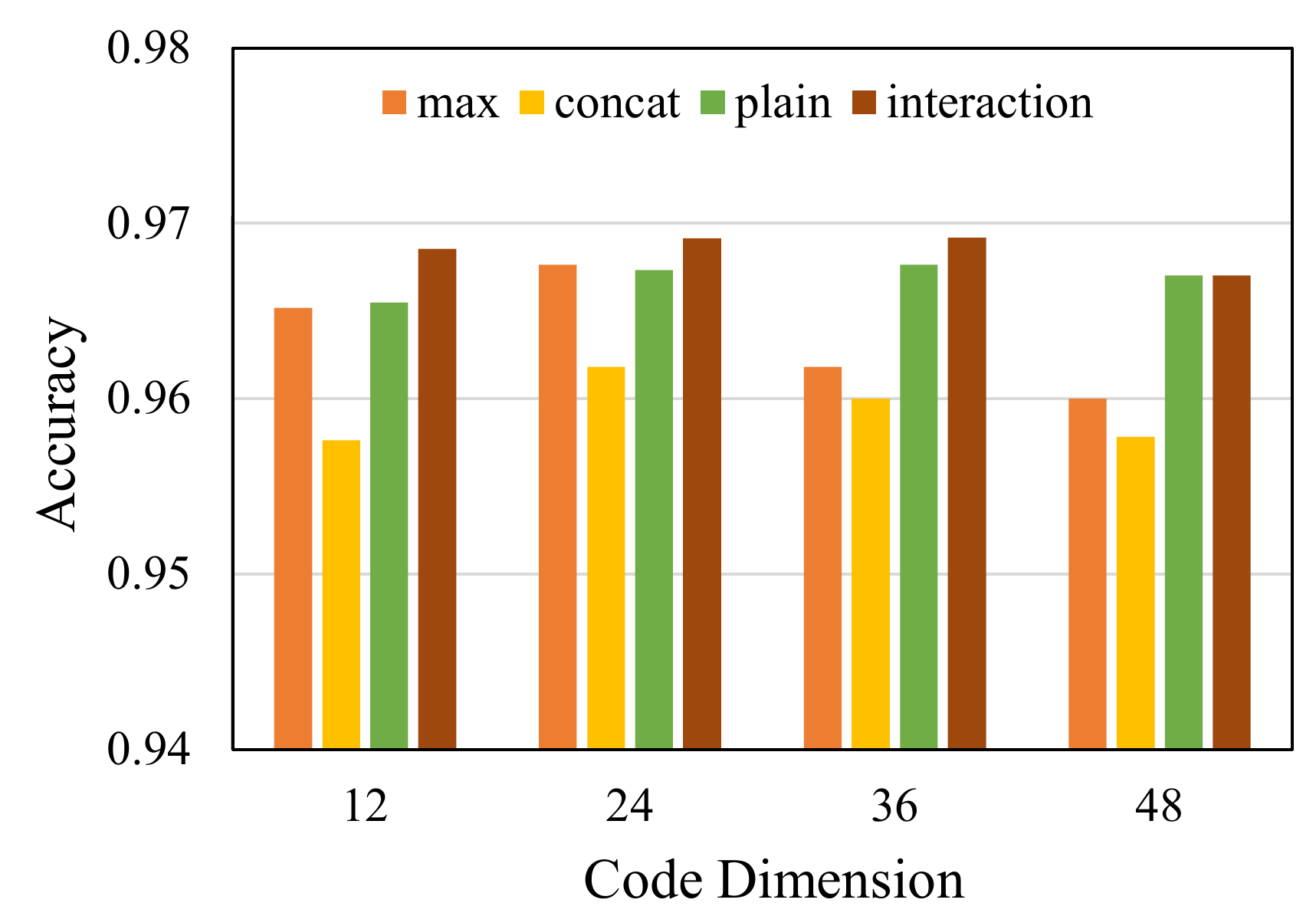}
    \end{minipage}%
    }%
	 \subfigure[PT $(\lambda=0.1,\alpha=0.6)$]{
    \begin{minipage}[t]{0.24\linewidth}
    \centering
    \includegraphics[width=\linewidth]{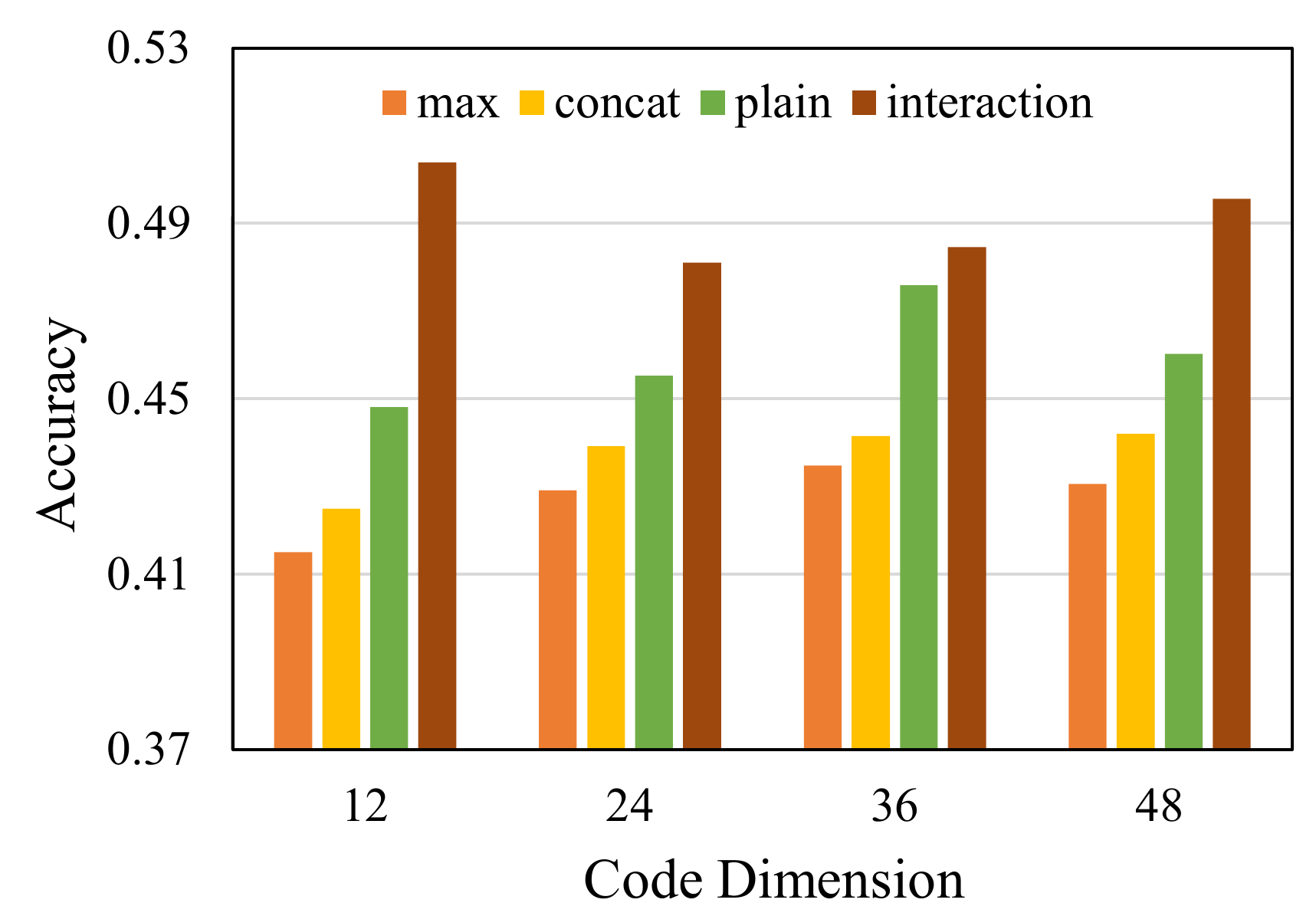}
    \end{minipage}%
    }%
    \subfigure[ADT $(\lambda=0.2,\alpha=0.6)$]{
    \begin{minipage}[t]{0.24\linewidth}
    \centering
    \includegraphics[width=\linewidth]{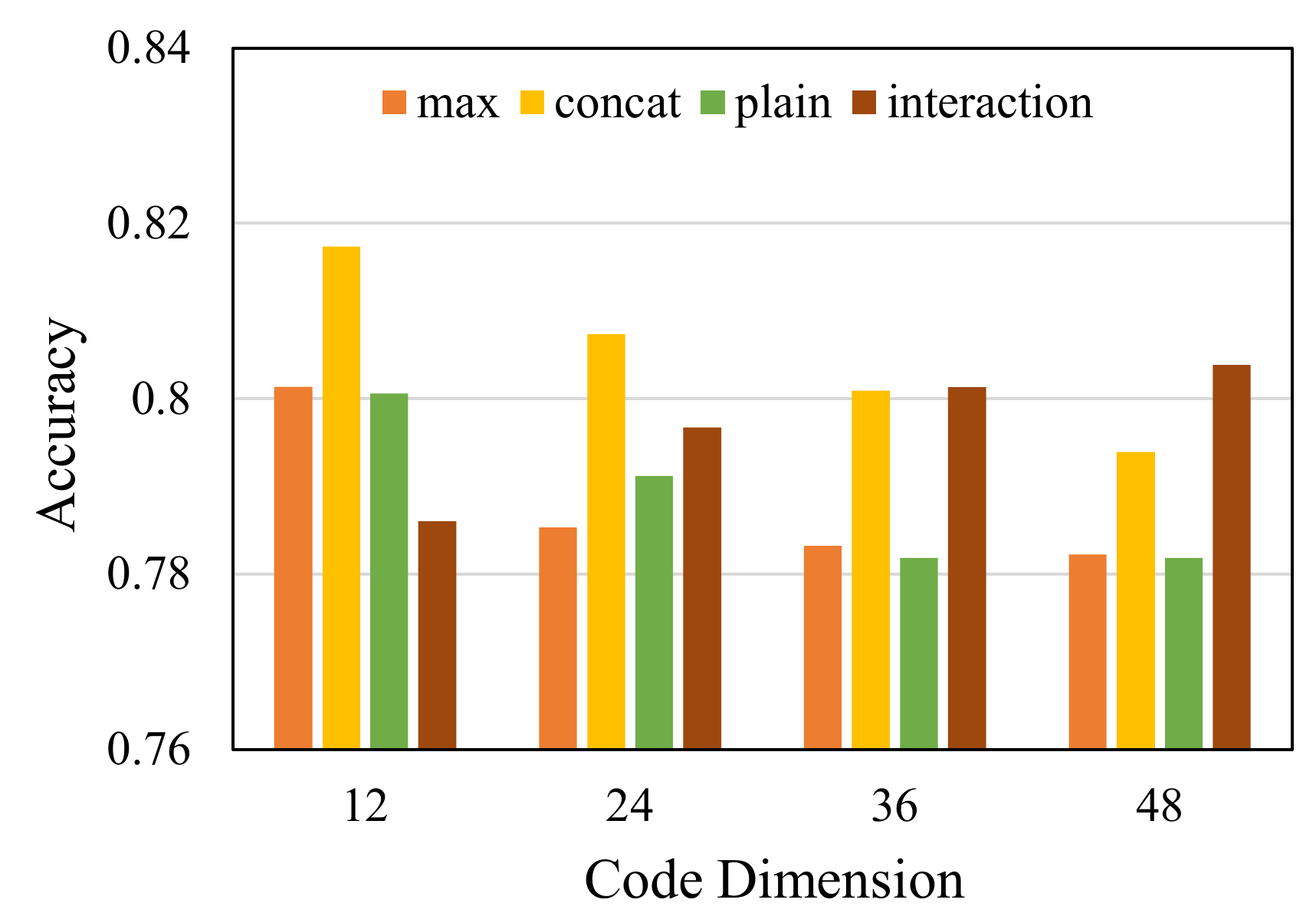}
    \end{minipage}%
    }%
    \subfigure[Dota $(\lambda=0.0,\alpha=0.8)$]{
    \begin{minipage}[t]{0.24\linewidth}
    \centering
    \includegraphics[width=\linewidth]{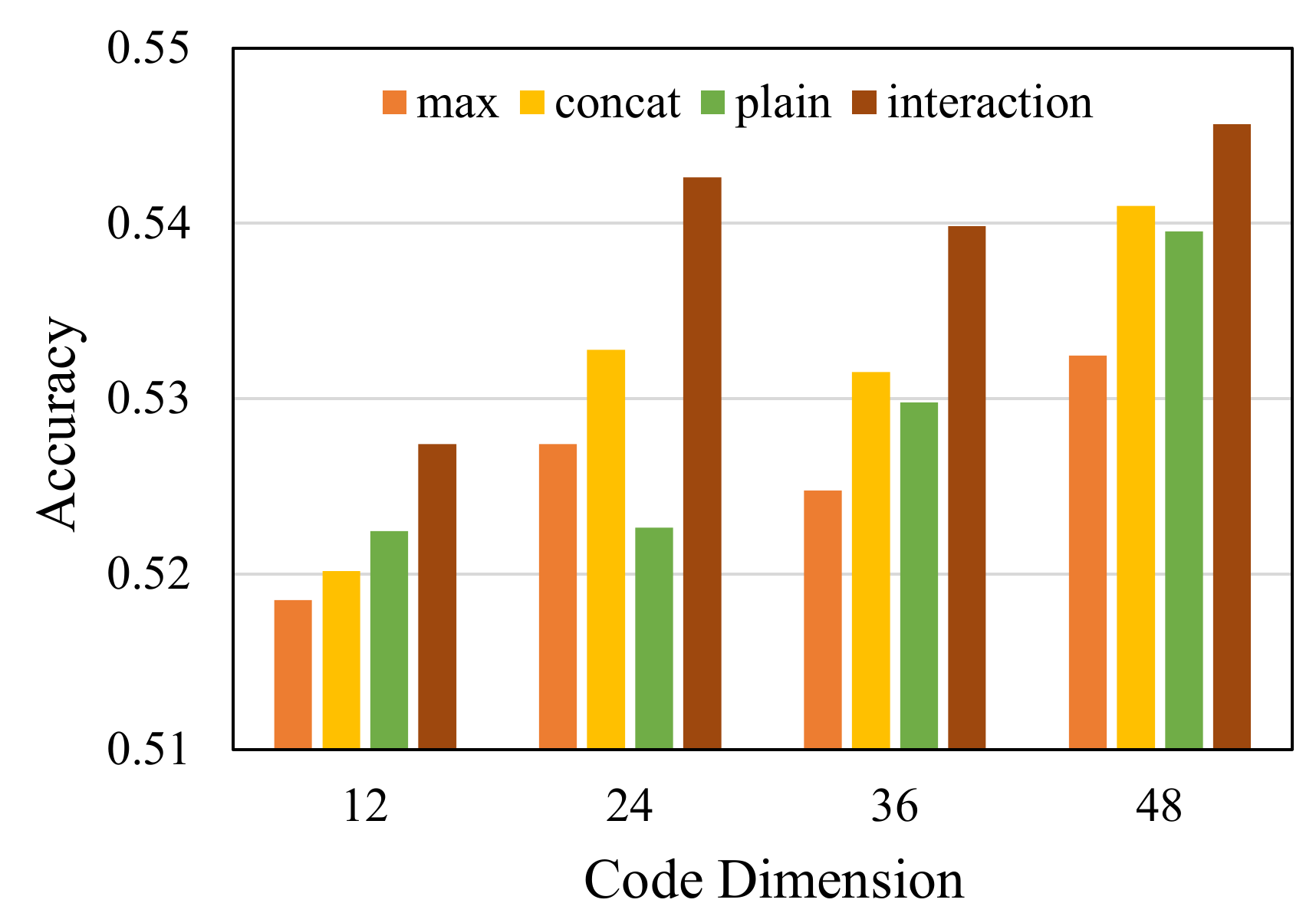}
    \end{minipage}%
    }%
    
	\subfigure[Font $(\lambda=0.2,\alpha=0.8)$]{
    \begin{minipage}[t]{0.24\linewidth}
    \centering
    \includegraphics[width=\linewidth]{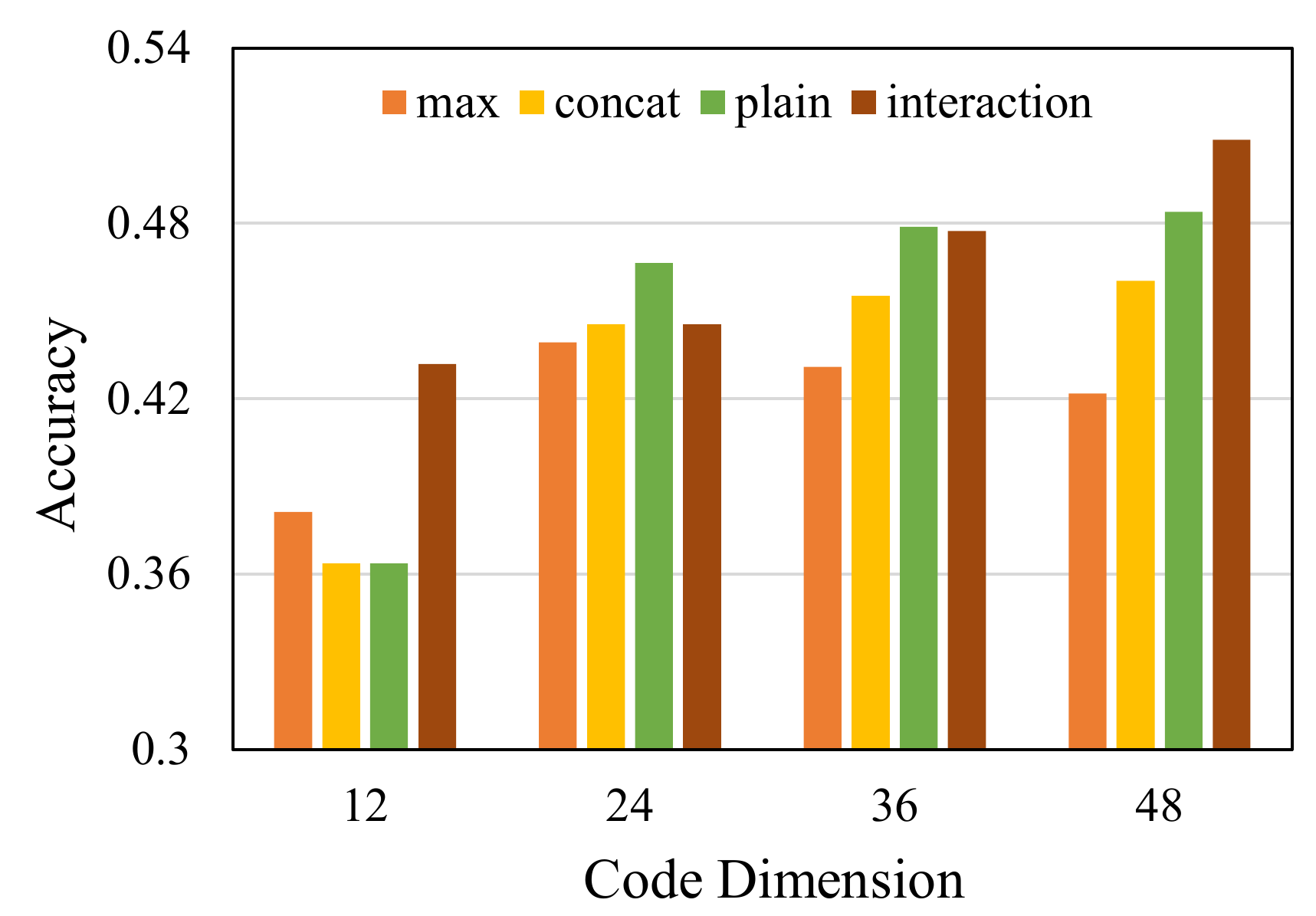}
    \end{minipage}%
    }%
    \subfigure[MT $(\lambda=0.4,\alpha=0.6)$]{
    \begin{minipage}[t]{0.24\linewidth}
    \centering
    \includegraphics[width=\linewidth]{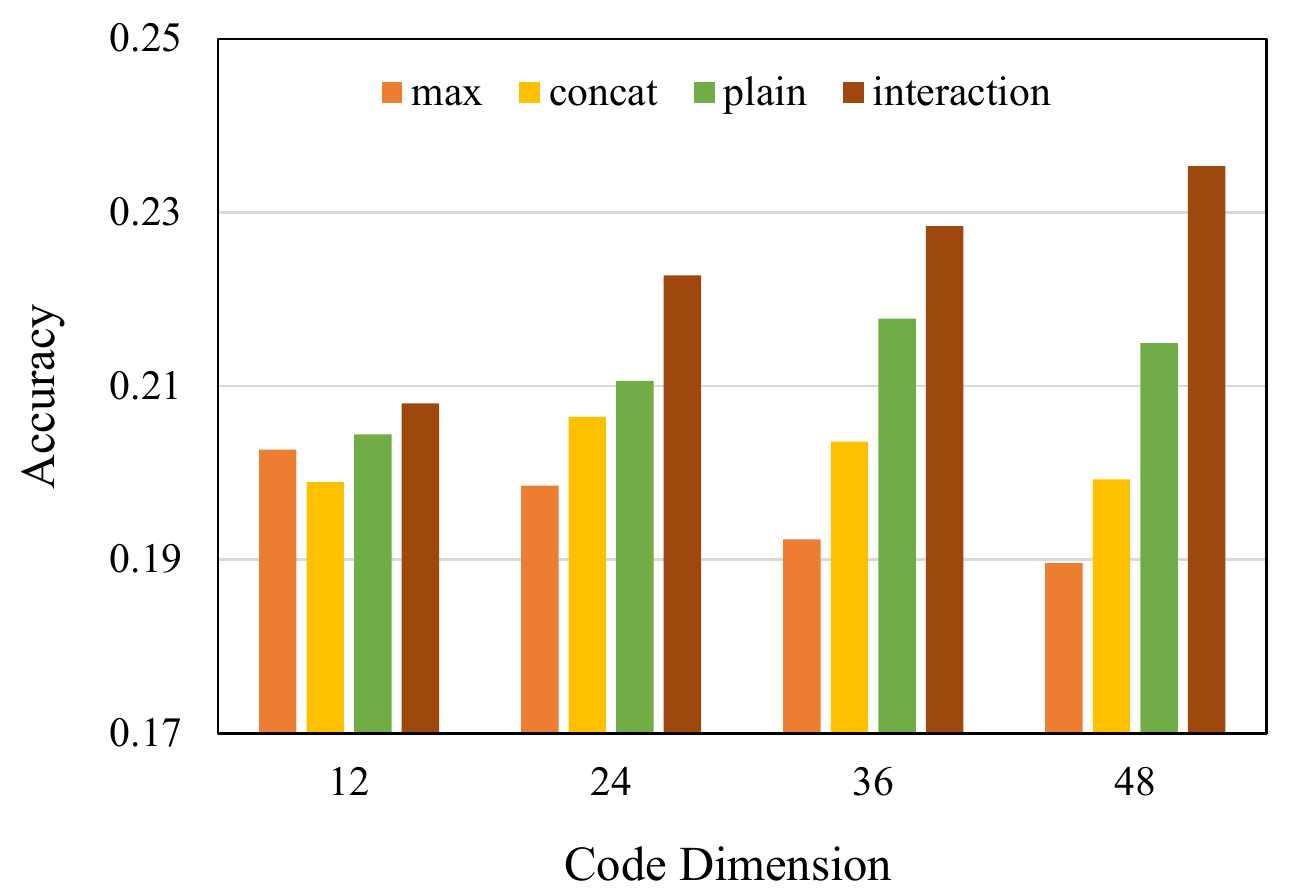}
    \end{minipage}%
    }%
    \subfigure[CT $(\lambda=0.2,\alpha=0.6)$]{
    \begin{minipage}[t]{0.24\linewidth}
    \centering
    \includegraphics[width=\linewidth]{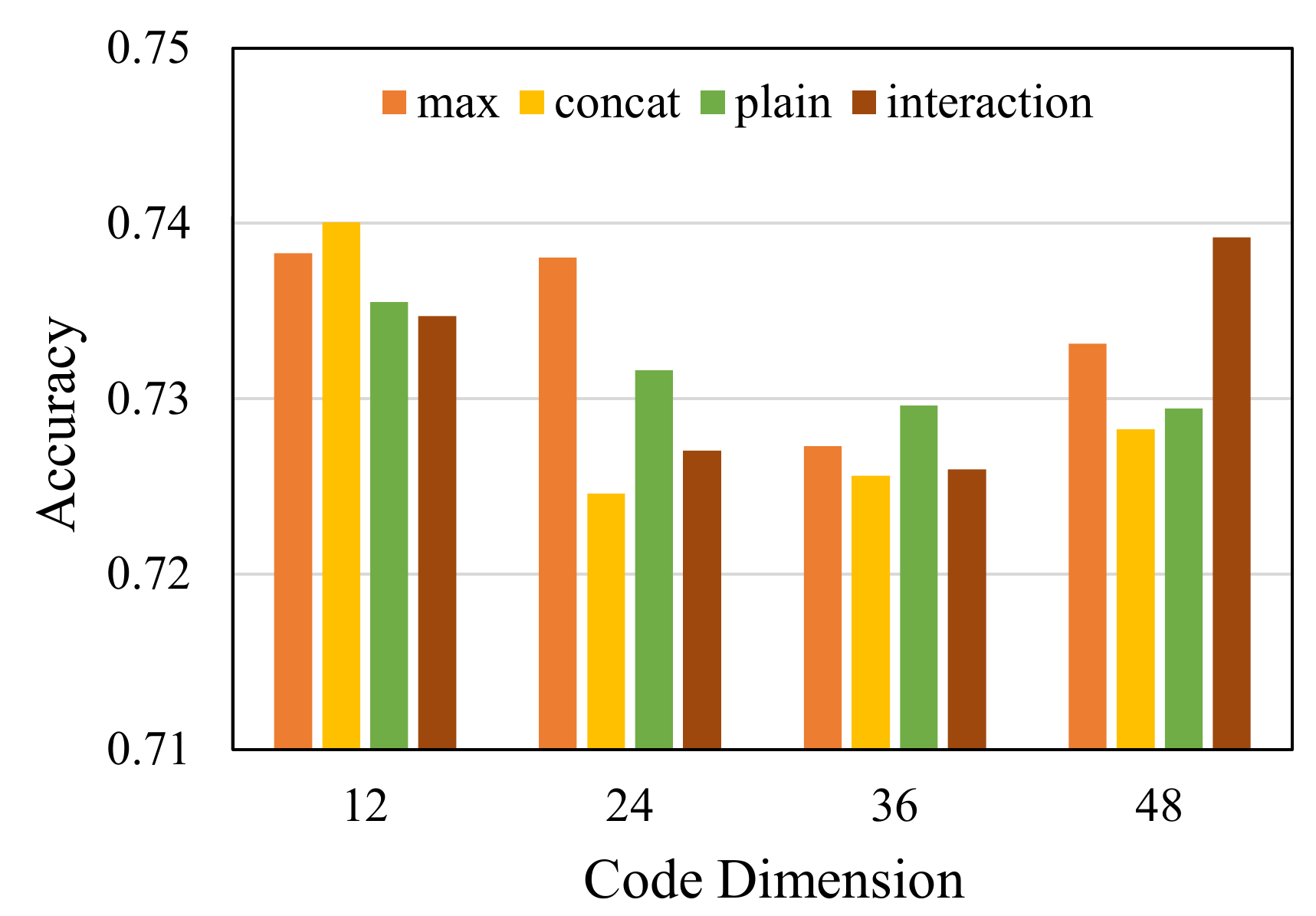}
    \end{minipage}%
    }%
    \subfigure[ML1M $(\lambda=0.2,\alpha=0.6)$]{
    \begin{minipage}[t]{0.24\linewidth}
    \centering
    \includegraphics[width=\linewidth]{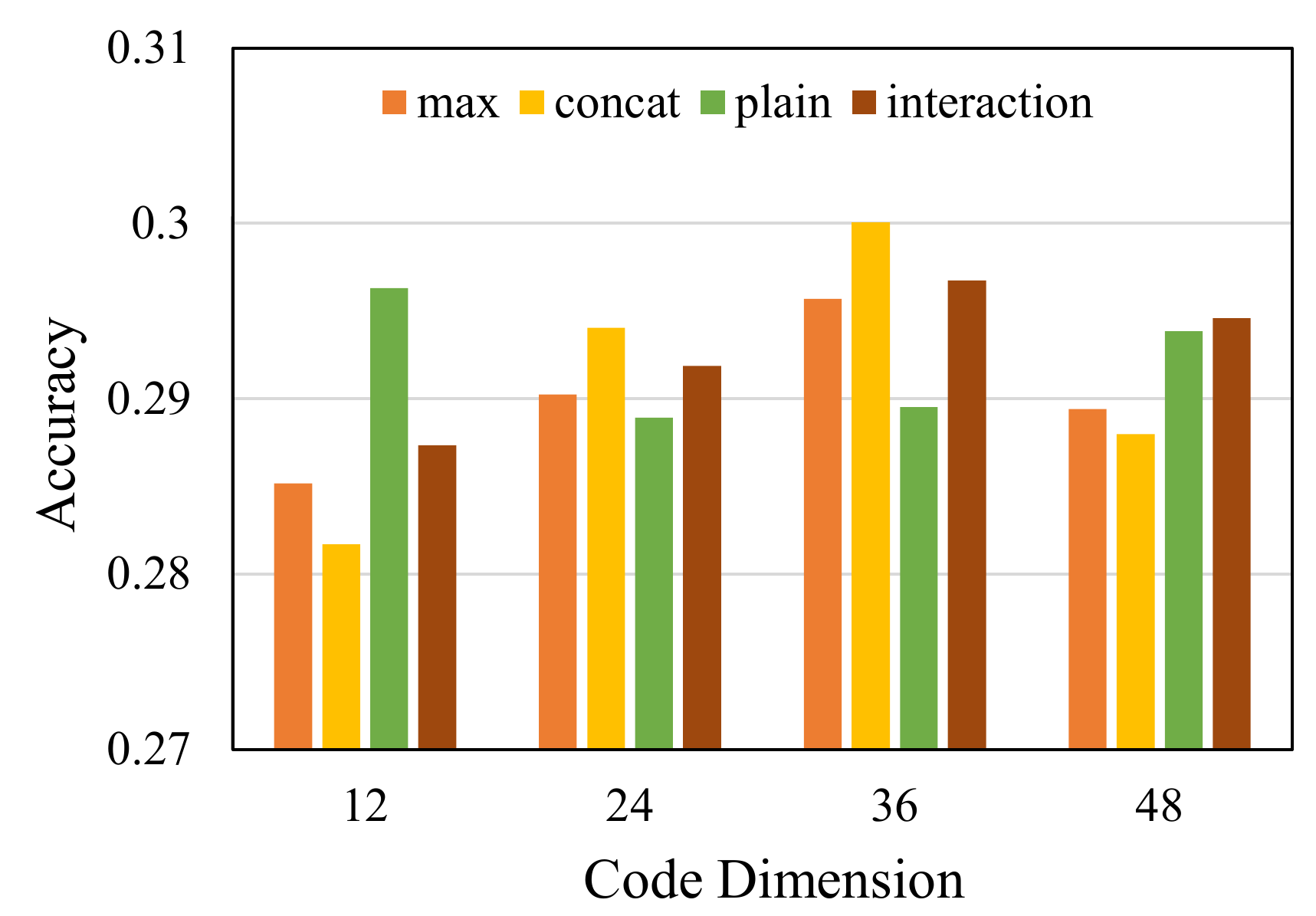}
    \end{minipage}%
    }%
    \caption{Accuracy comparison of \textit{interaction} and different variants under various code dimensions.}
    \label{fig:ablation_results}
\end{figure*}

\subsection{Retrieval Evaluation (Q2)}
In this section, we investigate the retrieval performance of HeCBR to verify the contributions of HeCBR to improving retrieval performance.

\subsubsection{Retrieval Accuracy}
We compare HeCBR with the state-of-the-art baselines in terms of retrieval performance. Specifically, we report the evaluation results of mean average precision (MAP) and precision on the retrieved top-$N$ most similar cases, as shown in Figures \ref{fig:map_results} and \ref{fig:precision_results} respectively. From the results, we observe the following findings:
\begin{itemize}
    \item HeCBR establishes a new state-of-the-art on all datasets except for PT, and it outperforms the baselines in terms of MAP and precision under different numbers of retrieved cases, especially on ADV, ADT, Dota and ML1M. The results show the superiority of HeCBR over the baselines in retrieving similar cases, confirming the better classification accuracy of HeCBR over the baselines in the evaluation of classification performance.
    \item Compared with the baselines, HeCBR shows less fluctuation and a smooth trend in MAP and has a more stable precision along with the increasing case numbers. The results reflect that the introduction of feature embeddings and feature interactions in HeCBR is effective to capture the intrinsic heterogeneity in (high-dimensional) cases and improve the robustness of HeCBR on data of different scales and types.
    \item In addition, data-independent methods, i.e., the LSH-based methods and ITQ, are more vulnerable to the number of retrieved case numbers and generally perform worse than data-dependent hashing methods, for example ITQ and PMH on ADT, FlyH on Dota, and LSH and PMH on MT. The results are attributed that data-independent methods rely on specific distance measurement or transformation (PCA) and hardly learn data specific features to address complex data issues, e.g., heterogeneity and high dimensionality.
\end{itemize}

From the above results, we conclude that HeCBR effectively captures feature heterogeneity and interactions to improve case representation and performs more accurate similar case retrieval.

\begin{figure*}[!t]
    \centering
    \subfigure[ADV]{
    \begin{minipage}[t]{0.24\linewidth}
    \centering
    \includegraphics[width=\linewidth,trim={3cm 3cm 3cm 3cm},keepaspectratio]{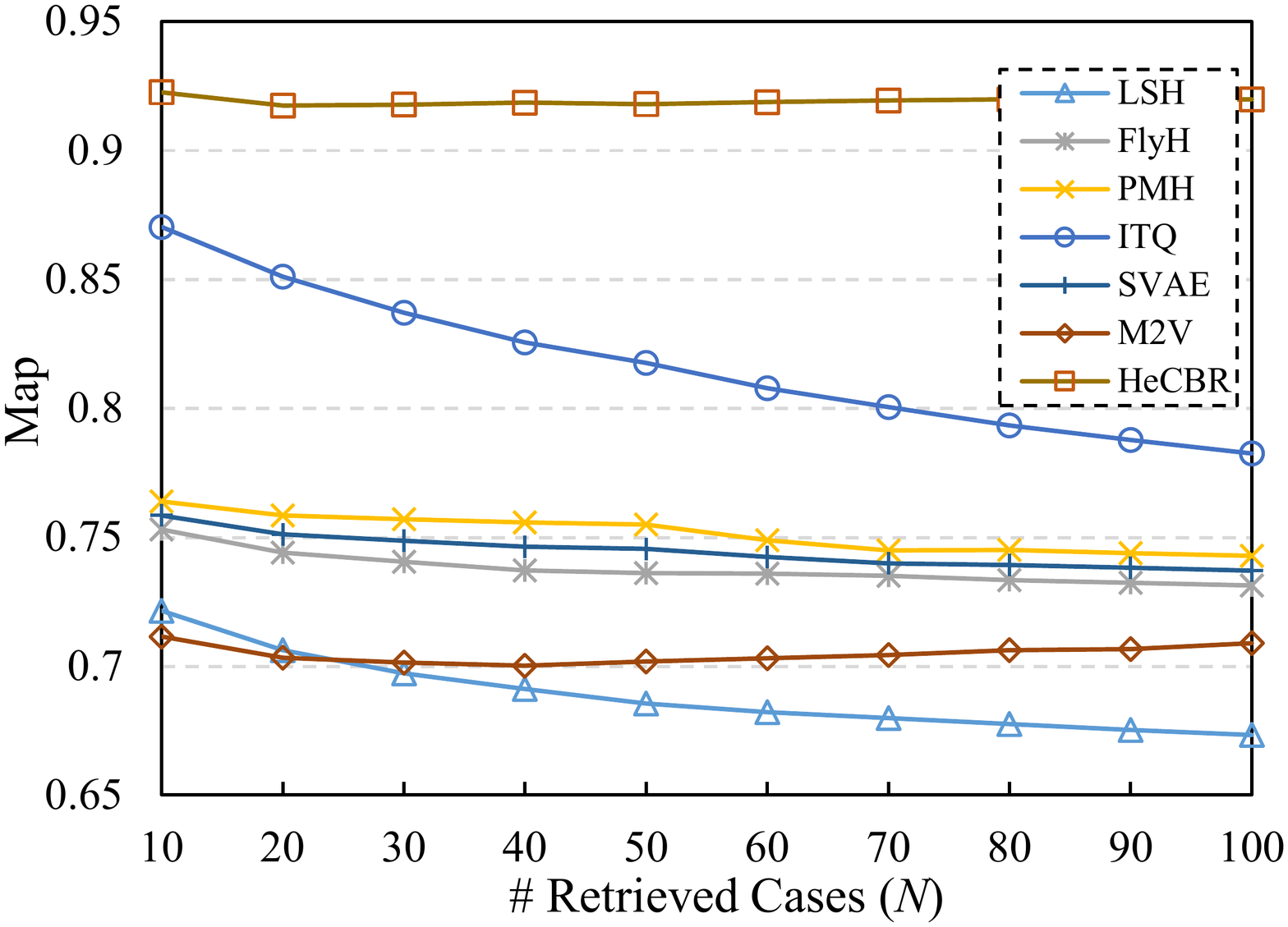}
    \end{minipage}%
    }%
	 \subfigure[PT]{
    \begin{minipage}[t]{0.24\linewidth}
    \centering
    \includegraphics[width=\linewidth,trim={3cm 3cm 3cm 3cm},keepaspectratio]{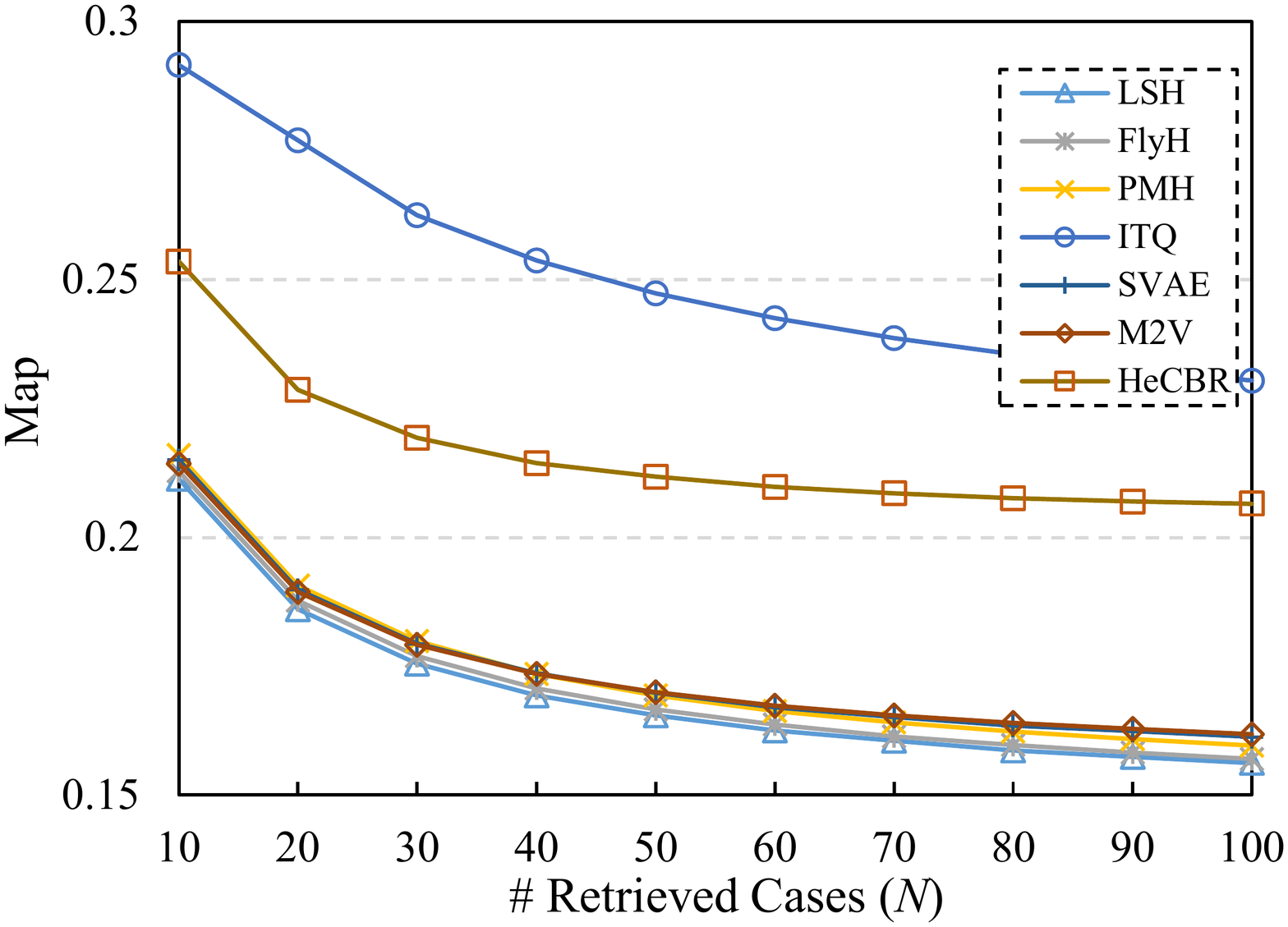}
    \end{minipage}%
    }%
	 \subfigure[ADT]{
    \begin{minipage}[t]{0.24\linewidth}
    \centering
    \includegraphics[width=\linewidth,trim={3cm 3cm 3cm 3cm},keepaspectratio]{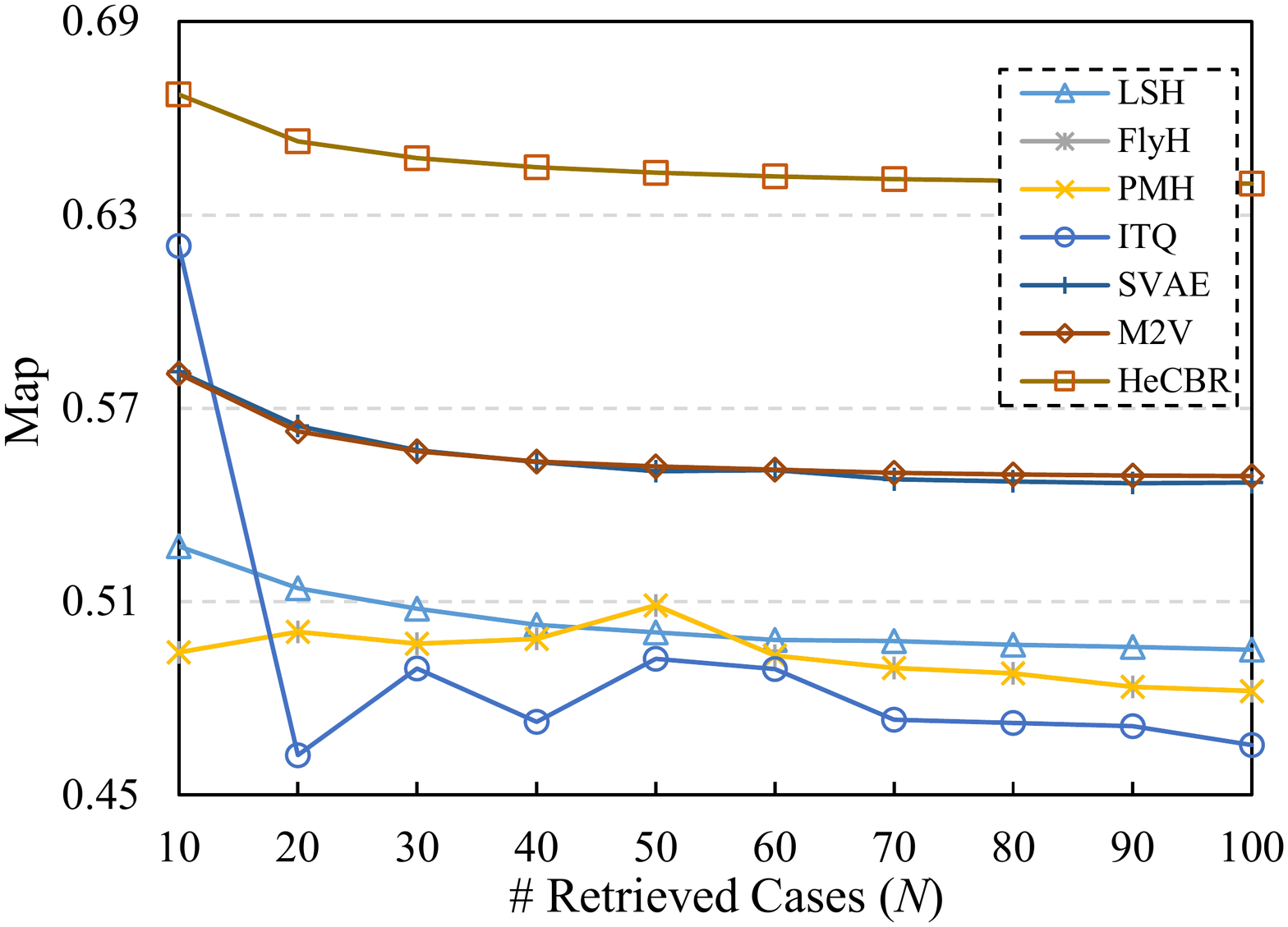}
    \end{minipage}%
    }%
    \subfigure[Dota]{
    \begin{minipage}[t]{0.24\linewidth}
    \centering
    \includegraphics[width=\linewidth,trim={3cm 3cm 3cm 3cm},keepaspectratio]{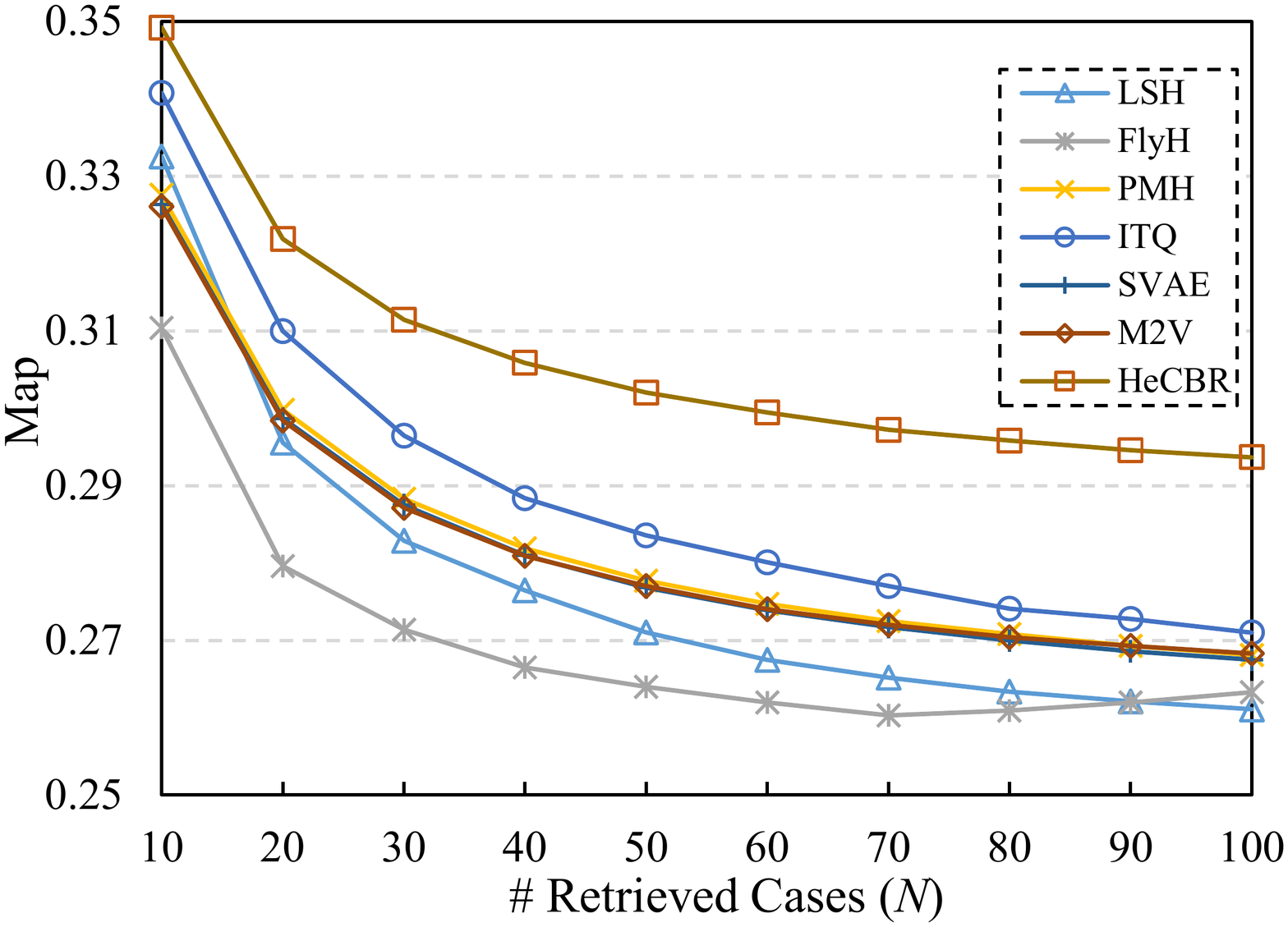}
    \end{minipage}%
    }%
    
	 \subfigure[Font]{
    \begin{minipage}[t]{0.24\linewidth}
    \centering
    \includegraphics[width=\linewidth,trim={3cm 3cm 3cm 3cm},keepaspectratio]{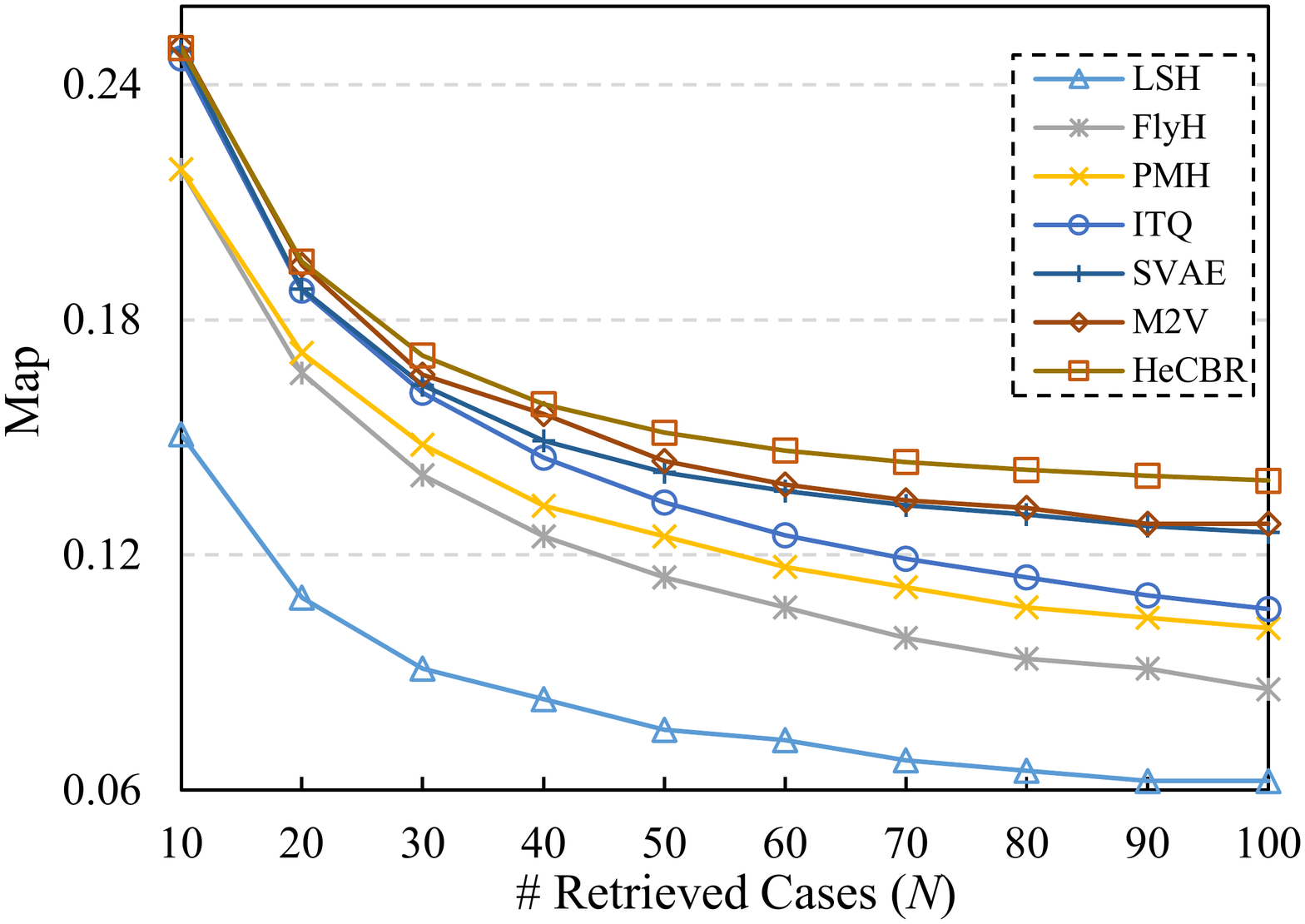}
    \end{minipage}%
    }%
    \subfigure[MT]{
    \begin{minipage}[t]{0.24\linewidth}
    \centering
    \includegraphics[width=\linewidth,trim={3cm 3cm 3cm 3cm},keepaspectratio]{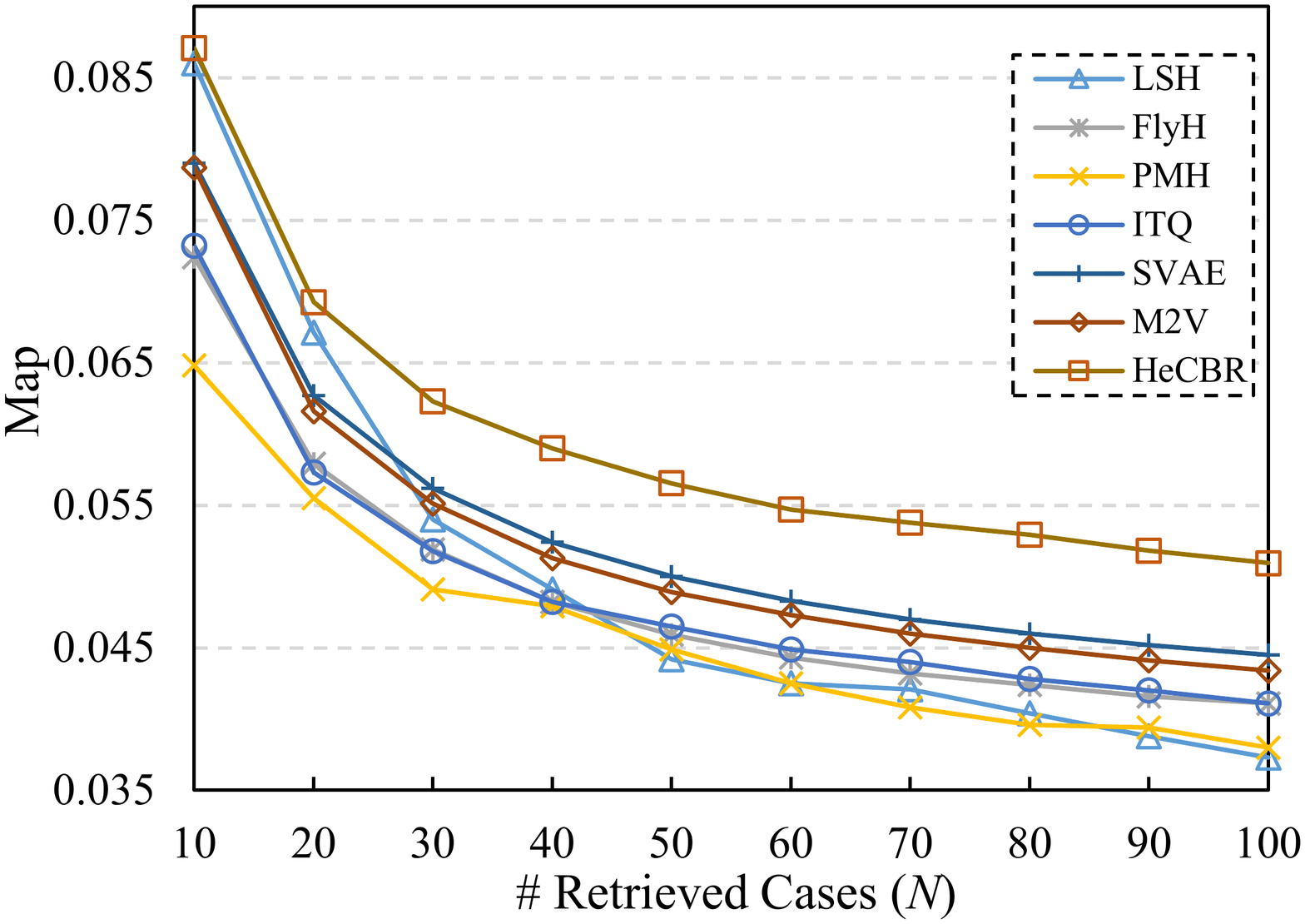}
    \end{minipage}%
    }%
	 \subfigure[CT]{
    \begin{minipage}[t]{0.24\linewidth}
    \centering
    \includegraphics[width=\linewidth,trim={3cm 3cm 3cm 3cm},keepaspectratio]{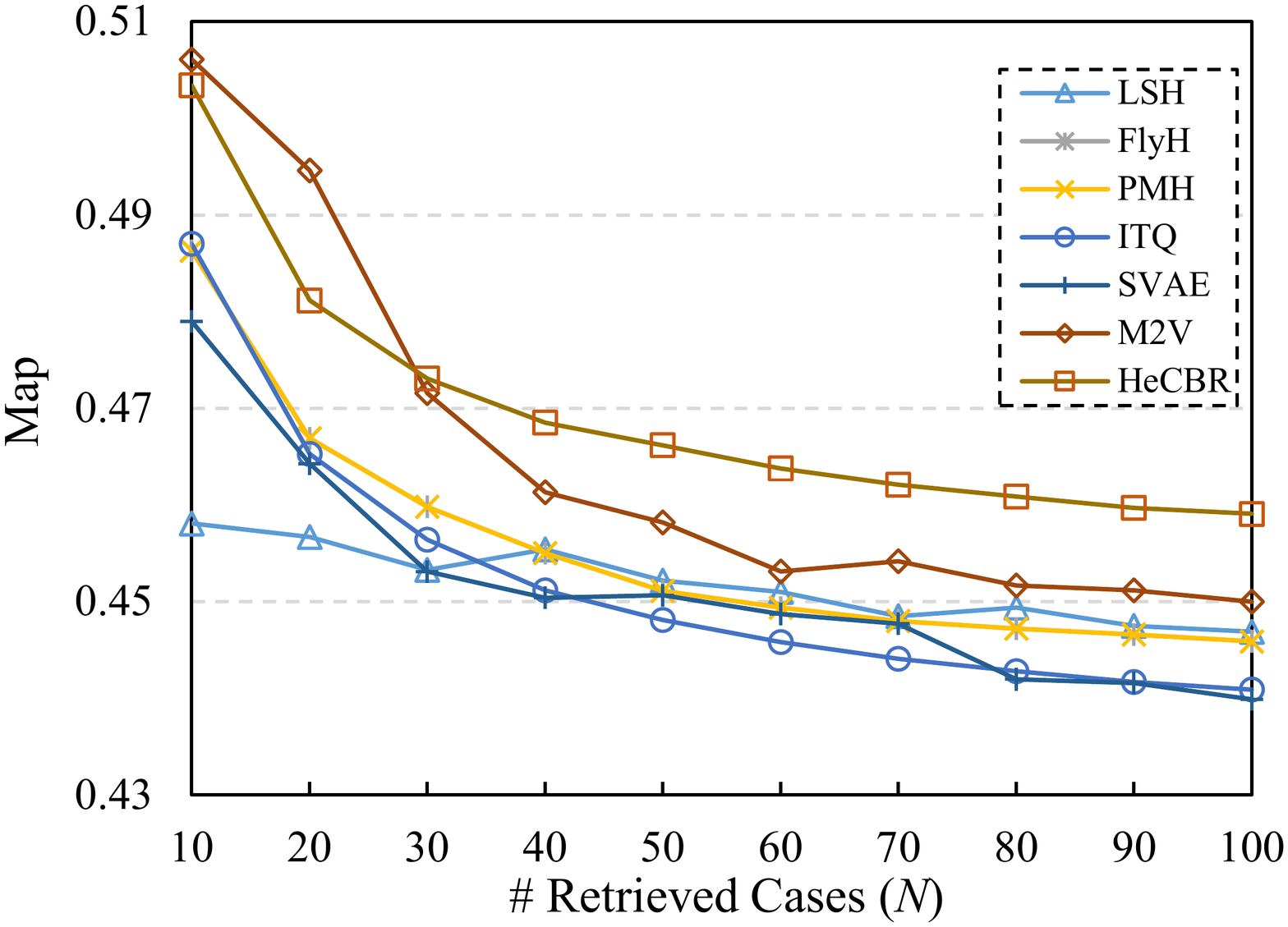}
    \end{minipage}%
    }%
    \subfigure[ML1M]{
    \begin{minipage}[t]{0.24\linewidth}
    \centering
    \includegraphics[width=\linewidth,trim={3cm 3cm 3cm 3cm},keepaspectratio]{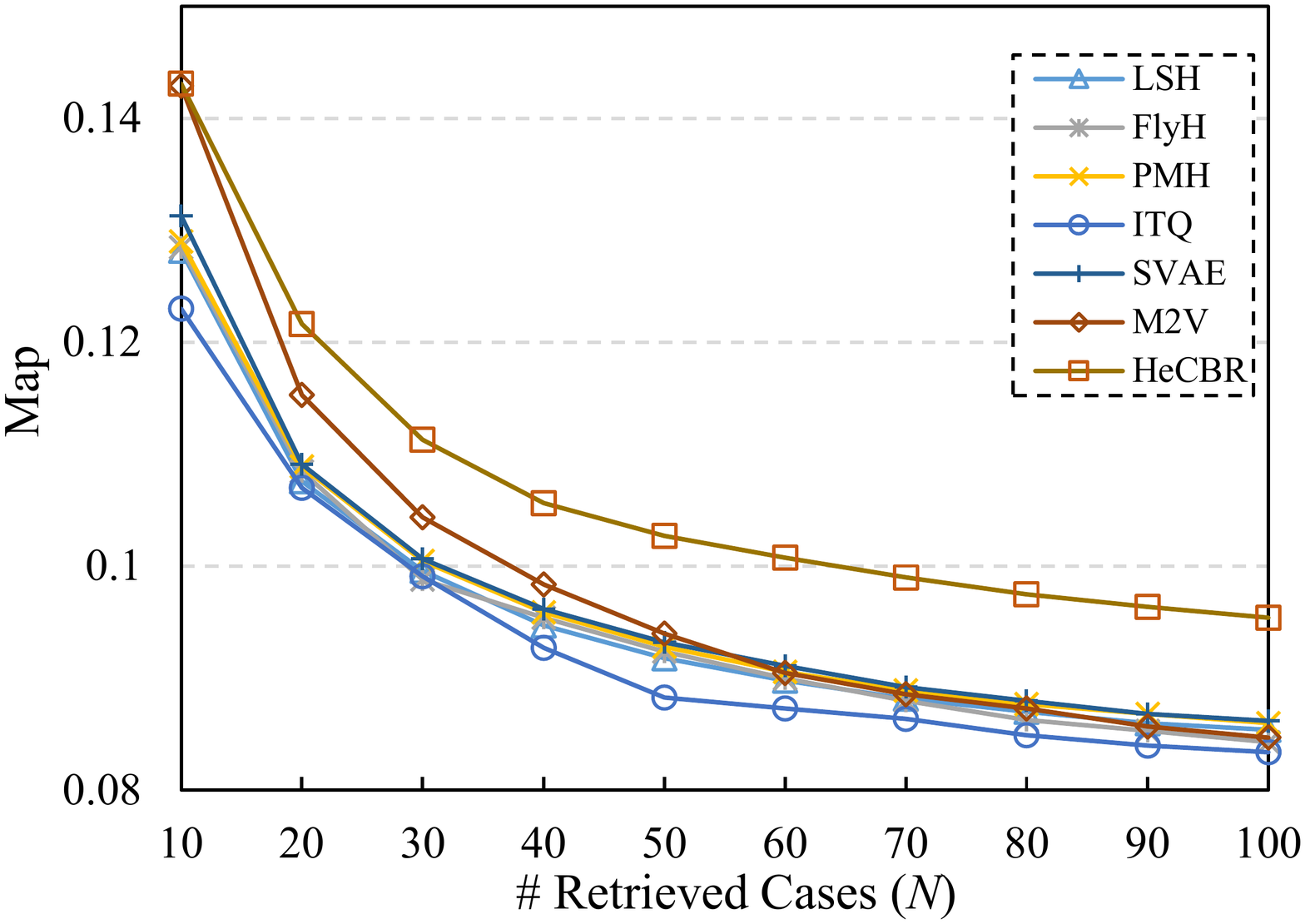}
    \end{minipage}%
    }%
    \caption{Comparison of retrieval performance in terms of mean average precision with different numbers of retrieved cases (MAP@N). Six representative baselines are compared with HeCBR under $36$-bit binary codes.}
    \label{fig:map_results}
\end{figure*}

\begin{figure*}[!t]
    \centering
    \subfigure[ADV]{
    \label{fig:hete}
    \begin{minipage}[t]{0.24\linewidth}
    \centering
    \includegraphics[width=\linewidth,trim={3cm 3cm 3cm 3cm},keepaspectratio]{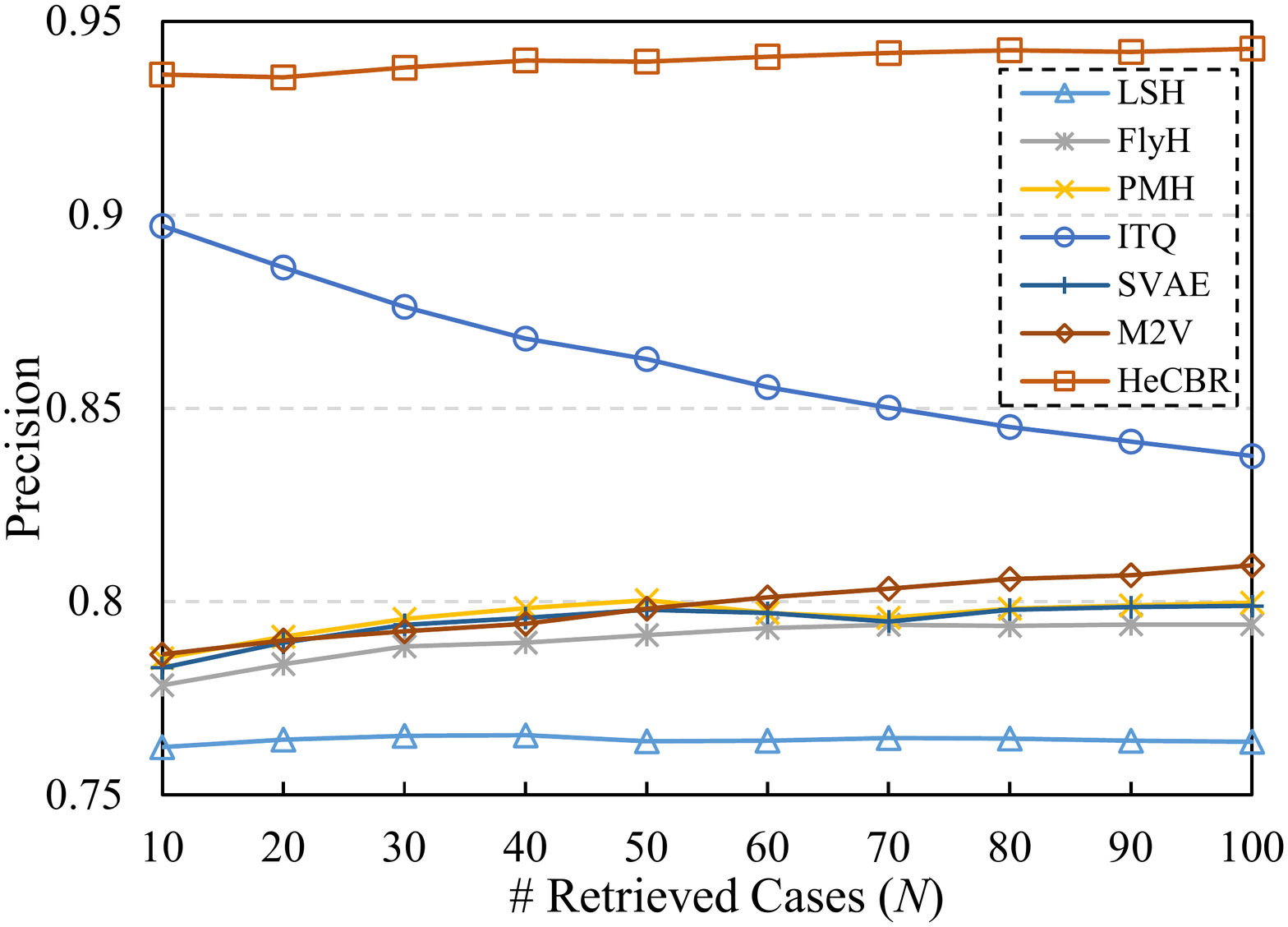}
    \end{minipage}%
    }%
	 \subfigure[PT]{
    \begin{minipage}[t]{0.24\linewidth}
    \centering
    \includegraphics[width=\linewidth,trim={3cm 3cm 3cm 3cm},keepaspectratio]{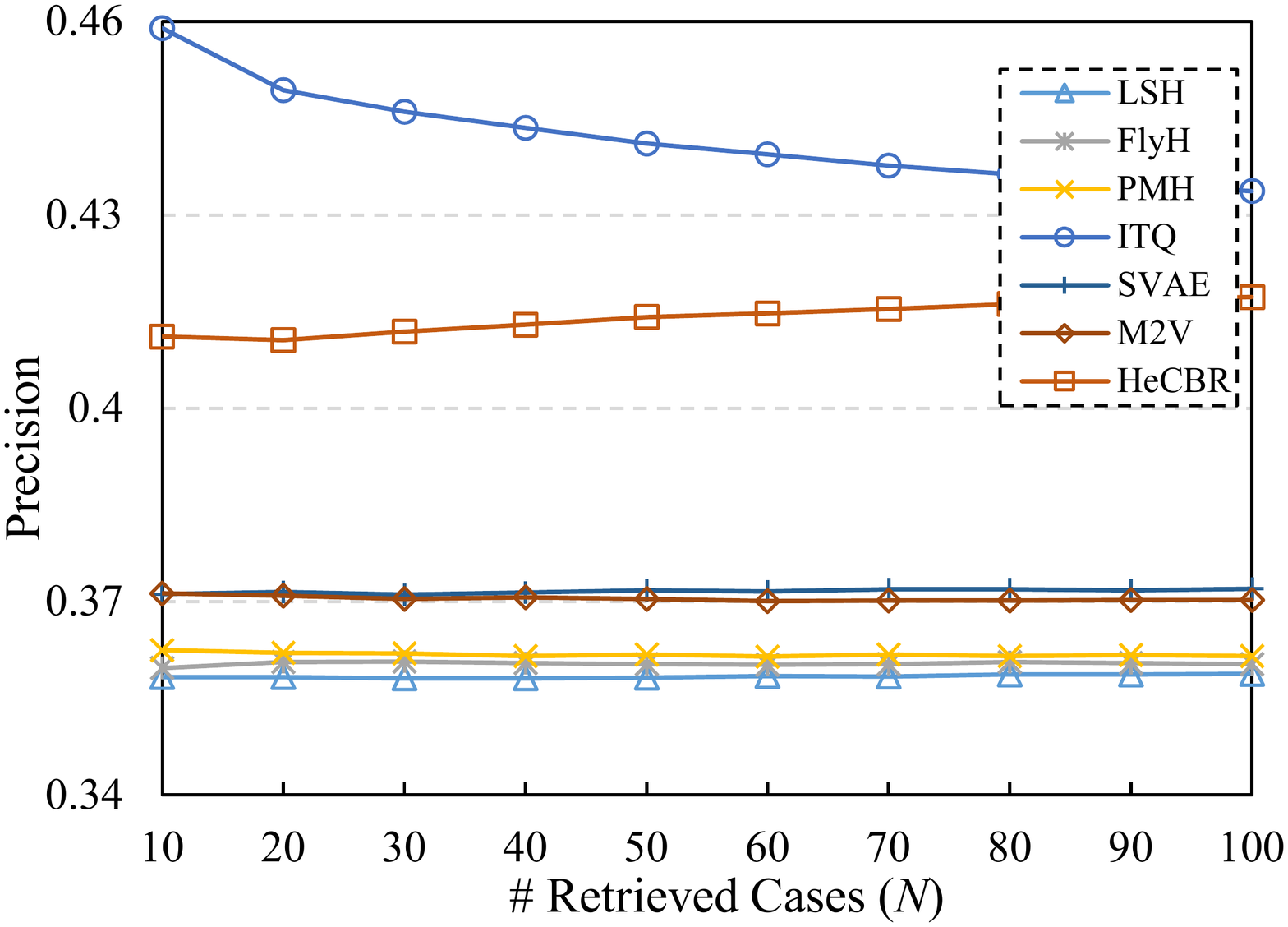}
    \end{minipage}%
    }%
	 \subfigure[ADT]{
    \begin{minipage}[t]{0.24\linewidth}
    \centering
    \includegraphics[width=\linewidth,trim={3cm 3cm 3cm 3cm},keepaspectratio]{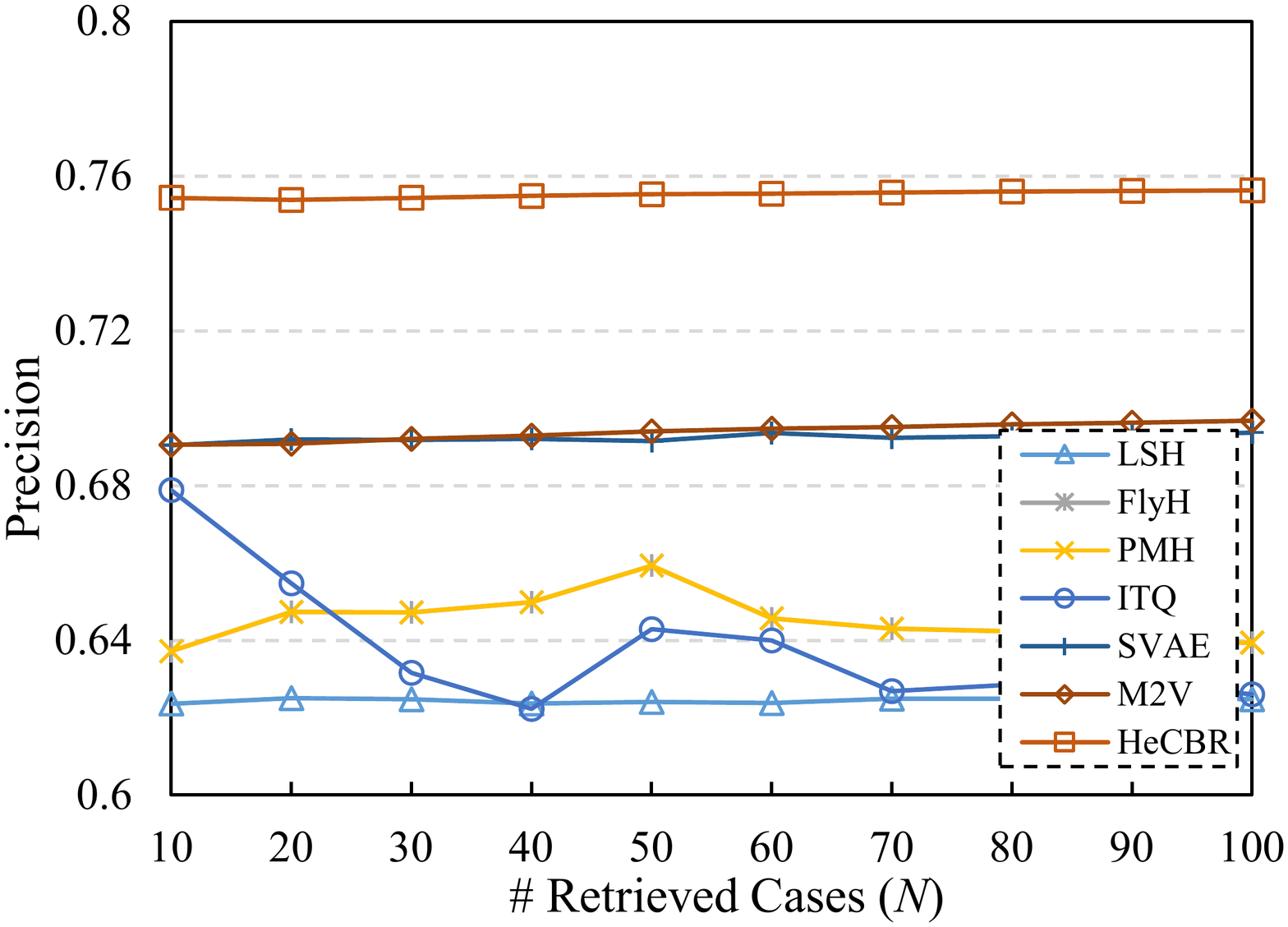}
    \end{minipage}%
    }%
    \subfigure[Dota]{
    \begin{minipage}[t]{0.24\linewidth}
    \centering
    \includegraphics[width=\linewidth,trim={3cm 3cm 3cm 3cm},keepaspectratio]{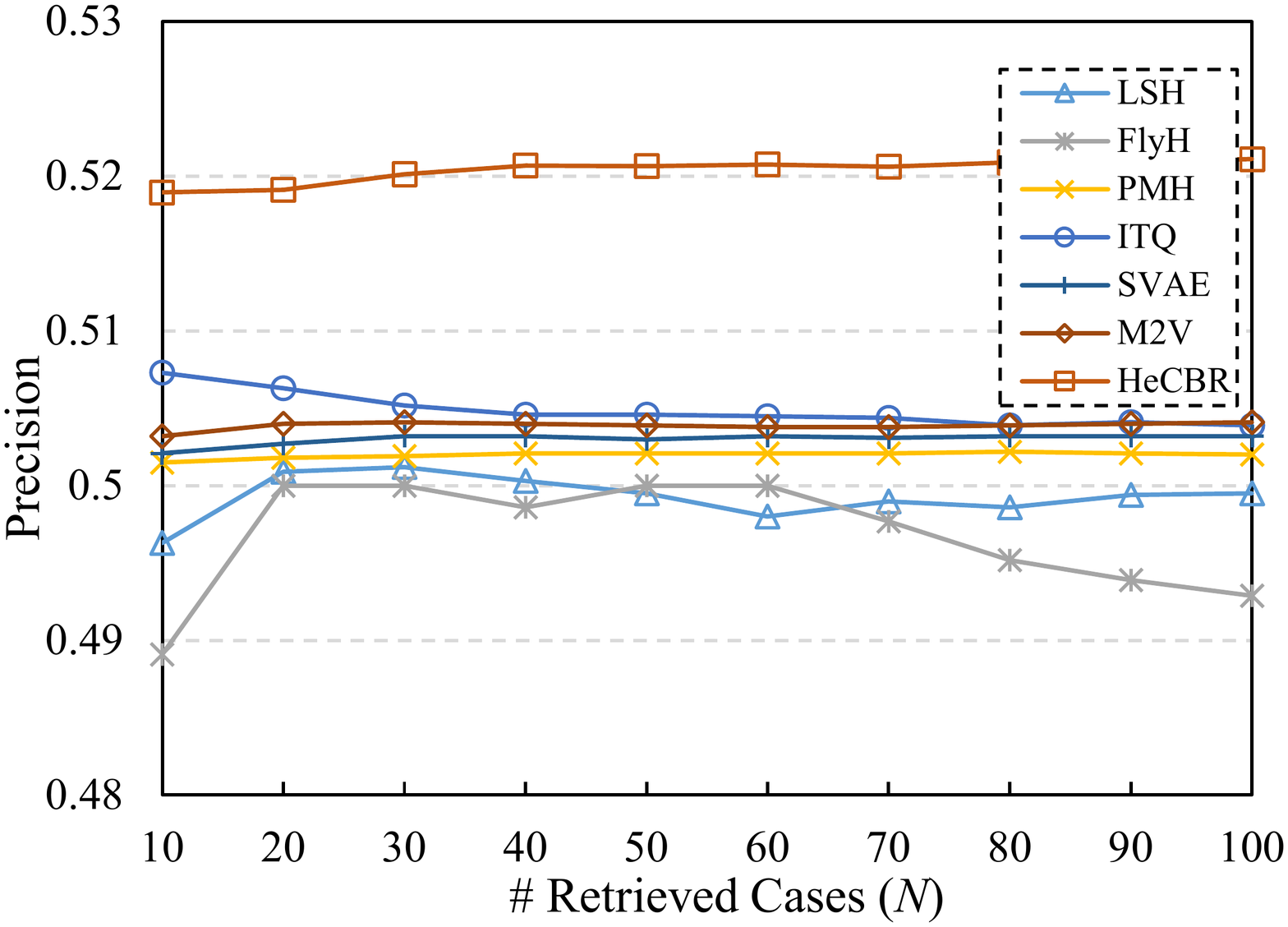}
    \end{minipage}%
    }%
    
	 \subfigure[Font]{
    \begin{minipage}[t]{0.24\linewidth}
    \centering
    \includegraphics[width=\linewidth,trim={3cm 3cm 3cm 3cm},keepaspectratio]{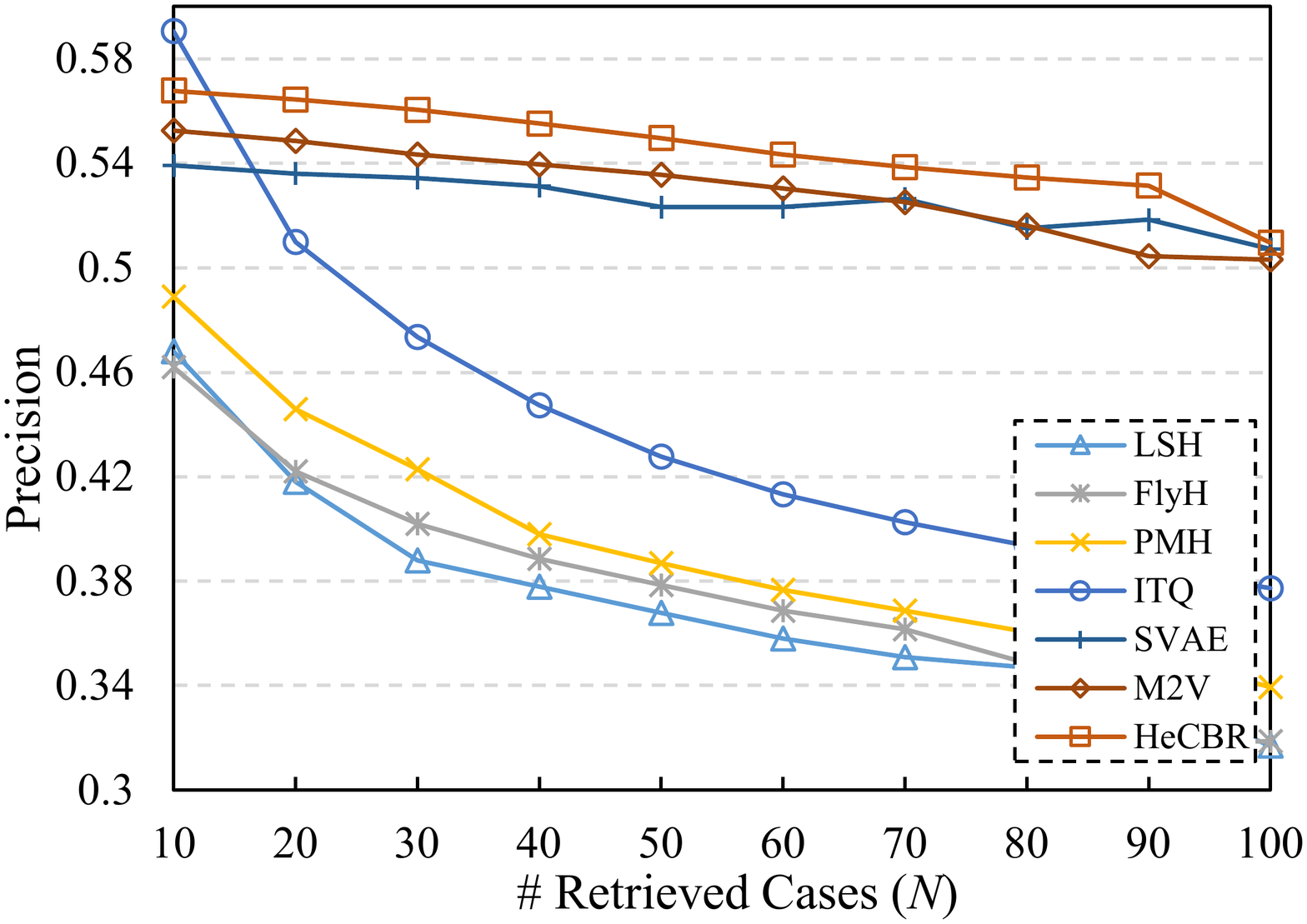}
    \end{minipage}%
    }%
    \subfigure[MT]{
    \begin{minipage}[t]{0.24\linewidth}
    \centering
    \includegraphics[width=\linewidth,trim={3cm 3cm 3cm 3cm},keepaspectratio]{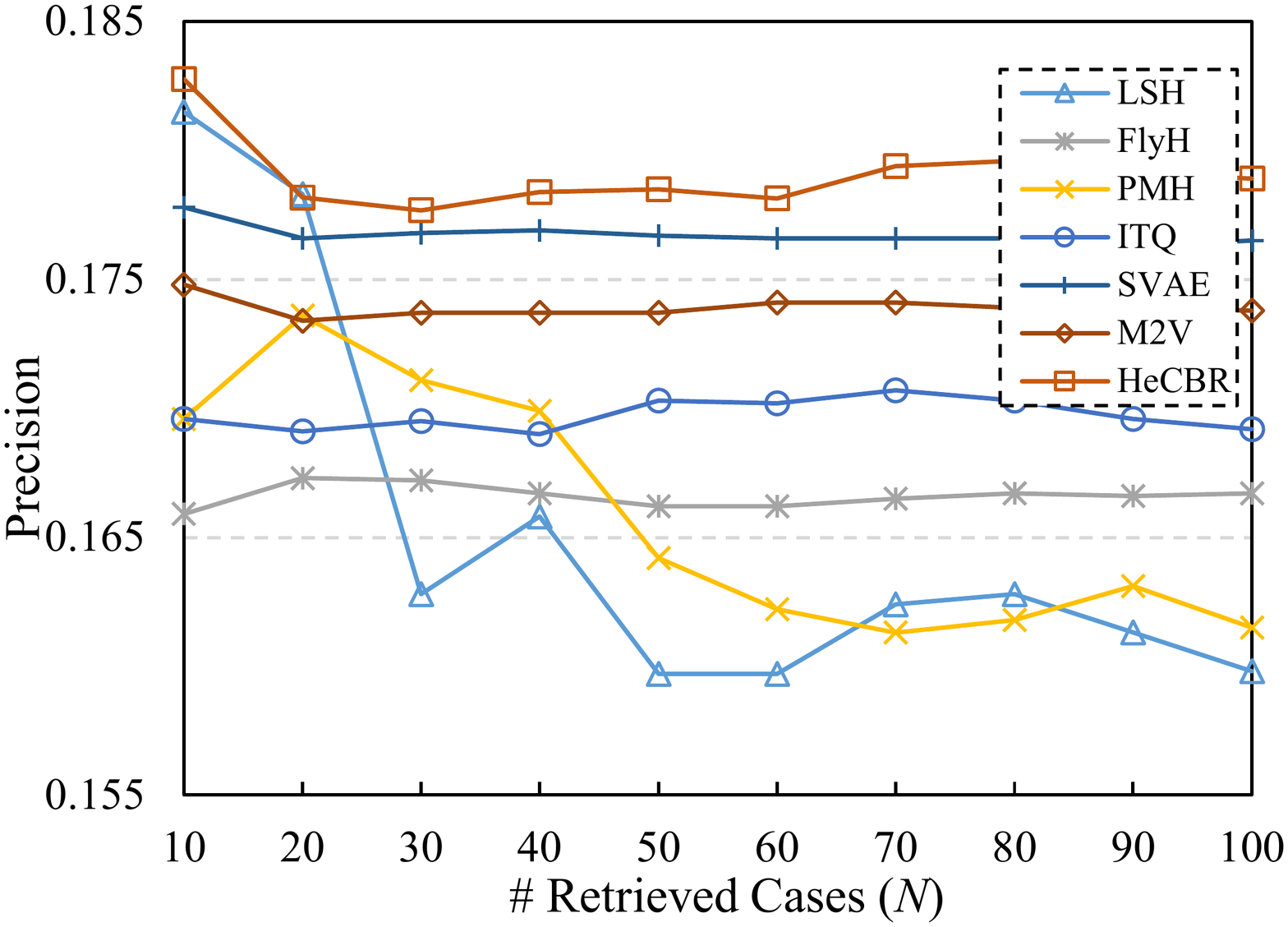}
    \end{minipage}%
    }%
	 \subfigure[CT]{
    \begin{minipage}[t]{0.24\linewidth}
    \centering
    \includegraphics[width=\linewidth,trim={3cm 3cm 3cm 3cm},keepaspectratio]{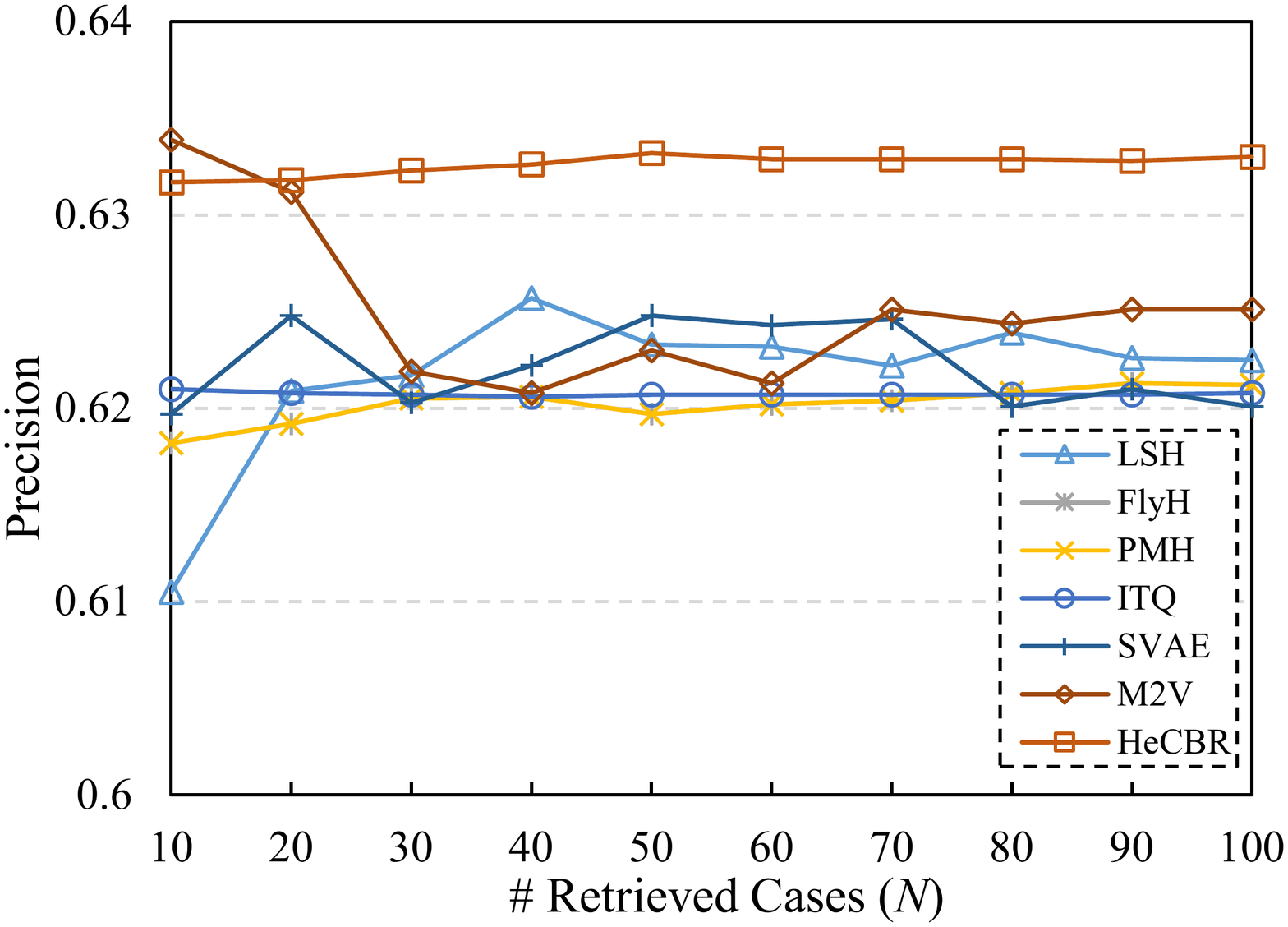}
    \end{minipage}%
    }%
    \subfigure[ML1M]{
    \begin{minipage}[t]{0.24\linewidth}
    \centering
    \includegraphics[width=\linewidth,trim={3cm 3cm 3cm 3cm},keepaspectratio]{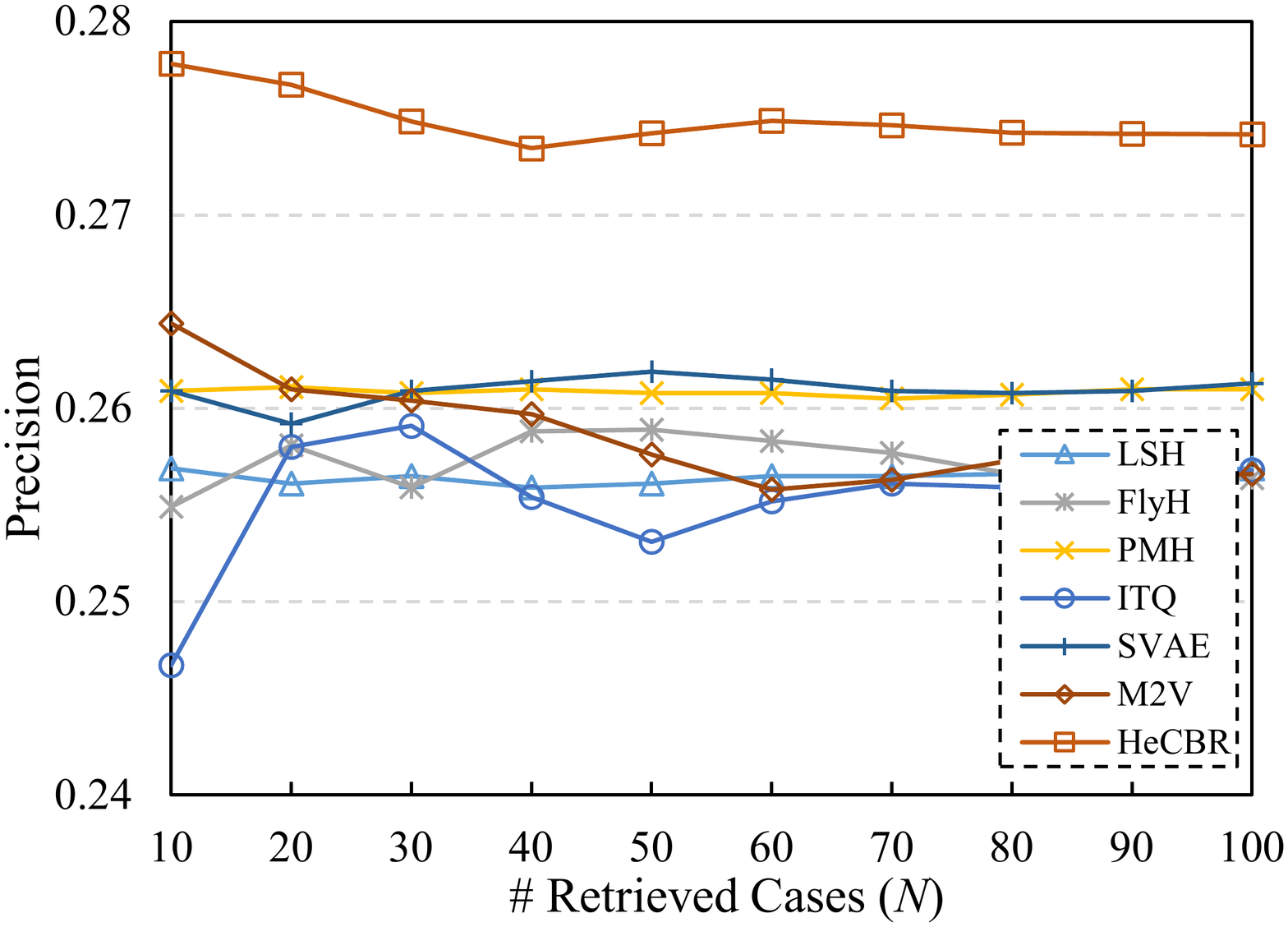}
    \end{minipage}%
    }%
    \caption{Comparison of retrieval performance in terms of precision with different numbers of retrieved cases N (Precision@N). Six representative baselines are compared with HeCBR under $36$-bit binary codes.}
    \label{fig:precision_results}
\end{figure*}

\begin{figure*}[!t]
    \centering
    \subfigure[ADV]{
    \begin{minipage}[t]{0.29\linewidth}
    \centering
    \includegraphics[width=0.95\linewidth]{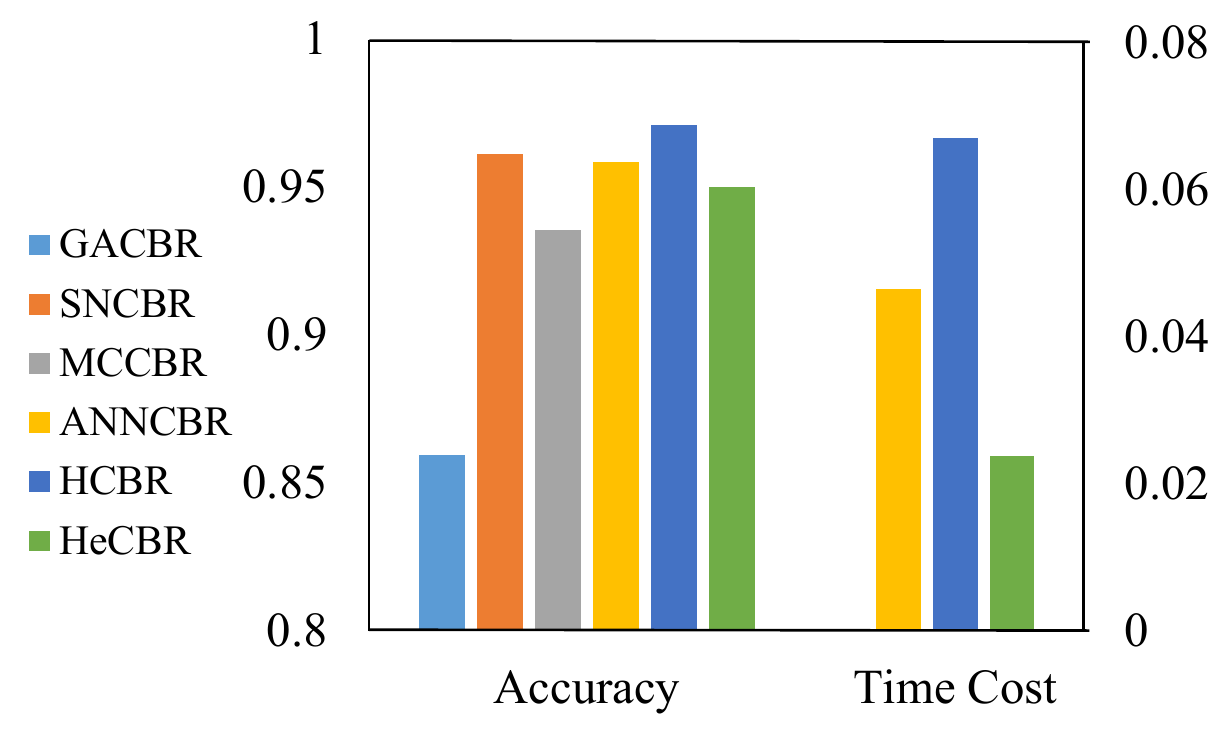}
    \end{minipage}%
    }%
	 \subfigure[PT]{
    \begin{minipage}[t]{0.236\linewidth}
    \centering
    \includegraphics[width=0.95\linewidth]{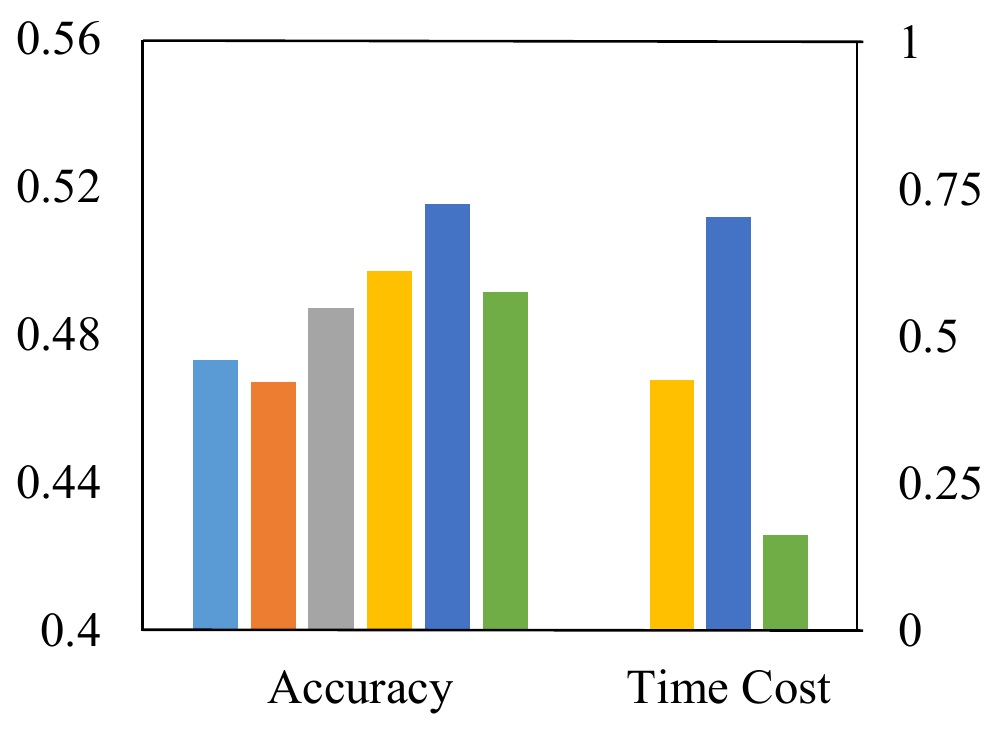}
    \end{minipage}%
    }%
    \subfigure[ADT]{
    \begin{minipage}[t]{0.236\linewidth}
    \centering
    \includegraphics[width=0.95\linewidth]{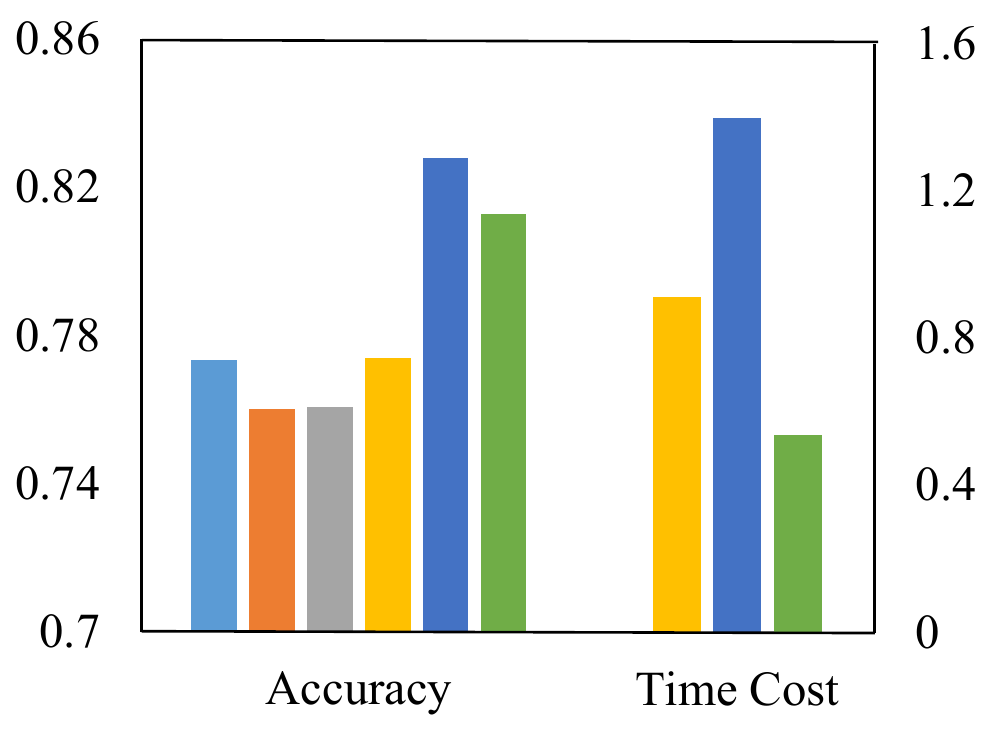}
    \end{minipage}%
    }%
    \subfigure[Dota]{
    \begin{minipage}[t]{0.236\linewidth}
    \centering
    \includegraphics[width=0.95\linewidth]{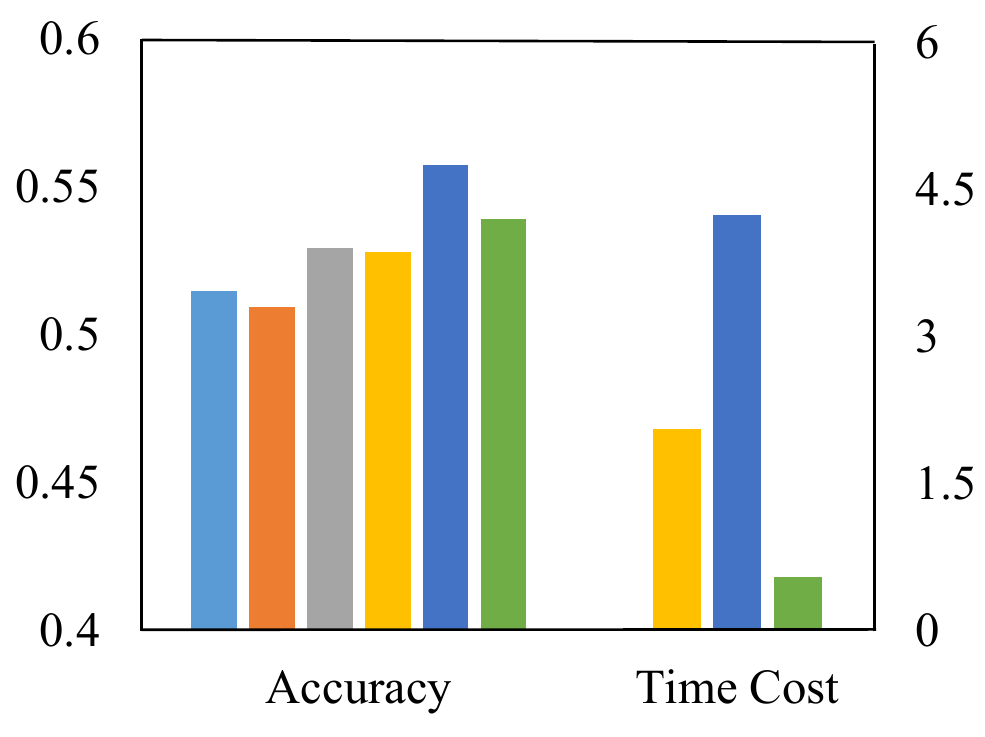}
    \end{minipage}%
    }%
    
	\subfigure[Font]{
    \begin{minipage}[t]{0.29\linewidth}
    \centering
    \includegraphics[width=0.95\linewidth]{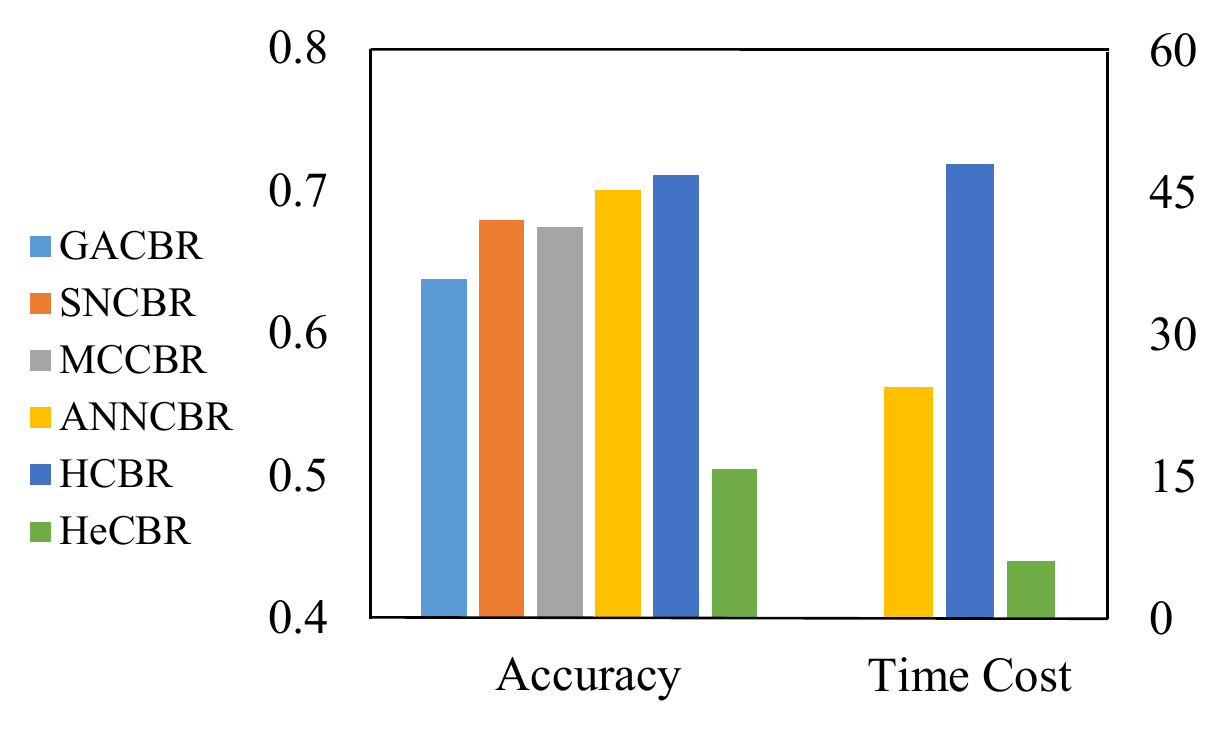}
    \end{minipage}%
    }%
    \subfigure[MT]{
    \begin{minipage}[t]{0.236\linewidth}
    \centering
    \includegraphics[width=0.95\linewidth]{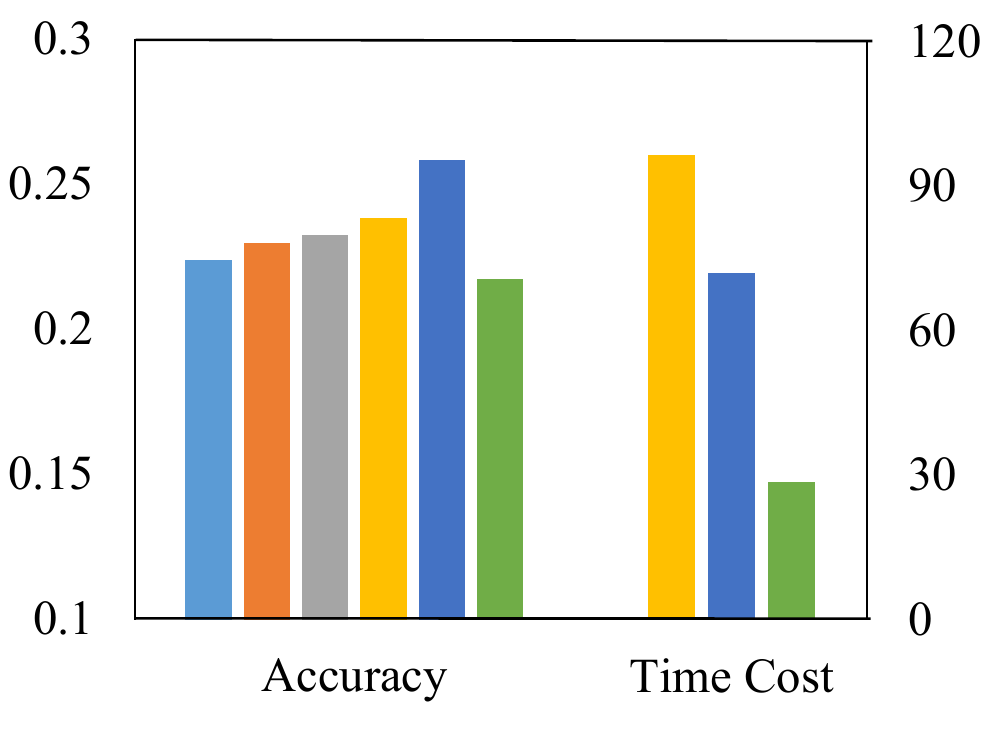}
    \end{minipage}%
    }%
    \subfigure[CT]{
    \begin{minipage}[t]{0.236\linewidth}
    \centering
    \includegraphics[width=0.95\linewidth]{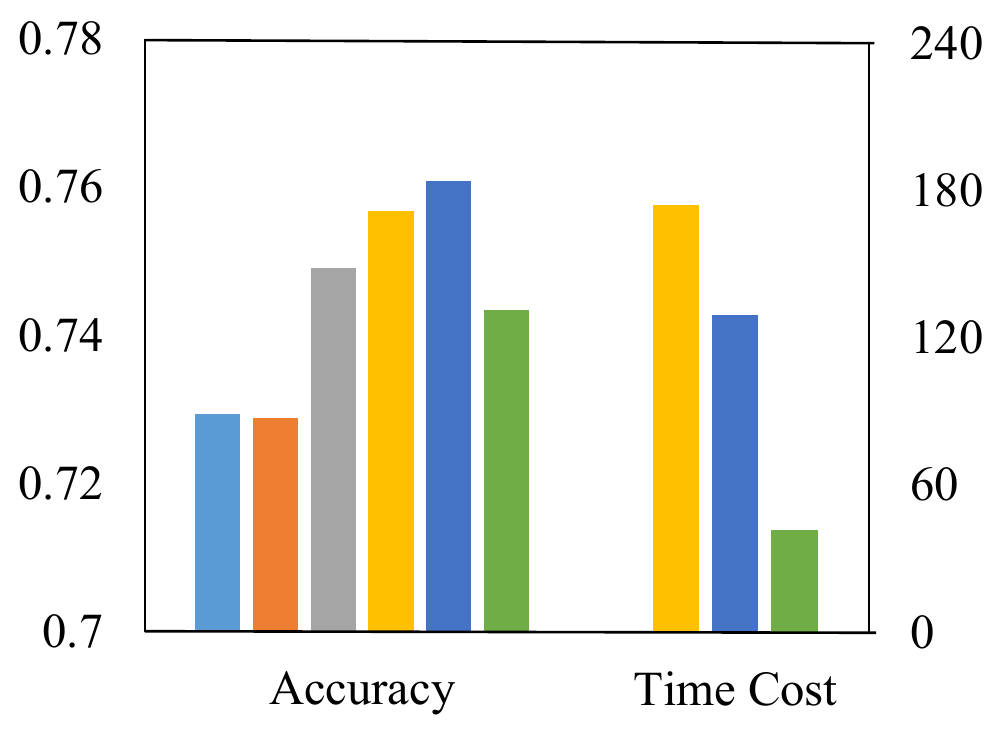}
    \end{minipage}%
    }%
    \subfigure[ML1M]{
    \begin{minipage}[t]{0.236\linewidth}
    \centering
    \includegraphics[width=0.95\linewidth]{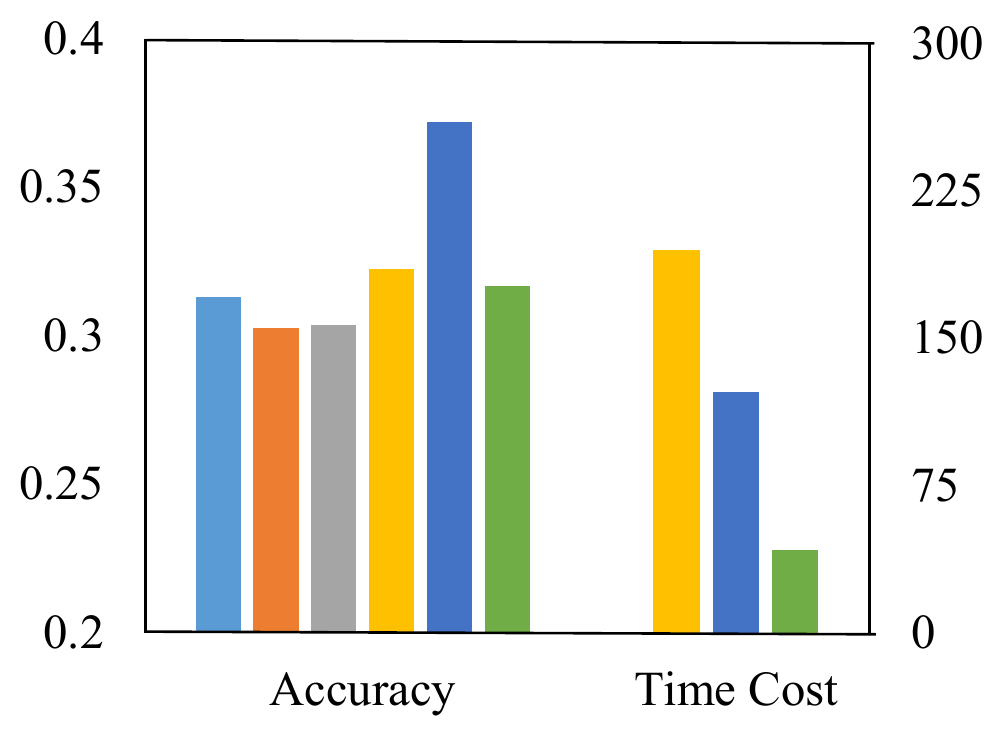}
    \end{minipage}
    }
    \caption{Accuracy and efficiency of comparison with state-of-the-art case-based classification methods. Classification accuracy results and retrieval time costs (seconds/per $10$ retrievals) are reported on the right and left of each figure respectively.}
    \label{fig:perfeffcomp}
\end{figure*}

\subsubsection{Retrieval Efficiency}
\label{sec:retrieval_efficiency}
Since HeCBR and the comparative LSH-based CBR and representation methods perform the same two-step retrieval process, i.e., retrieving candidates from hash tables and reranking the candidates for top-N similar cases, we compare HeCBR with five state-of-the-art case-based methods to further investigate the retrieval efficiency of HeCBR. Note that HCBR proposes an efficient retrieval algorithm, while the other four baselines  adopt the same linear nearest neighbor search (NNS) algorithm. Therefore, we report the time cost of HeCBR, HCBR and ANNCBR in the retrieval efficiency comparison, where ANNCBR is representative for the four NNS-based baselines.

The efficiency comparison is reported in Fig.\ref{fig:perfeffcomp} along with the classification accuracy comparison. From the figures, we observe that HeCBR performs worse than the CBR baselines in terms of accuracy, while it achieves much higher efficiency than the baselines. Specifically, HCBR has the highest accuracy on all datasets, which is reasonable since HCBR leverages the structural information among cases to effectively optimize case similarity measures. HeCBR achieves comparable and even better accuracy compared with the other baselines on most of the datasets, especially on ADT and Dota. Although HeCBR loses accuracy during learning the discrete hash codes, it enhances case representation to compensate for the loss by introducing heterogeneous embedding and capturing feature interaction. In addition, HeCBR greatly improves the retrieval efficiency on all datasets such that it reduces $41\%\sim78\%$ retrieval time costs over ANNCBR and up to $62\%\sim84\%$ retrieval time costs over HCBR as shown on the right of each subfigure in Fig.\ref{fig:perfeffcomp}. As known, the time complexity of the NNS algorithm is proportional to the cardinality of the retrieval set, i.e., the number of cases, HeCBR thus performs increasingly efficiently over ANNCBR with the increase of the number of instances (cases) from ADV to ML1M. Regarding HCBR, it achieves desirable efficiency on MT, CT and ML1M, which is attributed that HCBR structurally organizes all cases and performs large-scale pruning, suitable for large-scale and sparse datasets. From the results, we conclude that HeCBR achieves desirable retrieval accuracy over the CBR-based based, but it greatly reduces retrieval costs. The slight sacrifice of retrieval accuracy for large efficiency improvement is considerably acceptable, especially for the online or real-time scenarios with high demands of retrieval efficiency.

\subsection{Hyperparameter Study (Q3)}
To investigate the parameter sensitivity of HeCBR, we further evaluate the classification accuracy of HeCBR in terms of the weight parameter $\lambda$, the scaling parameter $\alpha$, the view dimension $k_v$, and the embedding dimension $k_w$. All the experiments are conducted under the settings: $\lambda=0.2$, $\alpha=0.6$, $k_v=64$ and $k_w=64$ if not specified.

\subsubsection{Evaluating HeCBR w.r.t. Different $\alpha$ and $\lambda$}
\begin{table}[!t]
  \centering
  \caption{Average accuracy comparison of HeCBR under different values of the weight parameter $\lambda$ and scaling parameter $\alpha$ respectively. The lowest average accuracy values among $\lambda$ and $\alpha$ on each dataset are underlined respectively, and the highest average accuracy value on each dataset is highlighted in bold.}
    \begin{tabular}{l|ccccc|cccc}
    \toprule
    \multirow{2}[4]{*}{Dataset} & \multicolumn{5}{c|}{Weight Parameter ($\lambda$)} & \multicolumn{4}{c}{Scaling Parameter ($\alpha$)} \\
\cmidrule{2-10}          & 0     & 0.2   & 0.4   & 0.6   & 0.8   & 0.2   & 0.4   & 0.6   & 0.8 \\
    \midrule
    ADV   & 0.9661 & 0.9585 & 0.9508 & 0.9261 & \underline{0.9102} & 0.9655 & 0.9661 & \textbf{0.967} & \underline{0.9652} \\
    PT    & \underline{0.433} & 0.4677 & 0.4461 & 0.4544 & 0.4469 & \underline{0.4718} & \textbf{0.5046} & 0.4785 & 0.4866 \\
    ADT   & \textbf{0.8097} & 0.7869 & 0.7808 & \underline{0.7562} & 0.7586 & \underline{0.7972} & 0.8063 & 0.7991 & 0.8033 \\
    Dota  & 0.5378 & 0.503 & 0.5044 & 0.502 & \underline{0.4934} & \textbf{0.5464} & 0.5391 & 0.5429 & \underline{0.5334} \\
    Font  & 0.5031 & 0.5047 & 0.4963 & 0.4834 & \underline{0.4692} & 0.4872 & 0.4934 & \textbf{0.505} & 0.5013 \\
    MT    & \textbf{0.2172} & 0.2048 & 0.1964 & 0.2005 & \underline{0.1889} & 0.1976 & 0.1948 & 0.2041 & \underline{0.1889} \\
    CT    & \underline{0.7383} & 0.7403 & \textbf{0.7436} & 0.7436 & 0.7445 & 0.7433 & \underline{0.7273} & 0.7346 & 0.7308 \\
    ML1M  & \textbf{0.3145} & 0.2903 & \underline{0.2817} & 0.2825 & 0.2855 & 0.2885 & 0.2898 & 0.2916 & \underline{0.2867} \\
    \bottomrule
    \end{tabular}%
  \label{tab:comp_lab_al}%
\end{table}%

As shown in Table \ref{tab:comp_lab_al}, we perform a grid search on $\lambda$ over $\{0, 0.2,0.4,0.6,0.8\}$ and $\alpha$ over $\{0.2,0.4,0.6,0.8\}$ and report the average accuracy under each value of $\lambda$ and $\alpha$ respectively. From the table, we observe that: 1) HeCBR obtains relatively higher performance at $\lambda=0.2$ and $\alpha=0.6$ respectively since HeCBR does not hit the lowest accuracy at $\lambda=0.2$ or $\alpha=0.6$ and achieves much better accuracy compared with other settings. 2) In contrast to $\alpha$, HeCBR has larger but more desirable accuracy fluctuation on $\lambda$, indicating the necessity of a more careful selection of $\lambda$ than that of $\alpha$. This is reasonable because the weight parameter $\lambda$ largely influences the learning objective while $\alpha$ mainly scales the similarities to guarantee higher gradients from the Sigmoid function during backpropagation. 3) In addition, with the increase of $\lambda$, HeCBR may have an increase of accuracy but then has a great decrease when $\lambda$ reaches larger values, e.g., $0.8$, which indicates large $\lambda$ may mislead the learning objective and suggests a relatively small value of $\lambda$. On the contrary, HeCBR achieves better performance at $\alpha=0.4$ or $\alpha=0.6$, suggesting a moderate value of $\alpha$. The results show that large $\alpha$ may not work as a scaling parameter while small $\alpha$ may 
excessively erase the similarity difference.

\subsubsection{Evaluating HeCBR w.r.t. Different View Dimension $k_v$ and Embedding Dimension $k_w$}
Similarly, we perform a grid search on $k_v$ over $\{16,32,64,128,256\}$ and $k_w$ over $\{16,32,64,128,256\}$ and report the average accuracy under each value of $k_v$ and $k_w$ respectively in Table \ref{tab:comp_kw_kv}. From the results, we observe that: 1) HeCBR achieves slightly better performance when $k_v$ and $k_w$ take more moderate values, e.g., $k_v, k_v\in\{32, 64, 128\}$, while smaller or larger embedding dimensions may lead to more inferior performance possibly due to underfitting and overfitting  respectively; 2) overall, HeCBR achieves stable and desirable performance, e.g., without terribly poor accuracy, on all the datasets over the grid search.

\begin{table}[!t]
  \centering
  \caption{Average accuracy comparison of HeCBR under different values of the view dimension $k_v$ and embedding dimension $k_w$ respectively. The lowest average accuracy values between $k_v$ and $k_w$ on each dataset are underlined respectively, and the highest average accuracy value on each dataset is highlighted in bold.}
  \scalebox{0.95}{
    \begin{tabular}{l|ccccc|ccccc}
    \toprule
    \multirow{2}[4]{*}{Dataset} & \multicolumn{5}{c|}{View Dimension}   & \multicolumn{5}{c}{Embedding Dimension} \\
\cmidrule{2-11}          & 16    & 32    & 64    & 128   & 256   & 16    & 32    & 64    & 128   & 256 \\
    \midrule
    ADV   & 0.9685 & 0.9673 & \underline{0.9664} & \textbf{0.9698} & 0.9682 & 0.9658 & 0.967 & 0.9658 & \underline{0.9634} & 0.9652 \\
    PT    & \underline{0.433} & 0.492 & 0.446 & 0.454 & 0.447 & \underline{0.472} & \textbf{0.505} & 0.492 & 0.487 & 0.485 \\
    ADT   & \underline{0.7972} & 0.8049 & 0.8102 & \textbf{0.8114} & 0.8079 & \underline{0.7891} & 0.7937 & 0.7963 & 0.7998 & 0.8003 \\
    Dota  & \underline{0.5308} & 0.5353 & \textbf{0.5458} & 0.5341 & 0.5421 & 0.5375 & 0.5446 & 0.5413 & 0.5396 & \underline{0.5313} \\
    Font  & 0.5023 & 0.5047 & 0.4976 & 0.4821 & \underline{0.4687} & \underline{0.4867} & 0.4953 & \textbf{0.5049} & 0.5041 & 0.4943 \\
    MT    & 0.2037 & 0.1955 & 0.1951 & 0.1971 & \underline{0.1839} & 0.2005 & \textbf{0.2061} & 0.1996 & 0.2006 & \underline{0.1969} \\
    CT    & 0.7346 & 0.7375 & 0.7339 & 0.7326 & \underline{0.7308} & \underline{0.7352} & 0.7386 & \textbf{0.7393} & 0.7364 & 0.7361 \\
    ML1M  & \textbf{0.3253} & 0.2945 & 0.2969 & 0.2936 & \underline{0.2919} & 0.2855 & 0.294 & 0.2938 & \underline{0.2768} & 0.287 \\
    \bottomrule
    \end{tabular}%
    }
  \label{tab:comp_kw_kv}%
\end{table}%

\begin{figure*}[!t]
    \centering
    \subfigure[ADV]{
    \begin{minipage}[t]{0.24\linewidth}
    \centering
    \includegraphics[width=\linewidth]{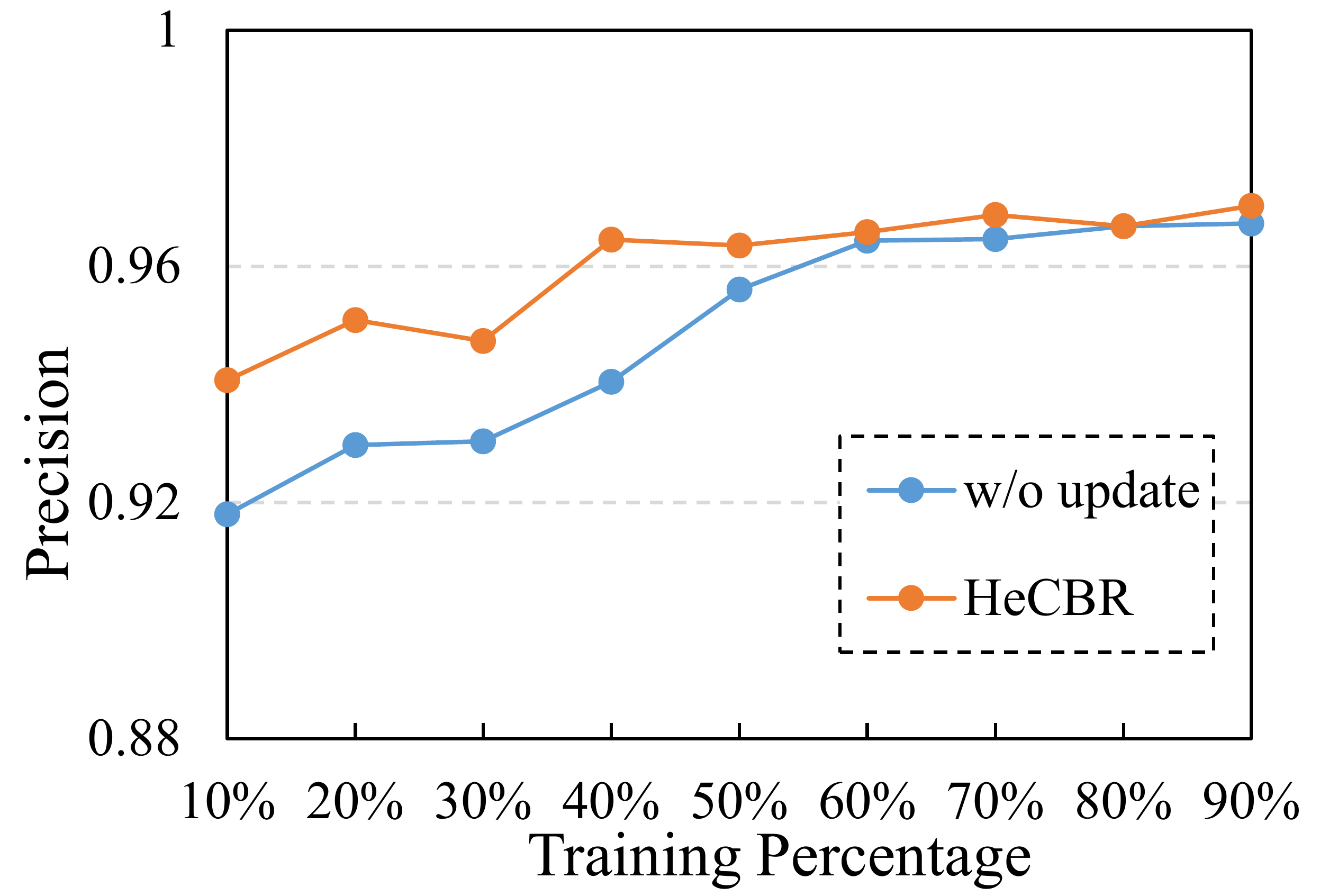}
    \end{minipage}%
    }%
	 \subfigure[PT]{
    \begin{minipage}[t]{0.24\linewidth}
    \centering
    \includegraphics[width=\linewidth]{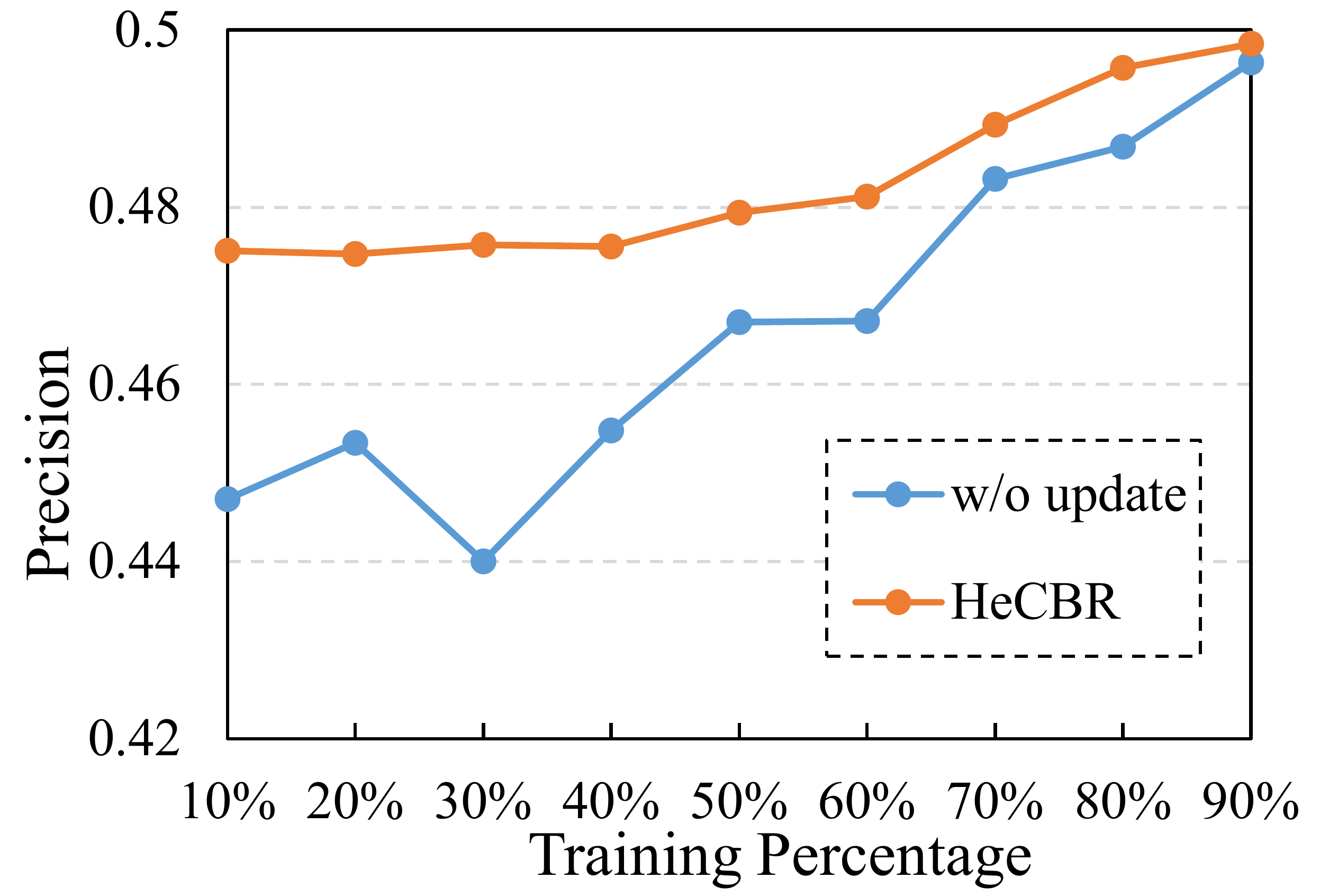}
    \end{minipage}%
    }%
    \subfigure[ADT]{
    \begin{minipage}[t]{0.24\linewidth}
    \centering
    \includegraphics[width=\linewidth]{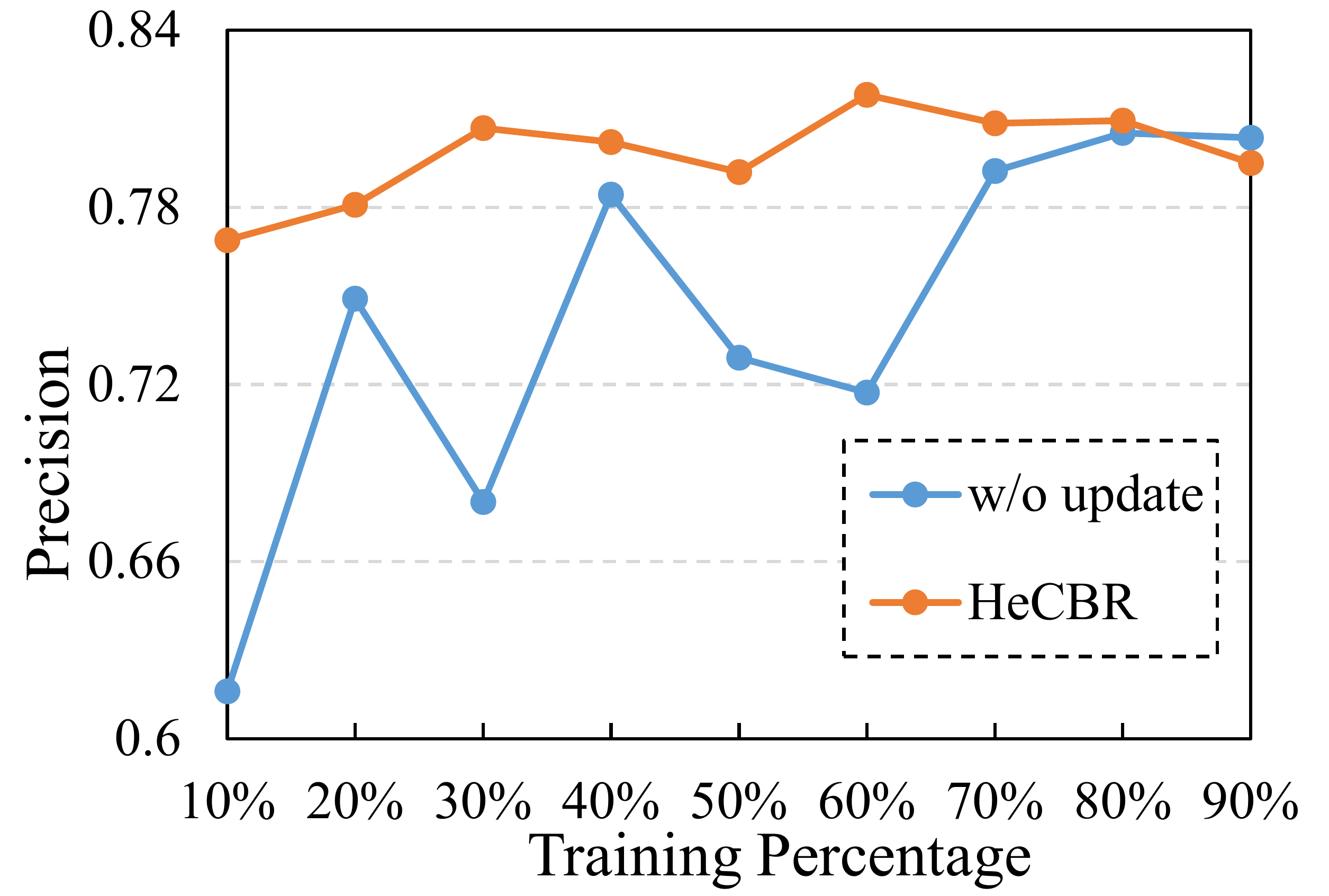}
    \end{minipage}%
    }%
    \subfigure[Dota]{
    \begin{minipage}[t]{0.24\linewidth}
    \centering
    \includegraphics[width=\linewidth]{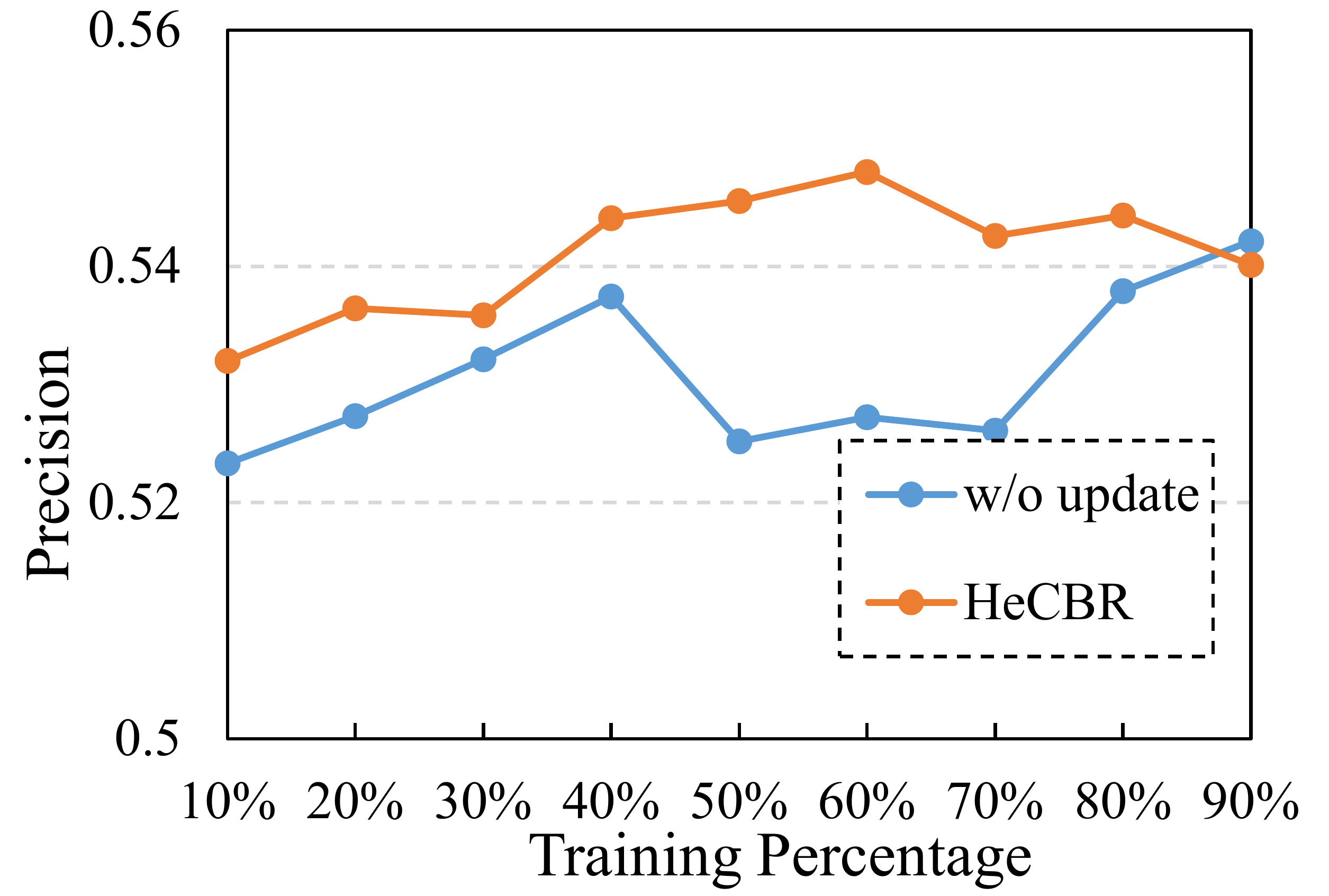}
    \end{minipage}%
    }%
    
	\subfigure[Font]{
    \begin{minipage}[t]{0.24\linewidth}
    \centering
    \includegraphics[width=\linewidth]{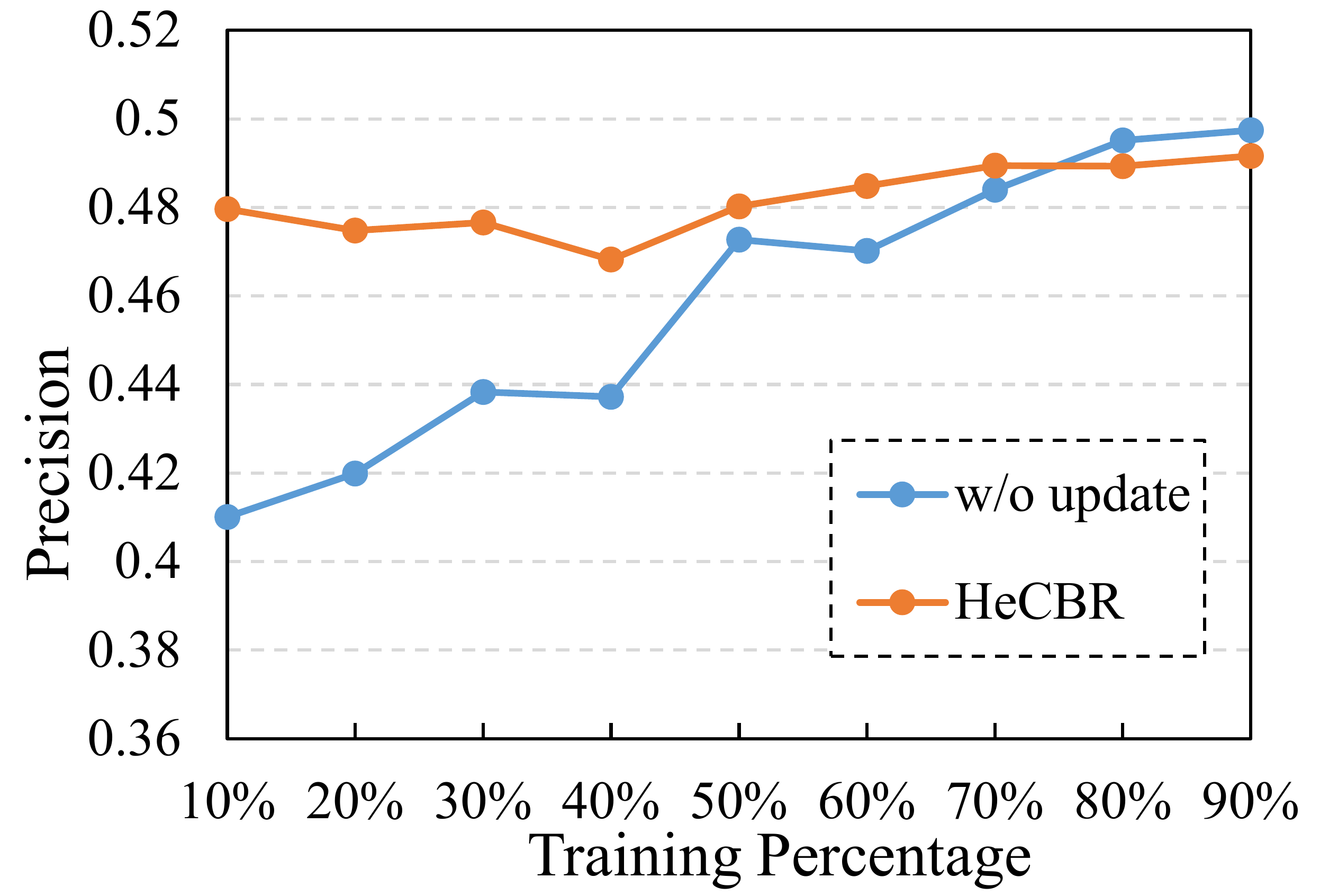}
    \end{minipage}%
    }%
    \subfigure[MT]{
    \begin{minipage}[t]{0.24\linewidth}
    \centering
    \includegraphics[width=\linewidth]{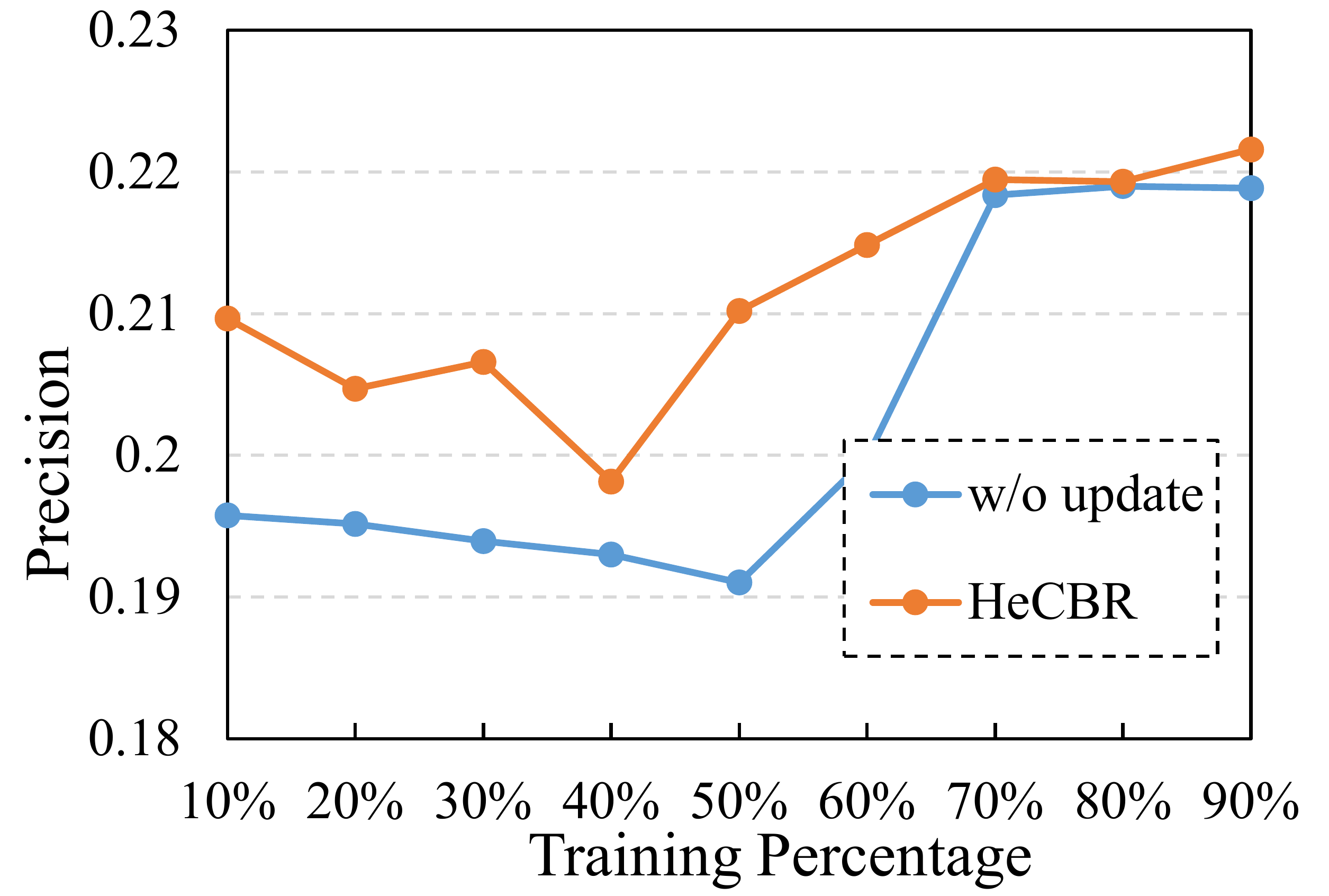}
    \end{minipage}%
    }%
    \subfigure[CT]{
    \begin{minipage}[t]{0.24\linewidth}
    \centering
    \includegraphics[width=\linewidth]{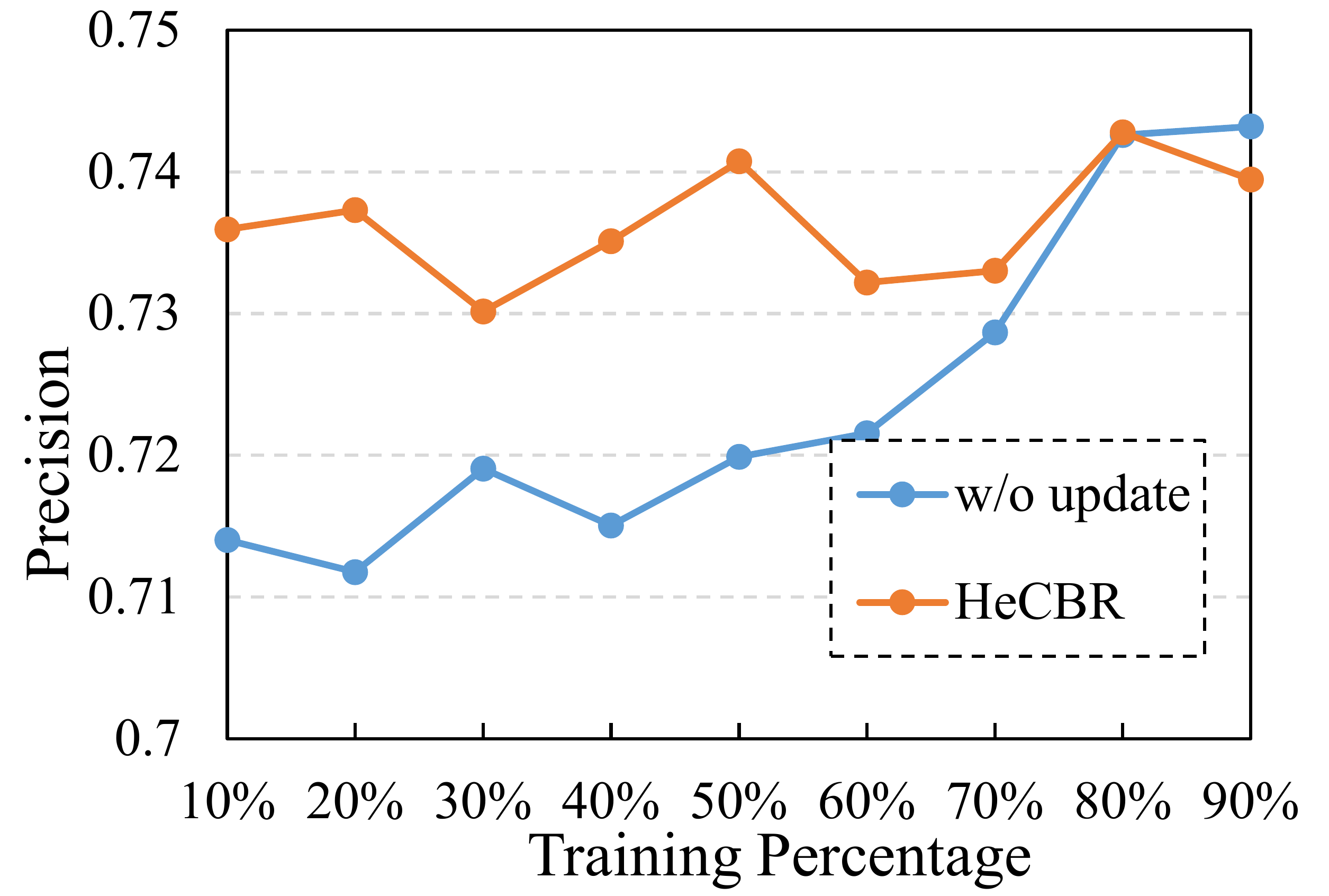}
    \end{minipage}%
    }%
    \subfigure[ML1M]{
    \begin{minipage}[t]{0.24\linewidth}
    \centering
    \includegraphics[width=\linewidth]{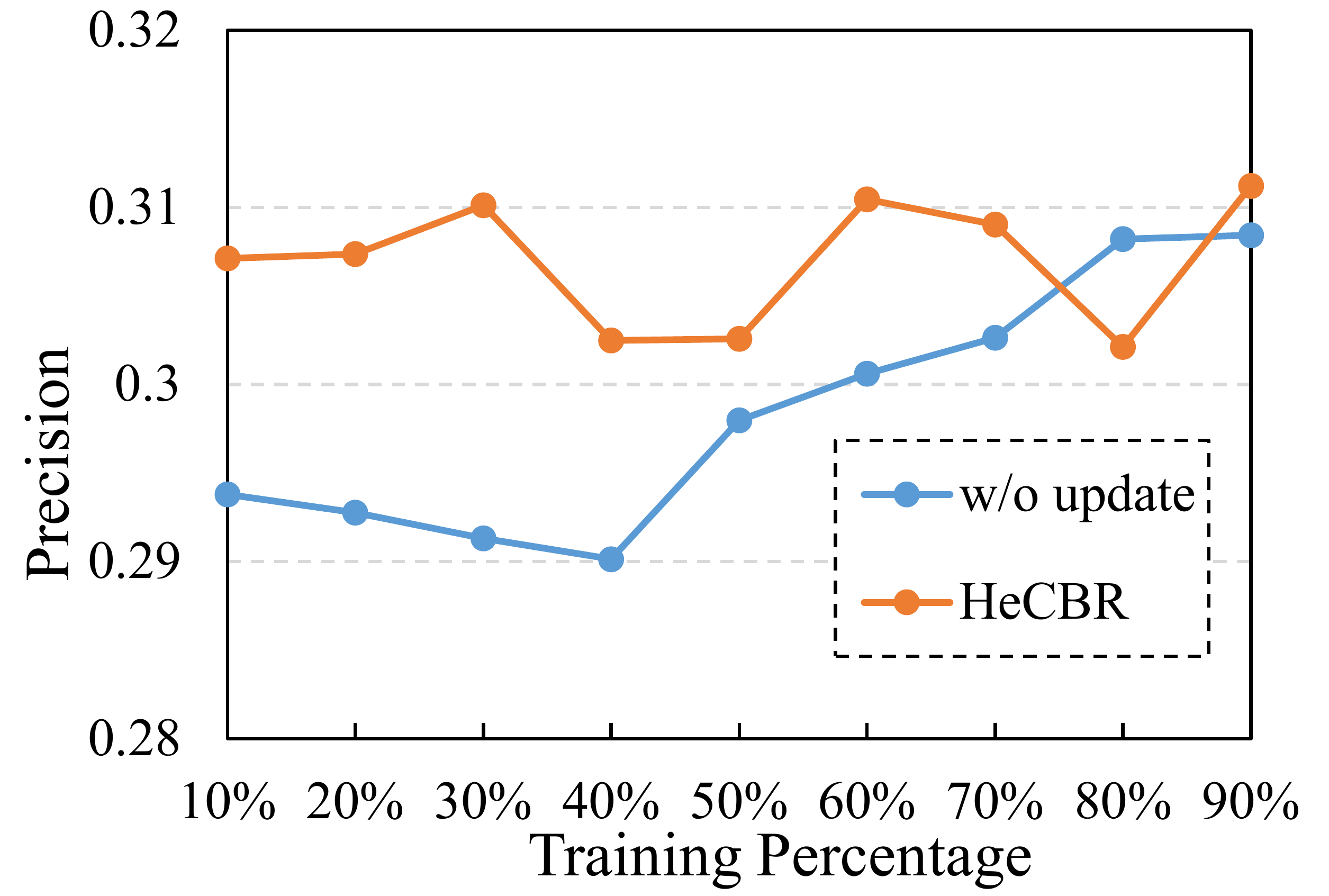}
    \end{minipage}%
    }%
    \caption{Accuracy comparison of HeCBR and its variant \textbf{w/o update} under different training sample rates.}
    \label{fig:update_results}
\end{figure*}

\subsection{Performance Under Adaptive Update (Q4)}
To investigate the stability of HeCBR under different training sample rates, we gradually increase the proportion of training samples on each dataset in Table {\ref{tab:data_stats}} from $10\%$ to $90\%$ in the experiments. We compare HeCBR with its variant (denoted as \textbf{w/o update}) that does not apply the adaptive update strategy to retain newly-solved cases in terms of classification accuracy, and the results are shown in Fig.\ref{fig:update_results}, which indicate three observations:
\begin{itemize}
    \item HeCBR outperforms its variant \textbf{w/o update} nearly under all sampling rates on all datasets, indicating that the update strategy effectively retains beneficial cases to update the case base and hash function for improving future classification.
    \item With the increase of training sample rates, HeCBR shows less fluctuation and increase than the variant \textbf{w/o update}, demonstrating the robustness of HeCBR in relation to different proportions of training samples. The large accuracy increase of \textbf{w/o update} is attributed to the fact that more samples generally benefit the performance even for non-incremental methods.
    \item When the training sample rates reach $80\%~90\%$, the variant \textbf{w/o update} gets comparable or slightly better accuracy than HeCBR. This is because HeCBR may lose its superiority or even introduce noises with the update strategy when the training samples are sufficient. Notably, the proposed adaptive strategy effectively avoids invalid updates and alleviates the inferior performance to a great extent.
\end{itemize}

\section{Conclusion}
In this work, we propose a novel deep hashing network to enhance case-based reasoning. Specifically, the proposed network introduces multiview feature interactions to represent high-dimensional and heterogeneous cases and generates binary hash codes with a quantization regularizer to control the quantization loss. We further propose an adaptive learning loss to strategically update the hash function in the phase of case retraining. Extensive experimental results on public datasets show the superiority of HeCBR over the state-of-the-art hash-based CBR methods in terms of classification and retrieval performance and demonstrate the higher efficiency of HeCBR than the state-of-the-art CBR methods. In the near future, we will consider more complicated scenarios, e.g., case-based planning, detection and decision support, where it may not have sufficient supervision information, we attempt to design an unsupervised/sim-supervised hash learning paradigm for case-based reasoning.

\begin{acks}
This work is supported in part by Australian Research Council Discovery Grant (DP190101079), ARC Future Fellowship Grant (FT190100734), the National Key R\&D Program of China (2019YFB1406300), National Natural Science Foundation of China (No. 61502033), and the Fundamental Research Funds for the Central Universities.
\end{acks}

\bibliographystyle{ACM-Reference-Format}
\bibliography{acmart}

\end{document}